\documentclass[10pt,twoside]{IEEEtran} \onecolumn 

\title{Recursive Sparse Recovery in \\ Large but Structured Noise -- Part 2}
\author{Chenlu Qiu and Namrata Vaswani
\thanks{C. Qiu and N. Vaswani are with the ECE dept at Iowa State University. Email: \{chenlu,namrata\}@iastate.edu. This work was supported by NSF grant CCF-1117125. A shorter version of this work is submitted to ISIT 2013.}
}

\usepackage{graphicx,subfigure}
\usepackage{epsfig,algorithm,amsmath,amssymb,url,algorithmic}

\begin{document}
\maketitle

\setlength{\arraycolsep}{0.03cm}
\newcommand{\xhat}{\hat{x}}
\newcommand{\xpred}{\hat{x}_{t|t-1}}
\newcommand{\xupd}{\hat{x}_{t|t}}
\newcommand{\Ppred}{P_{t|t-1}}
\newcommand{\ty}{\tilde{y}_t}
\newcommand{\tty}{\tilde{y}_{t,\text{res}}}
\newcommand{\tw}{\tilde{w}_t}
\newcommand{\ttw}{\tilde{w}_{t,f}}
\newcommand{\betahat}{\hat{\beta}}

\newcommand{\ypast}{y_{1:t-1}}
\newcommand{\sone}{S_{*}}
\newcommand{\sinf}{{S_{**}}}
\newcommand{\smax}{S_{\max}}
\newcommand{\smin}{S_{\min}}
\newcommand{\samax}{S_{a,\max}}
\newcommand{\Nhat}{{\hat{N}}}

\newcommand{\Dnum}{D_{num}}
\newcommand{\pss}{p^{**,i}}
\newcommand{\fr}{f_{r}^i}

\newcommand{\A}{{\cal A}}
\newcommand{\Z}{{\cal Z}}
\newcommand{\B}{{\cal B}}
\newcommand{\R}{{\cal R}}
\newcommand{\reg}{{\cal G}}
\newcommand{\const}{\mbox{const}}

\newcommand{\trace}{\mbox{tr}}

\newcommand{\hsim}{{\hspace{0.0cm} \sim  \hspace{0.0cm}}}
\newcommand{\he}{{\hspace{0.0cm} =  \hspace{0.0cm}}}

\newcommand{\vect}[2]{\left[\begin{array}{cccccc}
     #1 \\
     #2
   \end{array}
  \right]
  }

\newcommand{\matr}[2]{ \left[\begin{array}{cc}
     #1 \\
     #2
   \end{array}
  \right]
  }
\newcommand{\vc}[2]{\left[\begin{array}{c}
     #1 \\
     #2
   \end{array}
  \right]
  }

\newcommand{\gdot}{\dot{g}}
\newcommand{\Cdot}{\dot{C}}
\newcommand{\re}{\mathbb{R}}
\newcommand{\n}{{\cal N}}  
\newcommand{\N}{{\overrightarrow{\bf N}}}  
\newcommand{\chat}{\tilde{C}_t}
\newcommand{\chati}{\chat^i}

\newcommand{\cmin}{C^*_{min}}
\newcommand{\twi}{\tilde{w}_t^{(i)}}
\newcommand{\twj}{\tilde{w}_t^{(j)}}
\newcommand{\wi}{{w}_t^{(i)}}
\newcommand{\twio}{\tilde{w}_{t-1}^{(i)}}

\newcommand{\tWi}{\tilde{W}_n^{(m)}}
\newcommand{\tWj}{\tilde{W}_n^{(k)}}
\newcommand{\Wi}{{W}_n^{(m)}}
\newcommand{\tWio}{\tilde{W}_{n-1}^{(m)}}

\newcommand{\ds}{\displaystyle}

\newcommand{\SAR}{S$\!$A$\!$R }
\newcommand{\MAR}{MAR}
\newcommand{\MMRF}{MMRF}
\newcommand{\AR}{A$\!$R }
\newcommand{\GMRF}{G$\!$M$\!$R$\!$F }
\newcommand{\DTM}{D$\!$T$\!$M }
\newcommand{\MSE}{M$\!$S$\!$E }
\newcommand{\RCS}{R$\!$C$\!$S }
\newcommand{\uomega}{\underline{\omega}}
\newcommand{\lu}{\mu}
\newcommand{\g}{g}
\newcommand{\s}{{\bf s}}
\newcommand{\bft}{{\bf t}}
\newcommand{\refmap}{{\cal R}}
\newcommand{\totrefl}{{\cal E}}
\newcommand{\beq}{\begin{equation}}
\newcommand{\eeq}{\end{equation}}
\newcommand{\bdm}{\begin{displaymath}}
\newcommand{\edm}{\end{displaymath}}
\newcommand{\hatz}{\hat{z}}
\newcommand{\hatu}{\hat{u}}
\newcommand{\tilz}{\tilde{z}}
\newcommand{\tilu}{\tilde{u}}
\newcommand{\hhatz}{\hat{\hat{z}}}
\newcommand{\hhatu}{\hat{\hat{u}}}
\newcommand{\tilc}{\tilde{C}}
\newcommand{\hatc}{\hat{C}}
\newcommand{\tim}{n}

\newcommand{\ssp}{\renewcommand{\baselinestretch}{1.0}}
\newcommand{\defd}{\mbox{$\stackrel{\mbox{$\triangle$}}{=}$}}
\newcommand{\goes}{\rightarrow}
\newcommand{\tends}{\rightarrow}
\newcommand{\se}{&=&}
\newcommand{\sdefn}{& :=  &}
\newcommand{\sle}{& \le &}
\newcommand{\sge}{& \ge &}
\newcommand{\plusminus}{\stackrel{+}{-}}
\newcommand{\Ey}{E_{Y_{1:t}}}
\newcommand{\ey}{E_{Y_{1:t}}}

\newcommand{\equivto}{\mbox{~~~which is equivalent to~~~}}
\newcommand{\nonzero}{i:\pi^n(x^{(i)})>0}
\newcommand{\nonzeroc}{i:c(x^{(i)})>0}

\newcommand{\supn}{\sup_{\phi:||\phi||_\infty \le 1}}

\newtheorem{theorem}{Theorem}[section]
\newtheorem{lem}[theorem]{Lemma}
\newtheorem{sigmodel}[theorem]{Signal Model}
\newtheorem{corollary}[theorem]{Corollary}
\newtheorem{definition}[theorem]{Definition}
\newtheorem{remark}[theorem]{Remark}
\newtheorem{example}[theorem]{Example}
\newtheorem{ass}[theorem]{Assumption}
\newtheorem{proposition}[theorem]{Proposition}
\newtheorem{fact}[theorem]{Fact}
\newtheorem{heuristic}[theorem]{Heuristic}

\newcommand{\eps}{\epsilon}
\newcommand{\bd}{\begin{definition}}
\newcommand{\ed}{\end{definition}}
\newcommand{\udq}{\underline{D_Q}}
\newcommand{\td}{\tilde{D}}
\newcommand{\epsinv}{\epsilon_{inv}}
\newcommand{\al}{\mathcal{A}}

\newcommand{\bfx} {\bf X}
\newcommand{\bfy} {\bf Y}
\newcommand{\bfz} {\bf Z}
\newcommand{\ddas}{\mbox{${d_1}^2({\bf X})$}}
\newcommand{\ddbs}{\mbox{${d_2}^2({\bfx})$}}
\newcommand{\dda}{\mbox{$d_1(\bfx)$}}
\newcommand{\ddb}{\mbox{$d_2(\bfx)$}}
\newcommand{\xinc}{{\bfx} \in \mbox{$C_1$}}
\newcommand{\eqa}{\stackrel{(a)}{=}}
\newcommand{\eqb}{\stackrel{(b)}{=}}
\newcommand{\eqe}{\stackrel{(e)}{=}}
\newcommand{\leqc}{\stackrel{(c)}{\le}}
\newcommand{\leqd}{\stackrel{(d)}{\le}}

\newcommand{\leqa}{\stackrel{(a)}{\le}}
\newcommand{\leqb}{\stackrel{(b)}{\le}}
\newcommand{\leqe}{\stackrel{(e)}{\le}}
\newcommand{\leqf}{\stackrel{(f)}{\le}}
\newcommand{\leqg}{\stackrel{(g)}{\le}}
\newcommand{\leqh}{\stackrel{(h)}{\le}}
\newcommand{\leqi}{\stackrel{(i)}{\le}}
\newcommand{\leqj}{\stackrel{(j)}{\le}}

\newcommand{\w}{{W^{LDA}}}
\newcommand{\halpha}{\hat{\alpha}}
\newcommand{\hsigma}{\hat{\sigma}}
\newcommand{\slmax}{\sqrt{\lambda_{max}}}
\newcommand{\slmin}{\sqrt{\lambda_{min}}}
\newcommand{\lmax}{\lambda_{max}}
\newcommand{\lmin}{\lambda_{min}}

\newcommand{\da} {\frac{\alpha}{\sigma}}
\newcommand{\chka} {\frac{\check{\alpha}}{\check{\sigma}}}
\newcommand{\sumo}{\sum _{\underline{\omega} \in \Omega}}
\newcommand{\distance}{d\{(\hatz _x, \hatz _y),(\tilz _x, \tilz _y)\}}
\newcommand{\col}{{\rm col}}
\newcommand{\rcs}{\sigma_0}
\newcommand{\CalR}{{\cal R}}
\newcommand{\df}{{\delta p}}
\newcommand{\dq}{{\delta q}}
\newcommand{\dZ}{{\delta Z}}
\newcommand{\pprime}{{\prime\prime}}

\newcommand{\vn}{N}

\newcommand{\bv}{\begin{vugraph}}
\newcommand{\ev}{\end{vugraph}}
\newcommand{\bi}{\begin{itemize}}
\newcommand{\ei}{\end{itemize}}
\newcommand{\ben}{\begin{enumerate}}
\newcommand{\een}{\end{enumerate}}
\newcommand{\be}{\protect\[}
\newcommand{\ee}{\protect\]}
\newcommand{\bean}{\begin{eqnarray*} }
\newcommand{\eean}{\end{eqnarray*} }
\newcommand{\bea}{\begin{eqnarray} }
\newcommand{\eea}{\end{eqnarray} }
\newcommand{\nn}{\nonumber}
\newcommand{\ba}{\begin{array} }
\newcommand{\ea}{\end{array} }
\newcommand{\ep}{\mbox{\boldmath $\epsilon$}}
\newcommand{\epp}{\mbox{\boldmath $\epsilon '$}}
\newcommand{\Lep}{\mbox{\LARGE $\epsilon_2$}}
\newcommand{\und}{\underline}
\newcommand{\pdif}[2]{\frac{\partial #1}{\partial #2}}
\newcommand{\odif}[2]{\frac{d #1}{d #2}}
\newcommand{\dt}[1]{\pdif{#1}{t}}
\newcommand{\urho}{\underline{\rho}}

\newcommand{\spc}{{\cal S}}
\newcommand{\tspc}{{\cal TS}}

\newcommand{\uv}{\underline{v}}
\newcommand{\us}{\underline{s}}
\newcommand{\uc}{\underline{c}}
\newcommand{\utheta}{\underline{\theta}^*}
\newcommand{\ualpha}{\underline{\alpha^*}}

\newcommand{\uxy}{\underline{x}^*}
\newcommand{\uxyj}{[x^{*}_j,y^{*}_j]}
\newcommand{\arcl}[1]{arclen(#1)}
\newcommand{\one}{{\mathbf{1}}}

\newcommand{\uxyjt}{\uxy_{j,t}}
\newcommand{\E}{\mathbf{E}}

\newcommand{\rhomat}{\left[\begin{array}{c}
                        \rho_3 \ \rho_4 \\
                        \rho_5 \ \rho_6
                        \end{array}
                   \right]}
\newcommand{\deltat}{\tau} 
\newcommand{\deltatt}{\Delta t_1}
\newcommand{\ceil}[1]{\ulcorner #1 \urcorner}

\newcommand{\xxi}{x^{(i)}}
\newcommand{\txi}{\tilde{x}^{(i)}}
\newcommand{\txj}{\tilde{x}^{(j)}}

\newcommand{\mi}[1]{{#1}^{(m,i)}}

\newcommand{\cs}{\text{cs}}
\newcommand{\bigO}{\mathcal{O}}
\newcommand{\new}{\text{new}}
\newcommand{\newset}{\text{new-set}}
\newcommand{\old}{\text{old}}
\newcommand{\rank}{\text{rank}}
\newcommand{\Shat}{\hat{S}}
\newcommand{\Lhat}{\hat{L}}
\newcommand{\Phat}{\hat{P}}
\newcommand{\Span}{\text{span}}
\newcommand{\del}{\text{del}}
\newcommand{\diag}{\text{diag}}
\newcommand{\add}{\text{add}}

\newcommand{\train}{\text{train}}
\newcommand{\thresh}{\hat{\lambda}^-}
\newcommand{\betathresh}{\eta} 
\newcommand{\That}{\hat{T}}
\newcommand{\zetasp}{{\zeta_*^+}}
\newcommand{\SE}{\text{SE}}

\newcommand{\egam}{\Gamma^e} 
\newcommand{\tegam}{\tilde{\Gamma}^e} 
\newcommand{\calb}{{\cal B}}
\newcommand{\calc}{{\cal C}}
\newcommand{\calf}{{\cal F}}
\newcommand{\ecalb}{{\calb}^e}
\newcommand{\ecalc}{{\calc}^e}


\begin{abstract}
We study the problem of recursively recovering a time sequence of sparse vectors, $S_t$, from measurements $M_t: = S_t + L_t$ that are corrupted by structured noise $L_t$ which is dense and can have large magnitude. The structure that we require is that $L_t$ should lie in a low dimensional subspace that is either fixed or changes ``slowly enough"; and the eigenvalues of its covariance matrix are ``clustered". We do not assume any model on the sequence of sparse vectors. Their support sets and their nonzero element values may be either independent or correlated over time (usually in many applications they are correlated). The only thing required is that there be {\em some} support change every so often.
We introduce a novel solution approach called Recursive Projected Compressive Sensing with cluster-PCA (ReProCS-cPCA) that addresses some of the limitations of earlier work. Under mild assumptions, we show that, with high probability, ReProCS-cPCA can exactly recover the support set of $S_t$ at all times; and the reconstruction errors of both $S_t$ and $L_t$ are upper bounded by a time-invariant and small value.
\end{abstract} 


\section{Introduction} \label{intro}
In this work, we study the problem of recursively recovering a time sequence of sparse vectors, $S_t$, from measurements $M_t: = S_t + L_t$ that are corrupted by structured noise $L_t$ which is dense and can have large magnitude. The structure that we require is that $L_t$ should lie in a low dimensional subspace that is either fixed or changes ``slowly enough" as discussed in Sec \ref{slowss}; and the eigenvalues of its covariance matrix are ``clustered" as explained in Sec \ref{eigencluster}.
As a by-product, at certain times, the basis vectors for the subspace in which the most recent several $L_t$'s lies is also recovered. Thus, at these times, we also solve the recursive robust principal components' analysis (PCA) problem. For the recursive robust PCA problem, $L_t$ is the signal of interest while $S_t$ can be interpreted as the outlier (large but sparse noise).




A key application where the above problem occurs is in video analysis where the goal is to separate a slowly changing background from moving foreground objects \cite{Torre03aframework,rpca}. If one stacks each image frame as a column vector, the background is well modeled as lying in a low dimensional subspace that may gradually change over time, while the moving foreground objects constitute the sparse vectors \cite{error_correction_PCP_l1,rpca} which change in a correlated fashion over time. Another key application is online detection of brain activation patterns from functional MRI (fMRI) sequences. In this case, the ``active" region of the brain is the correlated sparse vector.%

\subsection{Related Work}

Many of the older works on sparse recovery with structured noise study the case of sparse recovery from large but sparse noise (outliers), e.g., \cite{error_correction_PCP_l1,Laska_exactsignal,trac_tran}. However, here we are interested in sparse recovery in large but low dimensional noise. On the other hand, most older works on robust PCA cannot recover the outlier ($S_t$) when its nonzero entries have magnitude much smaller than that of the low dimensional part ($L_t$) \cite{ipca_weightedand,Torre03aframework,Li03anintegrated}. The main goal of this work is to study sparse recovery and hence we do not discuss these older works here. Some recent works on robust PCA such as \cite{mccoy_tropp11,outlier_pursuit} assume that an entire measurement vector $M_t$ is either an inlier ($S_t$ is a zero vector) or an outlier (all entries of $S_t$ can be nonzero), and a certain number of $M_t$'s are inliers. These works also cannot be used when all $S_t$'s are nonzero but sparse.


In a series of recent works \cite{rpca,rpca2}, a new and elegant solution, which is referred to as Principal Components' Pursuit (PCP) in \cite{rpca}, has been proposed. It redefines batch robust PCA as a problem of separating a low rank matrix, ${\cal L}_t := [L_1,\dots,L_t]$, from a sparse matrix, ${\cal S}_t := [S_1,\dots,S_t]$, using the measurement matrix, ${\cal M}_t := [M_1,\dots,M_t] = {\cal L}_t+ {\cal S}_t$. Thus these works can be interpreted as batch solutions to sparse recovery in large but low dimensional noise. Other recent works that also study batch algorithms for recovering a sparse ${\cal S}_t$ and a low-rank ${\cal L}_t$ from ${\cal M}_t := {\cal L}_t+ {\cal S}_t$ or from undersampled measurements include \cite{rpca_tropp, linear_inverse_prob, rpca_hu,SpaRCS,rpca_vayatis,rpca_zhang, rpca_Giannakis,compressivePCP,rpca_reduced,noisy_undersampled_yuan}.

It was shown in \cite{rpca} that, with high probability (w.h.p.), one can recover ${\cal L}_t$ and ${\cal S}_t$ exactly by solving 
\beq
\underset{{\cal L},{\cal S}} {\min}\|{\cal L}\|_* + \lambda\|{\cal S}\|_{1,\text{vec}} \ \text{subject to}  \ \ {\cal L} + {\cal S} = {\cal M}_t
\label{pcp_prob}
\eeq
provided that (a) ${\cal L}_t$ is dense (its left and right singular vectors satisfy certain conditions);
(b) any element of the matrix ${\cal S}_t$ is nonzero w.p. $\varrho$, and zero w.p. $1-\varrho$, independent of all others (in particular, this means that the support sets of the different $S_t$'s are independent over time);
%
and (c) the rank of ${\cal L}_t$ and the support size of ${\cal S}_t$ are small enough. Here $\|B\|_*$ is the nuclear norm of $B$ (sum of singular values of $B$) while $\|B\|_{1,\text{vec}}$ is the $\ell_1$ norm of $B$ seen as a long vector.
In most applications, it is fair to assume that the low dimensional part, $L_t$ (background in case of video) is dense. However, the assumption that the support of the sparse part (foreground in case of video) is independent over time is often not valid. Foreground objects typically move in a correlated fashion, and may even not move for a few frames. This results in ${\cal S}_t$ being sparse and low rank.

The question then is, what can we do if ${\cal L}_t$ is low rank and dense, but ${\cal S}_t$ is sparse and may also be low rank? In this case, without any extra information, in general, it is not possible to separate ${\cal S}_t$ and ${\cal L}_t$. In \cite{rrpcp_perf}, we introduced the Recursive Projected Compressive Sensing (ReProCS) algorithm that provided one possible solution to this problem by using the extra piece of information that an initial short sequence of $L_t$'s, or $L_t$'s in small noise, is available (which can be used to get an accurate estimate of the subspace in which the initial $L_t$'s lie) and assuming slow subspace change (as explained in Sec. \ref{slowss}). The key idea of ReProCS is as follows. At time $t$, assume that a $n \times r$ matrix with orthonormal columns, $\Phat_{(t-1)}$, is available with $\Span(\Phat_{(t-1)}) \approx \Span({\cal L}_{t-1})$. We project $M_t$ perpendicular to $\Span(\Phat_{(t-1)})$. 
Because of slow subspace change, this cancels out most of the contribution of $L_t$. Recovering $S_t$ from the projected measurements then becomes a classical sparse recovery / compressive sensing (CS) problem in small noise \cite{candes_rip}. Under a denseness assumption on $\Span({\cal L}_{t-1})$, one can show that $S_t$ can be accurately recovered via $\ell_1$ minimization. Thus, $L_t = M_t-S_t$ can also be recovered accurately. We use the estimates of $L_t$ in a projection-PCA based subspace estimation algorithm to update $\Phat_{(t)}$.

ReProCS is designed under the assumption that the subspace in which the most recent several $L_t$'s lie can only grow over time. It assumes a model in which at every subspace change time, $t_j$, some new directions get added to this subspace. After every subspace change, it uses projection-PCA to estimate the newly added subspace. As a result the rank of $\Phat_{(t)}$ keeps increasing with every subspace change. Therefore, the number of effective measurements available for the CS step, $(n-\rank(\Phat_{(t-1)}))$, keeps reducing. To keep this number large enough at all times, ReProCS needs to assume a bound on the total number of subspace changes, $J$.

\subsection{Our Contributions and More Related Work}

In practice, usually, the dimension of the subspace in which the most recent several $L_t$'s lie typically remains roughly constant. A simple way to model this is to assume that at every change time, $t_j$, some new directions can get added and some existing directions can get deleted from this subspace and to assume an upper bound on the difference between the total number of added and deleted directions (the earlier model in \cite{rrpcp_perf} is a special case of this). ReProCS still applies for this more general model as discussed in the extensions section of \cite{rrpcp_perf}. However, because it never deletes directions, the rank of $\Phat_{(t)}$ still keeps increasing with every subspace change time and so it still requires a bound on $J$.

In this work, we address the above limitation by introducing a novel approach called {\em cluster-PCA} that re-estimates the current subspace after the newly added directions have been accurately estimated. This re-estimation step ensures that the deleted directions have been ``removed" from the new $\Phat_{(t)}$. We refer to the resulting algorithm as {\em ReProCS-cPCA}.
The design and analysis of cluster-PCA and ReProCS-cPCA is the focus of the current paper.  We will see that ReProCS-cPCA does not need a bound on $J$ as long as the delay between subspace change times increases in proportion to $\log J$. An extra assumption that is needed though is that the eigenvalues of the covariance matrix of $L_t$ are sufficiently clustered at certain times as explained in Sec \ref{eigencluster}. As discussed in Sec \ref{discuss}, this is a practically valid assumption.

Under the clustering assumption and some other mild assumptions, we show that, w.h.p, at all times, ReProCS-cPCA can exactly recover the support of $S_t$, and the reconstruction errors of both $S_t$ and $L_t$ are upper bounded by a time invariant and small value. Moreover, we show that the subspace recovery error decays roughly exponentially with every projection-PCA step.
The proof techniques developed in this work are very different from those used to obtain performance guarantees in recent batch robust PCA works such as \cite{rpca,rpca2, novel_m_estimator,mccoy_tropp11, outlier_pursuit,rpca_tropp, linear_inverse_prob, rpca_reduced, rpca_Giannakis,rpca_zhang,compressivePCP,noisy_undersampled_yuan}. As explained earlier, \cite{mccoy_tropp11,outlier_pursuit} also study a different problem. Our proof utilizes sparse recovery results \cite{candes_rip}; results from matrix perturbation theory (sin $\theta$ theorem \cite{davis_kahan} and Weyl's theorem \cite{hornjohnson}) and the matrix Hoeffding inequality \cite{tail_bound}.


Our result for ReProCS-cPCA (and also that for ReProCS from \cite{rrpcp_perf}) does not assume any model on the sparse vectors', $S_t$'s. In particular, it allows the support sets of the $S_t$'s to be either independent, e.g. generated via the model of \cite{rpca} (resulting in ${\cal S}_t$ being full rank w.h.p.), or correlated over time (can result in ${\cal S}_t$ being low rank). As explained in Sec \ref{discuss}, the only thing that is required is that there be {\em some} support changes every so often. We should point out that some of the other works that study the batch problem, e.g. \cite{rpca,rpca_zhang}, also allow ${\cal S}_t$ to be low rank.

A key difference of our work compared with most existing work analyzing finite sample PCA, e.g. \cite{nadler}, and references therein, is that in these works, the noise/error in the observed data is independent of the true (noise-free) data. However, in our case, because of how $\Lhat_t$ is computed, the error $e_t = L_t - \Lhat_t$ is correlated with $L_t$. As a result the tools developed in these earlier works cannot be used for our problem. This is the main reason we need to develop and analyze projection-PCA based approaches for both subspace addition and deletion.

In earlier conference papers \cite{rrpcp_allerton, rrpcp_allerton11}, we first introduced the ReProCS idea. However, these used an algorithm motivated by recursive PCA \cite{sequentialSVD} for updating the subspace estimates on-the-fly. As explained in Sec \ref{rep_del_perf} and also in \cite[Appendix F]{rrpcp_perf}, it is not clear how to obtain performance guarantees for recursive PCA (which is a fast algorithm for PCA) for our problem. Another online algorithm that addresses a problem similar to ours is given in \cite{grass_undersampled}. This also does not obtain guarantees.


The ReProCS-cPCA approach is related to that of \cite{decodinglp,rpca_regression, rpca_regression_sparse} in that all of these first try to nullify the low dimensional signal by projecting the measurement vector into a subspace perpendicular to that of the low dimensional signal, and then solve for the sparse ``error" vector. However, the big difference is that in all of these works the basis for the subspace of the low dimensional signal is {\em perfectly known.} We study {\em the case where the subspace is not known and can change over time}.


\subsection{Paper Organization}
We give the notation next followed by a review of results from existing work that we will need. The problem definition and the three key assumptions that are needed are explained in Sec \ref{probdef}. We develop the ReProCS-cPCA algorithm in Sec \ref{rep_del_perf}. We give its performance guarantees (Theorem \ref{thm2}) in Sec \ref{perf_g}. Here we also provide a discussion of the result and the assumptions it makes. We define the quantities needed for the proof and give the proof outline in Sec \ref{detailed}. The proof of Theorem \ref{thm2} is given in Sec \ref{thmproof}. The key lemmas needed for it are given and proved in Sec \ref{keylems}. In Sec \ref{sims}, we show numerical experiments demonstrating Theorem \ref{thm2}, as well as comparisons with ReProCS and PCP. Conclusions are given in Sec \ref{conc}.

\subsection{Notation}


For a set $T \subseteq \{1,2,\dots n\}$, we use $|T|$ to denote its cardinality, i.e., the number of elements in $T$. We use $T^c$ to denote its complement w.r.t. $\{1,2,\dots n\}$, i.e. $T^c:= \{i \in \{1,2,\dots n\}: i \notin T \}$. 
The notations $T_1 \subseteq T_2$ and $T_2 \supseteq T_1$ both mean that $T_1$ is a subset of $T_2$.


We use the notation $[t_1, t_2]$ to denote an interval which contains $t_1$ and $t_2$, as well as all integers between them, i.e. $[t_1, t_2]:=\{t_1, t_1+1, \cdots, t_2\}$. The notation $[L_t; t\in [t_1,t_2]]$ is used to denote the matrix $[L_{t_1}, L_{t_1+1}, \cdots, L_{t_2}]$.

For a vector $v$, $v_i$ denotes the $i$th entry of $v$ and $v_T$ denotes a vector consisting of the entries of $v$ indexed by $T$. We use $\|v\|_p$ to denote the $\ell_p$ norm of $v$. The support of $v$, $\text{supp}(v)$, is the set of indices at which $v$ is nonzero, $\text{supp}(v) := \{i : v_i\neq 0\}$. We say that $v$ is s-sparse if $|\text{supp}(v)| \leq  s$.

For a tall matrix $P$, $\Span(P)$ denotes the subspace spanned by the column vectors of $P$.

For a matrix $B$, $B'$ denotes its transpose, and $B^{\dag}$ denotes its pseudo-inverse. For a matrix with linearly independent columns, $B^{\dag} = (B'B)^{-1}B'$.
We use $\|B\|_2:= \max_{x \neq 0} \|Bx\|_2/\|x\|_2$ to denote the induced 2-norm of the matrix. Also, $\|B\|_*$ is the nuclear norm and $\|B\|_{\max}$ denotes the maximum over the absolute values of all its entries.
We let $\sigma_i(B)$ denote the $i$th largest singular value of $B$. For a Hermitian matrix, $B$, we use the notation $B \overset{EVD}{=} U \Lambda U'$ to denote the eigenvalue decomposition (EVD) of $B$. Here $U$ is an orthonormal matrix and $\Lambda$ is a diagonal matrix with entries arranged in non-increasing order. Also, we use $\lambda_i(B)$ to denote the $i$th largest eigenvalue of a Hermitian matrix $B$ and we use $\lambda_{\max}(B)$ and $\lambda_{\min}(B)$ denote its maximum and minimum eigenvalues. If $B$ is Hermitian positive semi-definite (p.s.d.), then $\lambda_i(B) = \sigma_i(B)$. For Hermitian matrices $B_1$ and $B_2$, the notation $B_1 \preceq B_2$ means that $B_2-B_1$ is p.s.d. Similarly, $B_1 \succeq B_2$ means that $B_1-B_2$ is p.s.d.


For a Hermitian matrix $B$, we have $\|B\|_2 = \sqrt{\max( \lambda_{\max}^2(B), \lambda_{\min}^2(B) )}$. Thus, for a $b\geq 0$, $\|B\|_2 \le b$ implies that $-b \le \lambda_{\min}(B) \le \lambda_{\max}(B) \le b$. If $B$ is a Hermitian p.s.d. matrix, then $\|B\|_2 = \lambda_{\max}(B)$.

The notation $[.]$ denotes an empty matrix. We use $I$ to denote an identity matrix. For an $m \times n$ matrix $B$ and an index set $T \subseteq \{1,2,\dots n\}$, $B_T$ is the sub-matrix of $B$ containing columns with indices in the set $T$. Notice that $B_T= B I_T$. We use $B \setminus B_{T}$ to denote $B_{T^c}$ (here $T^c:= \{i \in \{1,2,\cdots, n\}: i \notin T \}$).
Given another matrix $B_2$ of size $m \times n_2$, $[B \ B_2]$ constructs a new matrix by concatenating matrices $B$ and $B_2$ in horizontal direction. Thus, $[(B\setminus B_{T}) \  B_2] = [B_{T^c} \  B_2]$. For any matrix $B$ and sets $T_1, T_2$, $(B)_{T_1,T_2}$ denotes the sub-matrix containing the rows with indices in $T_1$ and columns with indices in $T_2$.

\begin{definition}
We refer to a tall matrix $P$ as a {\em basis matrix} if it satisfies $P'P=I$.
\end{definition}


\begin{definition}\label{defn_delta}
The {\em $s$-restricted isometry constant (RIC)} \cite{decodinglp}, $\delta_s$, for an $n \times m$ matrix $\Psi$ is the smallest real number satisfying $(1-\delta_s) \|x\|_2^2 \leq \|\Psi_T x\|_2^2 \leq (1+\delta_s) \|x\|_2^2$
for all sets $T \subseteq \{1,2,\dots n \}$ with $|T| \leq s$ and all real vectors $x$ of length $|T|$.%
\end{definition} 

It is easy to see that $\max_{T:|T| \le s} \|({\Psi_T}'\Psi_T)^{-1}\|_2 \le \frac{1}{1-\delta_s(\Psi)}$ \cite{decodinglp}.


\begin{definition}
Let $X$ and $Z$ be two random variables (r.v.) and let $\calb$ be a set of values that $Z$ can take.
\ben

\item We use $\ecalb$ to denote the {\em event} $Z \in \calb$, i.e. $\ecalb:= \{Z \in \calb\}$.

\item The probability of event $\ecalb$ can be expressed as \cite{grimmett},
$$\mathbf{P}(\ecalb) := \E[\mathbb{I}_\calb(Z)].$$
where
\bea
\mathbb{I}_\calb(Z): = \left\{ \begin{array}{cc}
                                 1 \ & \ \text{if} \ Z \in \calb \nn \\
                                 0  \ & \ \text{otherwise}
                                 \end{array} \right. \nn
\eea
is an indicator function of $Z$ on the set $\calb$ and $\E[\mathbb{I}_\calb(Z)]$ is the expectation of $\mathbb{I}_\calb(Z)$. 

\item Define $\mathbf{P}(\ecalb|X) := \E[\mathbb{I}_{\calb}(Z)|X]$ where $\E[\mathbb{I}_{\calb}(Z)|X]$ is the conditional expectation of $\mathbb{I}_\calb(Z)$ given $X$.

\een
\label{probdefs}
\end{definition}

Finally, RHS refers to the right hand side of an equation or inequality; w.p. means ``with probability"; and w.h.p. means ``with high probability".

\subsection{Preliminaries}
In this section we state certain results from literature, or certain lemmas which follow easily using these results, that will be used in proving our main result.

\subsubsection{Simple probability facts and matrix Hoeffding inequalities}
The following result follows directly from Definition \ref{probdefs}.
\begin{lem}
Suppose that $\calb$ is the set of values that the r.v.s $X,Y$ can take. Suppose that $\calc$ is a set of values that the r.v. $X$ can take.
For a $0 \le p \le 1$, if $\mathbf{P}(\ecalb|X) \ge p$ for all $X \in \calc$,  then $\mathbf{P}(\ecalb|\ecalc) \ge p$ as long as $\mathbf{P}(\ecalc)> 0$.
\label{rem_prob}
\end{lem}
Proof: This is the same as \cite[Lemma 11]{rrpcp_perf}.

The following lemma is an easy consequence of the chain rule of probability applied to a contracting sequence of events.
\begin{lem} \label{subset_lem}
For a sequence of events $E_0^e, E_1^e, \dots E_m^e$ that satisfy $E_0^e \supseteq E_1^e  \supseteq E_2^e \dots  \supseteq E_m^e$, the following holds
$$\mathbf{P}(E_m^e|E_0^e) = \prod_{k=1}^{m} \mathbf{P}(E_k^e | E_{k-1}^e).$$
\end{lem}
\begin{proof}
$\mathbf{P}(E_m^e|E_0^e) = \mathbf{P}(E_m^e, E_{m-1}^e, \dots E_0^e | E_0^e) = \prod_{k=1}^{m} \mathbf{P}(E_k^e | E_{k-1}^e, E_{k-2}^e, \dots E_0^e) = \prod_{k=1}^{m} \mathbf{P}(E_k^e | E_{k-1}^e)$.
\end{proof}

The following two results are corollaries of the matrix Hoeffding inequality \cite[Theorem 1.3]{tail_bound} that were proved in \cite{rrpcp_perf}. In the rest of the paper we often refer to them as the {\em Hoeffding corollaries.}


\begin{corollary}[Matrix Hoeffding conditioned on another random variable for a nonzero mean Hermitian matrix]\label{hoeffding_nonzero}
Given an $\alpha$-length sequence $\{Z_t\}$ of random Hermitian matrices of size $n\times n$, a r.v. $X$, and a set ${\cal C}$ of values that $X$ can take. Assume that, for all $X \in \calc$, (i) $Z_t$'s are conditionally independent given $X$; (ii) $\mathbf{P}(b_1 I \preceq Z_t \preceq b_2 I|X) = 1$ and (iii) $b_3 I \preceq \frac{1}{\alpha}\sum_t \E(Z_t|X) \preceq b_4 I $. Then for all $\epsilon > 0$,
\bea
&&\mathbf{P} (\lambda_{\max}(\frac{1}{\alpha}\sum_t Z_t) \leq b_4 + \epsilon|X) \geq 1- n \exp(-\frac{\alpha \epsilon^2}{8(b_2-b_1)^2}) \ \text{for all} \ X \in \calc
\nn \\
&&\mathbf{P} (\lambda_{\min}(\frac{1}{\alpha}\sum_t Z_t) \geq b_3 -\epsilon|X) \geq  1- n \exp(-\frac{\alpha \epsilon^2}{8(b_2-b_1)^2}) \ \text{for all} \ X \in \calc \nn
\eea
\end{corollary}
Proof: This is slight modification of \cite[Corollary 13]{rrpcp_perf}.

\begin{corollary}[Matrix Hoeffding conditioned on another random variable for an arbitrary nonzero mean matrix]\label{hoeffding_rec}
Given an $\alpha$-length sequence $\{Z_t\}$ of random Hermitian matrices of size $n\times n$, a r.v. $X$, and a set ${\cal C}$ of values that $X$ can take. Assume that, for all $X \in \calc$, (i) $Z_t$'s are conditionally independent given $X$; (ii) $\mathbf{P}(\|Z_t\|_2 \le b_1|X) = 1$ and (iii) $\|\frac{1}{\alpha}\sum_t \E( Z_t|X)\|_2 \le b_2$. Then, for all $\epsilon >0$,
$$\mathbf{P} (\|\frac{1}{\alpha}\sum_t Z_t\|_2 \leq b_2 + \epsilon|X) \geq 1-(n_1+n_2) \exp(-\frac{\alpha \epsilon^2}{32 b_1^2})  \ \text{for all} \ X \in \calc $$
\end{corollary}
Proof: This is slight modification of \cite[Corollary 14]{rrpcp_perf}.  

%

\subsubsection{Linear algebra results}
Kahan and Davis's $\sin \theta$ theorem \cite{davis_kahan} studies the effect of a Hermitian perturbation, $\mathcal{H}$, on a Hermitian matrix, $\mathcal{A}$.


\begin{theorem}[$\sin \theta$ theorem \cite{davis_kahan}] \label{sin_theta}
Given two Hermitian matrices $\mathcal{A}$ and $\mathcal{H}$ satisfying
\beq
\mathcal{A} = \left[ \begin{array}{cc} E & E_{\perp} \\ \end{array} \right]
\left[\begin{array}{cc} A\ & 0\ \\ 0 \ & A_{\perp} \\ \end{array} \right]
\left[ \begin{array}{c} E' \\ {E_{\perp}}' \\ \end{array} \right], \
\mathcal{H} = \left[ \begin{array}{cc} E & E_{\perp} \\ \end{array} \right]
\left[\begin{array}{cc} H \ & B'\ \\ B \ & H_{\perp} \\ \end{array} \right]
\left[ \begin{array}{c} E' \\ {E_{\perp}}' \\ \end{array} \right]
\label{def_A_H}
\eeq
where $[E \ E_{\perp}]$ is an orthonormal matrix. The two ways of representing $\mathcal{A}+\mathcal{H}$ are
\beq
\mathcal{A} + \mathcal{H}  = \left[ \begin{array}{cc} E & E_{\perp} \\ \end{array} \right]
\left[\begin{array}{cc} A + H \ & B'\ \\ B \ & A_{\perp} + H_{\perp} \\ \end{array} \right]
\left[ \begin{array}{c} E' \\ {E_{\perp}}' \\ \end{array} \right]
= \left[ \begin{array}{cc} F & F_{\perp} \\ \end{array} \right]
\left[\begin{array}{cc} \Lambda\ & 0\ \\ 0 \ & \Lambda_{\perp} \\ \end{array} \right]
\left[ \begin{array}{c} F' \\ {F_{\perp}}' \\ \end{array} \right] \nn
\eeq
where $[F\ F_{\perp}]$ is another orthonormal matrix. Let $\mathcal{R} := (\mathcal{A}+\mathcal{H}) E - \mathcal{A}E = \mathcal{H} E $. If $ \lambda_{\min}(A) >\lambda_{\max}(\Lambda_{\perp})$, then
\beq
\|(I-F F')E \|_2 \leq \frac{\|\mathcal{R}\|_2}{\lambda_{\min}(A) - \lambda_{\max}(\Lambda_{\perp})} \nn
\eeq
\end{theorem}

Next we state the Weyl's theorem (Weyl's inequality for matrices) \cite[page 181]{hornjohnson} and the Ostrowski's theorem  \cite[page 224]{hornjohnson}.

\begin{theorem}[Weyl \cite{hornjohnson}]\label{weyl} 
Let $\mathcal{A}$ and $\mathcal{H}$ be two  $n \times n$ Hermitian matrices. For each $i = 1,2,\dots,n$ we have
$$\lambda_i(\mathcal{A}) + \lambda_{\min}(\mathcal{H}) \leq \lambda_i(\mathcal{A}+\mathcal{H}) \leq \lambda_i(\mathcal{A}) + \lambda_{\max}(\mathcal{H})$$
\end{theorem}

\begin{theorem}[Ostrowski \cite{hornjohnson}]\label{ost} 
Let $H$ and $W$ be $n \times n$ matrices, with $H$ Hermitian and $W$ nonsingular. For each $i=1,2 \dots n$, there exists a positive real number $\theta_i$ such that $\lambda_{\min} (WW') \leq \theta_i \leq \lambda_{\max}(W{W}')$ and $\lambda_i(W H {W}') = \theta_i \lambda_i(H)$. Therefore,
$$\lambda_{\min}(W H {W}') \geq \lambda_{\min} (W{W}') \lambda_{\min} (H)$$
\end{theorem}

The following lemma uses the $\sin \theta$ theorem and Weyl's theorem. It generalizes the idea of \cite[Lemma 30]{rrpcp_perf}.
\begin{lem} 
\label{sin_theta_weyl}
Suppose that two Hermitian matrices $\mathcal{A}$ and $\mathcal{H}$ can be decomposed as in (\ref{def_A_H})
where $[E \ E_{\perp}]$ is an orthonormal matrix and $A$ is a $c \times c$ matrix.
Also, suppose that the EVD of $\mathcal{A}+\mathcal{H}$ is
$$\mathcal{A}+\mathcal{H} \overset{EVD}{=} \left[ \begin{array}{cc} F & F_{\perp} \\ \end{array} \right]
\left[\begin{array}{cc} \Lambda\ & 0\ \\ 0 \ & \Lambda_{\perp} \\ \end{array} \right]
\left[ \begin{array}{c} F' \\ {F_{\perp}}' \\ \end{array} \right]$$
where
$\Lambda$ is a $c \times c$ diagonal matrix.
If $\lambda_{\min}(A) >  \lambda_{\max}(A_\perp) + \|\mathcal{H}\|_2$, then
\beq
\|(I-F F')E \|_2 \leq \frac{\|\mathcal{H}\|_2}{\lambda_{\min}(A) -  \lambda_{\max}(A_\perp) - \|\mathcal{H}\|_2} \nn
\eeq
\end{lem}
\begin{proof} 
By definition of EVD, $[F \ F_\perp]$ is an orthonormal matrix.
By the $\sin \theta$ theorem, if $\lambda_{\min}(A) > \lambda_{\max}(\Lambda_{\perp})$, then  
$
\|(I-F F')E \|_2 \leq \frac{\|\mathcal{R}\|_2}{\lambda_{\min}(A) - \lambda_{\max}(\Lambda_{\perp})}
$
where $\mathcal{R} := \mathcal{H} E $. Clearly, $\|\mathcal{R}\|_2 \le \|\mathcal{H}\|_2$.
Since $\lambda_{\min}(A) > \lambda_{\max}(A_\perp)$ and $A$ is a $c \times c$ matrix, thus, $\lambda_{c+1}(\mathcal{A}) =\lambda_{\max}(A_\perp)$.

By definition of EVD (eigenvalues arranged in non-increasing order) and since $\Lambda$ is a $c \times c$ matrix, $\lambda_{c+1}(\mathcal{A}+\mathcal{H}) = \lambda_{\max}(\Lambda_{\perp})$.
By Weyl's theorem, $\lambda_{\max}(\Lambda_{\perp}) = \lambda_{c+1}(\mathcal{A}+\mathcal{H}) \le \lambda_{c+1}(\mathcal{A}) + \lambda_{\max}(\mathcal{H})$. Since $\lambda_{\max}(\mathcal{H}) \leq \|\mathcal{H}\|_2$, the result follows.
\end{proof}

The following lemma is a minor modification of \cite[Lemma 10]{rrpcp_perf}.
\begin{lem} \label{lemma0}
Suppose that $P$, $\Phat$ and $Q$ are three basis matrices, $P$ and $\Phat$ are of same size. Also, ${Q}'P = 0$ and $\|(I-\Phat{\Phat}')P\|_2 \le \zeta_*^+$. Then,
\ben
  \item $\|(I-\Phat{\Phat}')PP'\|_2 =\|(I - P{P}')\Phat{\Phat}'\|_2 =  \|(I - P P')\Phat\|_2 = \|(I - \Phat \Phat')P\|_2 \le  \zeta_*^+$
  \item $\|P{P}' - \Phat {\Phat}'\|_2 \leq 2 \|(I-\Phat{\Phat}')P\|_2 \le 2 \zeta_*^+$
  \item $\|{\Phat}' Q\|_2 \leq \zeta_*^+$ \label{lem_cross}
  \item $ \sqrt{1-{\zeta_*^+}^2} \leq \sigma_i((I-\Phat \Phat')Q)\leq 1 $
\een
\end{lem}
\begin{proof} The result follows exactly as in the proof of \cite[Lemma 10]{rrpcp_perf}. \end{proof} 

\subsubsection{Sparse Recovery Error Bound}
The following is a minor modification of \cite[Theorem 1]{candes_rip} applied to exact sparse signals.
\begin{theorem}[\cite{candes_rip}]
\label{candes_csbound}
Suppose we observe $y := \Psi x + z$
where $z$ is the noise. Let $\hat{x}$ be the solution to following problem
\beq
\min_{x} \|x\|_1 \ \text{subject to} \  \|y - \Psi x\|_2 \leq \xi  \label{*}
\eeq
Assume that $x$ is $s$-sparse, $\|z\|_2 \leq \xi$ and $\delta_{2s}(\Psi) \le b <  (\sqrt{2}-1)$. The solution of (\ref{*}) obeys
$\|\hat{x} - x\|_2 \leq C_1 \xi$ with $C_1 := \frac{4\sqrt{1+b}}{1-(\sqrt{2}+1)b}$.
\end{theorem}

%

\section{Problem Definition and Model Assumptions}\label{probdef}
We give the problem definition below followed by the model and three key assumptions.

\subsection{Problem Definition}
\label{model}
The measurement vector at time $t$, $M_t$, is an $n$ dimensional vector which can be decomposed as
\beq
M_t = L_t + S_t \label{problem_defn}
\eeq
Here $S_t$ is a sparse vector with support set size at most $s$ and minimum magnitude of nonzero values at least $S_{\min}$. $L_t$ is a dense but low dimensional vector, i.e. $L_t = P_{(t)} a_t$ where $P_{(t)}$ is an $n \times r_{(t)}$ basis matrix with $r_{(t)} \ll n$, that changes every so often. $P_{(t)}$ and $a_t$ change according to the model given below. We are given an accurate estimate of the subspace in which the initial $t_\train$ $L_t$'s lie, i.e. we are given a basis matrix $\Phat_0$ so that $\|(I-\Phat_0 \Phat_0')P_0 \|_2$ is small.
Here $P_0$ is a basis matrix for $\Span({\cal L}_{t_{\train}})$, i.e. $\Span(P_0) = \Span({\cal L}_{t_{\train}})$. Also,  for the first $t_{\train}$ time instants, $S_t$ is either zero or very small. 
The goal is
\ben
\item to estimate both $S_t$ and $L_t$ at each time $t > t_\train$, and
\item to estimate $\Span(P_{(t)})$ every-so-often, i.e., update $\hat{P}_{(t)}$ so that the subspace estimation error, $\SE_{(t)}:=\|(I- \hat{P}_{(t)} \hat{P}_{(t)}')P_{(t)}\|_2$ is small. \label{item2}
\een


{\bf Notation for $S_t$. }
Let $T_t:=\{i: \  (S_t)_i \neq 0 \}$ denote the support of $S_t$.
Define $$S_{\min}: = \min_{t > t_{\train}} \ \min_{i \in T_t} |(S_t)_i |, \ \ \text{and} \ \ s: = \max_t |T_t|$$

\begin{ass}[Model on $L_t$] \label{model_lt}
We assume that $L_t = P_{(t)} a_t$ where $P_{(t)}$ and $a_t$ satisfy the following.
\ben
\item $P_{(t)} = P_j$ for all $t_j \leq t <t_{j+1}$, $j=0,1,2 \cdots J$, where $P_j$ is an $n \times r_j$ basis matrix with $r_j  \ll n$ and $r_j \ll (t_{j+1} - t_j)$. We let $t_0=0$ and $t_{J+1}$ equal the sequence length. This can be infinity also.
At the change times, $t_j$, $P_j$ changes as $P_j = [( P_{j-1} \setminus P_{j,\old}) \ P_{j,\new}]$. Here, $P_{j,\new}$ is an $n \times c_{j,\new}$ basis matrix with $P_{j,\new}'P_{j-1} = 0$ and $P_{j,\old}$ contains $c_{j,\old}$ columns of $P_{j-1}$. Thus $r_j = r_{j-1} + c_{j,\new} - c_{j,\old}$. Also, $0 < t_{\train} \le t_1$.
This model is illustrated in Fig. \ref{add_del_model}.

\item There exists a constant $c_{\max}$ such that $0 \le c_{j,\new} \leq c_{\max}$ and $\sum_{i=1}^{j} (c_{i,\new} - c_{i,\old}) \leq c_{\max}$ for all $j$. Let $r_{\max}: = r_0 + c_{\max}$. Thus, $r_j = r_0+\sum_{i=1}^{j} (c_{i,\new} - c_{i,\old}) \leq r_0+c_{\max} = r_{\max}$, i.e., the rank of $P_j$ is upper bounded by $r_{\max}$.

\item $a_t:={P_{(t)}}'L_t$, is a $r_j$ length random variable (r.v.) with the following properties.
\ben
\item $a_t$'s are mutually independent over $t$.
\item $a_t$ is a zero mean bounded r.v., i.e. $\E(a_t) = 0$ and there exists a constant $\gamma_*$ such that $\|a_t\|_{\infty} \leq \gamma_*$ for all $t$.
\item Its covariance matrix $\Lambda_t: = \text{Cov}[a_t] = \E(a_ta_t')$ is diagonal with $\lambda^-: = \min_t \lambda_{\min}(\Lambda_t)  > 0$ and $\lambda^+:=\max_t \lambda_{\max} (\Lambda_t) < \infty$. Thus, the condition number of any $\Lambda_t$ is bounded by $f: = \frac{\lambda^+}{\lambda^-}$.
\een

\een
\end{ass}

Also, $P_j$ and $a_t$ satisfy the assumptions discussed in the next three subsections.

\begin{definition}\label{defn_at}
The following notation will be used frequently. Let $P_{j,*}:= P_{(t_j-1)} = P_{j-1}$. For $t \in [t_j, t_{j+1}-1]$, let $a_{t,*} := {P_{j,*}}'L_t = {P_{j-1}}'L_t$ be the projection of $L_t$ along $P_{j,*}$ of which $a_{t,*,\text{nz}}:= ( P_{j-1} \setminus P_{j,\old})' L_t$ is the nonzero part. Also, let $a_{t,\new} := P_{j,\new}'L_t$ be the projection of $L_t$ along the newly added directions. Thus,
$$a_{t,*} = \vect{a_{t,*,\text{nz}}}{{\bf 0}} \ \text{and} \  a_t = \vect{a_{t,*,\text{nz}}}{a_{t,\new}}$$
where ${\bf 0}$ is a $c_{j,\old}$ length zero vector (since ${P_{j,\old}}'L_t ={\bf  0}$).
Using the above, for $t \in [t_j, t_{j+1}-1]$, $L_t$ can be rewritten as
$$L_t =  P_j a_t = ( P_{j-1} \setminus P_{j,\old}) a_{t,*,\text{nz}} +  P_{j,\new} a_{t,\new} = P_{j,*}a_{t,*} +  P_{j,\new} a_{t,\new}$$
and $\Lambda_t$ can be split as
$$\Lambda_t = \left[ \begin{array}{cc}
(\Lambda_t)_{*,\text{nz}} &  0 \nn \\
0  & (\Lambda_t)_\new  \nn \\
\end{array}
\right]$$
where $(\Lambda_t)_{*,\text{nz}}: = \text{Cov}(a_{t,*,\text{nz}})$ and $(\Lambda_t)_\new = \text{Cov}(a_{t,\new})$ are diagonal matrices.
\end{definition}


\begin{figure}[t!]
\centerline{
\epsfig{file = 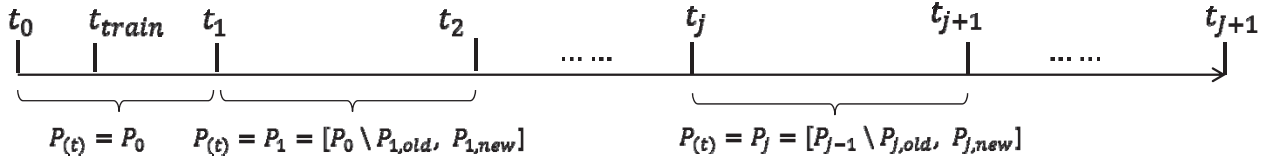, width =18cm}
}
\caption{\small{The subspace change model given in Sec \ref{model}. Here $t_0=0$. 
}}
\label{add_del_model}
\end{figure}

\subsection{Slow subspace change}\label{slowss}
By slow subspace change we mean all of the following.
\ben
\item First, the delay between consecutive subspace change times, $t_{j+1} - t_j$, is large enough. 

\item Second, the projection of $L_t$ along the newly added directions, $a_{t,\new}$, is initially small, i.e. $\max_{t_j \le t < t_j + \alpha} \|a_{t,\new}\|_{\infty} \le \gamma_\new, \ \text{with} \ \gamma_\new \ll \gamma_* \ \text{and} \  \gamma_\new \ll S_{\min},$ but can increase gradually. We model this as follows. Split the interval $[t_j, t_{j+1}-1]$ into $\alpha$ length periods. We assume that
    $$\max_j \max_{t \in [t_j + (k-1) \alpha,t_j + k \alpha-1]} \|a_{t,\new}\|_{\infty} \le \gamma_{\new,k}:=  \min(v^{k-1} \gamma_\new,\gamma_*)$$
    for a $v>1$ but not too large\footnote{Small $\gamma_\new$ and slowly increasing $\gamma_{\new,k}$ is needed for the noise seen by the sparse recovery step to be small. However, if $\gamma_\new$ is zero or very small, it will be impossible to estimate the new subspace. This will not happen in our model because $\gamma_\new \ge \lambda^- > 0$.}. This assumption is verified for real video data in \cite[Sec X-B]{rrpcp_perf}.

\item Third, the number of newly added directions is small, i.e.  $c_{j,\new} \le c_{\max} \ll r_{0}$. This is also verified in \cite[Sec X-B]{rrpcp_perf}.

\een

\subsection{Measuring denseness of a matrix and its relation with RIC}
\label{denseness}
For a tall $n \times r$ matrix, $B$, or for a $n \times 1$ vector, $B$, we define the the denseness coefficient as follows \cite{rrpcp_perf}:
\bea
\kappa_s(B): = \max_{|T| \le s} \frac{\|{I_T}' B\|_2}{\|B\|_2}.
\label{defkappa}
\eea
where $\|.\|_2$ is the matrix or vector 2-norm respectively. Clearly, $\kappa_s(B) \le 1$.
As explained in \cite{rrpcp_perf}, $\kappa_s$ measures the denseness (non-compressibility) of a vector $B$ or of the columns of a matrix $B$. For a vector, a small value indicates that its entries  are spread out, i.e. it is a dense vector. A large value indicates that it is compressible (approximately or exactly sparse). Similarly, for an $n \times r$ matrix $B$, a small $\kappa_s$ means that most (or all) of its columns are dense vectors.

For a basis matrix $P$, $\kappa_s(PP') = \kappa_s(P)$ and thus $\kappa_s(P)$ is a property of $\Span(P)$ \cite{rrpcp_perf}. 

\begin{remark}
A better way to quantify denseness of a matrix $B$ would be to define the denseness coefficient as $\max_{|T| \le s} \|{I_T}'Q(B)\|_2$ where $Q(B)$ is a basis matrix for $\Span(B)$, e.g. it can be obtained by QR decomposition on $B$. This definition will ensure that the denseness coefficient is a property of $\Span(B)$ for any matrix $B$. It is easy to see that $\|{I_T}' B\|_2 \le \|{I_T}'Q(B)\|_2 \|B\|_2$. Thus, even with this new definition, all our results, and all results of \cite{rrpcp_perf}, will go through without any change. However, we keep the definition of (\ref{defkappa}) because it was used in \cite{rrpcp_perf} and the current work uses certain lemmas from \cite{rrpcp_perf}.
\end{remark}

The following lemma was proved in  \cite{rrpcp_perf}.
\begin{lem}\label{delta_kappa}
For an $n \times r$ basis matrix $P$  (i.e $P$ satisfying $P'P=I$), 
$$\delta_s(I-PP') = \kappa_s^2 (P).$$
\end{lem}
In other words, if $P$ is dense enough (small $\kappa_s$), then the RIC of $I-PP'$ is small. 
%
%
As we explain in \cite[Sec IV-D]{rrpcp_perf}, $\kappa_s(B)$ is related to the denseness assumption required by PCP \cite{rpca}.

\begin{figure}[t!]
\centerline{
\epsfig{file = 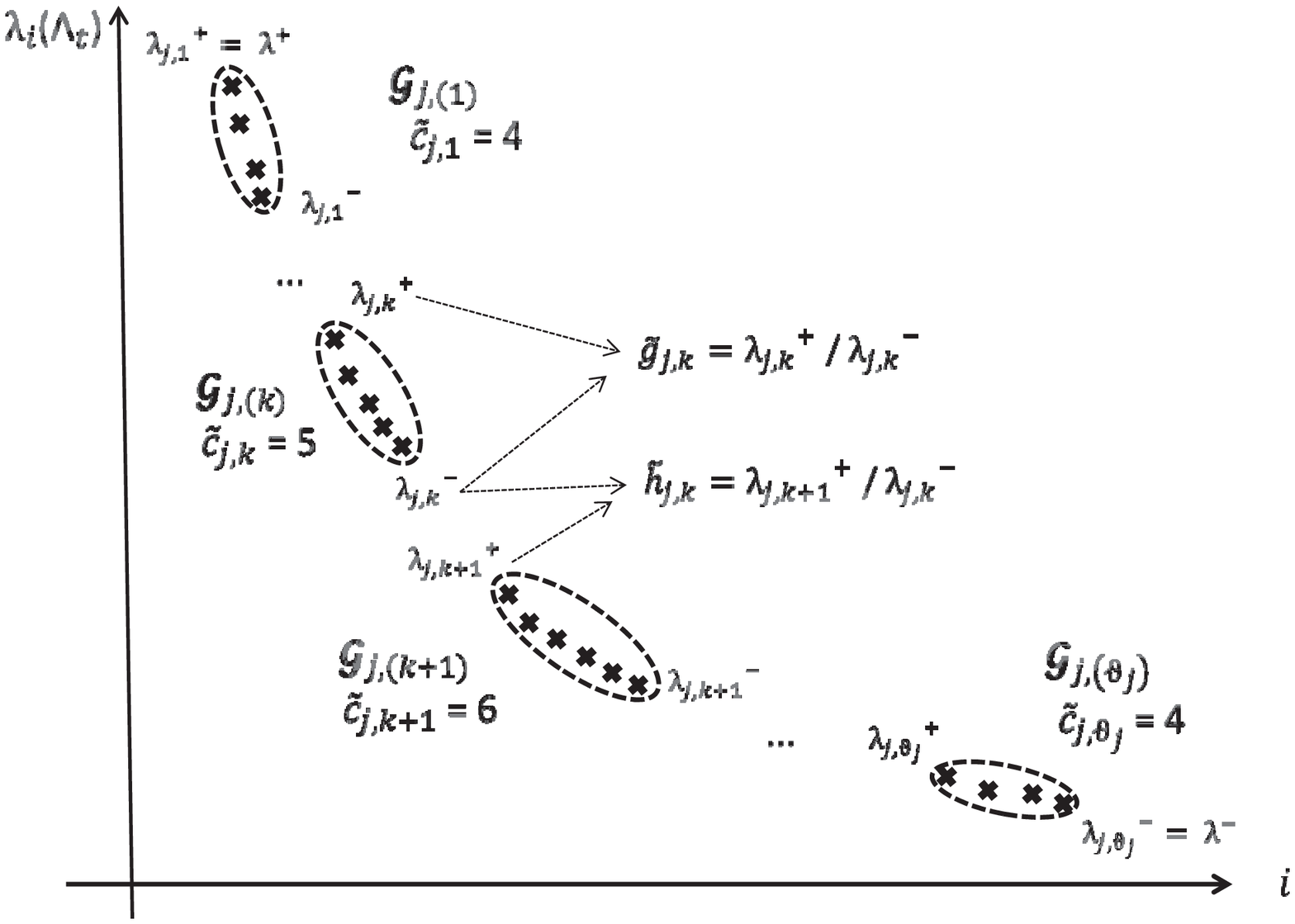, height = 4cm, width = 10cm}
}
\caption{\small{We illustrate the clustering assumption. Assume $\Lambda_t = \Lambda_{\tilde{t}_j}$.
}}
\label{clustering_diag}
\end{figure}

%
\subsection{Clustering assumption}\label{eigencluster}
\newcommand{\group}{\mathcal{G}}
For positive integers $K$ and $\alpha$, let $\tilde{t}_j := t_j + K \alpha$. We set their values in our main result, Theorem \ref{thm2}.  
Recall from the model on $L_t$ and the slow subspace change assumption that new directions, $P_{j,\new}$, get added at $t=t_j$ and initially, for the first $\alpha$ frames,  the projection of $L_t$ along these directions is small (and thus their variances are small), but can increase gradually. It is fair to assume that by $t=\tilde{t}_j$, the variances along these new directions have stabilized and do not change much for $t \in [\tilde{t}_j, t_{j+1}-1]$. It is also fair to assume that the same is true for the variances along the existing directions, $P_{j-1}$. In other words, we assume that the matrix $\Lambda_t$ is either constant or does not change much during this period. Under this assumption, we assume that we can cluster its eigenvalues (diagonal entries) into a few clusters such that the distance between consecutive clusters is large and the distance between the smallest and largest element of each cluster is small. 
We make this precise below.

\begin{ass} \label{assu1}
Assume the following.
\ben
\item 
Either $\Lambda_t = \Lambda_{\tilde{t}_j}$ for all ${t \in [\tilde{t}_j, t_{j+1}-1]}$ or $\Lambda_t$ changes very little during this period so that for each $i=1,2,\cdots,r_j$, $\min_{t \in [\tilde{t}_j, t_{j+1}-1]} \lambda_i(\Lambda_t) \ge \max_{t \in [\tilde{t}_j, t_{j+1}-1]} \lambda_{i+1}(\Lambda_t)$. 

\item  Let $ \group_{j,(1)}, \group_{j,(2)}, \cdots, \group_{j,(\vartheta_j)} $ be a partition of the index set $\{1,2, \dots r_j\}$ so that $\min_{i \in \group_{j,(k)}} \min_{t \in [\tilde{t}_j, t_{j+1}-1]} \lambda_i(\Lambda_t)  > \max_{i \in \group_{j,(k+1)}} \max_{t \in [\tilde{t}_j, t_{j+1}-1]} \lambda_i(\Lambda_t)$, i.e. the first group/cluster contains the largest set of eigenvalues, the second one the next smallest set and so on (see Fig \ref{clustering_diag}). Let
\ben
\item $G_{j,k} := (P_j)_{ \group_{j,(k)} }$ be the corresponding cluster of eigenvectors, then $P_j = [G_{j,1},G_{j,2},\cdots,G_{j,\vartheta_j}]$;
\item $\tilde{c}_{j,k} := |\group_{j,(k)}|$ be the number of elements in $\group_{j,(k)}$, then $\sum_{k=1}^{\vartheta_j} \tilde{c}_{j,k} = r_j$;

\item  ${\lambda_{j,k}}^- := \min_{i\in \group_{j,(k)} } \min_{t\in [\tilde{t}_j, t_{j+1}-1]}  \lambda_i (\Lambda_t)$,  ${\lambda_{j,k}}^+ := \max_{i \in \group_{j,(k)} } \max_{t\in [\tilde{t}_j, t_{j+1}-1]}  \lambda_i (\Lambda_t)$ and ${\lambda_{j,\vartheta_j+1} }^+:= 0$;

\item $\tilde{g}_{j,k} := {\lambda_{j,k}}^+ / {\lambda_{j,k}}^- $ (notice that $\tilde{g}_{j,k} \ge 1$);

\item $\tilde{h}_{j,k} := {\lambda_{j,k+1}}^+ / {\lambda_{j,k}}^-$ (notice that $\tilde{h}_{j,k} < 1$);
\item $\tilde{g}_{\max} := \max_j  \max_{k = 1,2,\cdots,\vartheta_j} \tilde{g}_{j,k}$, $\tilde{h}_{\max} := \max_j  \max_{k = 1,2,\cdots,\vartheta_j} \tilde{h}_{j,k}$, $\tilde{c}_{\min} : = \min_j \min_{k = 1,2,\cdots,\vartheta_j} \tilde{c}_{j,k}$
\item  $\vartheta_{\max}: = \max_j \vartheta_j$
\een

 We assume that $\tilde{g}_{\max}$ is small enough (the distance between the smallest and largest eigenvalues of a cluster is small)
and $\tilde{h}_{\max}$ is small enough (distance between consecutive clusters is large). We quantify this in Theorem \ref{thm2}.
\een
\end{ass}

\begin{remark}
The assumption above can, in fact, be relaxed to only require the following. The matrices $\Lambda_t$ are such that there exists a partition, $ \group_{j,(1)}, \group_{j,(2)}, \cdots, \group_{j,(\vartheta_j)} $, of the index set $\{1,2, \dots r_j\}$ so that $\min_{i \in \group_{j,(k)}} \min_{t \in [\tilde{t}_j, t_{j+1}-1]} \lambda_i(\Lambda_t)  > \max_{i \in \group_{j,(k+1)}} \max_{t \in [\tilde{t}_j, t_{j+1}-1]} \lambda_i(\Lambda_t)$. Define all quantities as above. We assume that $\tilde{g}_{\max}$ and $\tilde{h}_{\max}$ are small enough.
\end{remark}

\section{ReProCS with cluster-PCA (ReProCS-cPCA)} \label{rep_del_perf} 

We first briefly recap the main idea of projection-PCA (proj-PCA) which was used in \cite{rrpcp_perf}. The ReProCS with cluster-PCA (ReProCS-cPCA) algorithm is then explained. In Sec \ref{prac_algo}, we discuss how to set its parameters in practice when the model may not be known. The need for proj-PCA is explained in Sec \ref{ppca_need}.
We need the following notation.
\begin{definition}\label{defn_intervals}
Let  $\tilde{t}_j := t_j + K\alpha$. Define the following time intervals
\ben
\item $\mathcal{I}_{j,k}:= [t_j + (k-1)\alpha, t_j + k\alpha-1]$ for $k=1,2,\cdots,K$.
\item $\tilde{\mathcal{I}}_{j,k} := [\tilde{t}_j + (k-1) \tilde{\alpha}, \tilde{t}_j + k \tilde{\alpha}-1]$ for $k = 1,2,\cdots, \vartheta_j$.
\item $\tilde{\mathcal{I}}_{j,\vartheta_j+1} := [\tilde{t}_j + \vartheta_j \tilde{\alpha}, t_{j+1}-1]$.
\een
Notice that $[t_j, t_{j+1}-1] = (\cup_{k=1}^{K} \mathcal{I}_{j,k})\cup (\cup_{k=1}^{\vartheta_j} \tilde{\mathcal{I}}_{j,k}) \cup \mathcal{\tilde{I}}_{j,\vartheta_j+1}$. Also, $K$, $\alpha$ and $\tilde{\alpha}$ are parameters given in Algorithm \ref{ReProCS_del}.
\end{definition}

\subsection{The Projection-PCA algorithm}

Given a data matrix $\mathcal{D}$, a basis matrix $P$ and an integer $r$, projection-PCA (proj-PCA) applies PCA on $\mathcal{D}_{\text{proj}}:=(I-PP')\mathcal{D}$, i.e., it computes the top $r$ eigenvectors (the eigenvectors with the largest $r$ eigenvalues) of $\frac{1}{\alpha_{\mathcal{D}}} \mathcal{D}_{\text{proj}} {\mathcal{D}_{\text{proj}}}'$. Here $\alpha_{\mathcal{D}}$ is the number of column vectors in $\mathcal{D}$. This is summarized in Algorithm \ref{algo_pPCA}.

If $P =[.]$, then projection-PCA reduces to standard PCA, i.e. it computes the top $r$ eigenvectors of $\frac{1}{\alpha_{\mathcal{D}}} \mathcal{D} {\mathcal{D}}'$.

We should mention that the idea of projecting perpendicular to a partly estimated subspace has been used in different contexts in past work \cite{PP_PCA_Li_Chen, mccoy_tropp11}.

%
%

\begin{algorithm}
\caption{projection-PCA: $Q \leftarrow \text{proj-PCA}(\mathcal{D},P,r)$}\label{algo_pPCA}
\ben
\item Projection: compute $\mathcal{D}_{\text{proj}} \leftarrow (I - P P') \mathcal{D}$
\item PCA: compute $\frac{1}{\alpha_{\mathcal{D}}}  \mathcal{D}_{\text{proj}}{\mathcal{D}_{\text{proj}}}' \overset{EVD}{=}
\left[ \begin{array}{cc}Q & Q_{\perp} \\\end{array}\right]
\left[ \begin{array}{cc} \Lambda & 0 \\0 & \Lambda_{\perp} \\\end{array}\right]
\left[ \begin{array}{c} Q' \\ {Q_{\perp}}'\\\end{array}\right]$
where $Q$ is an $n \times r$ basis matrix and  $\alpha_{\mathcal{D}}$ is the number of columns in $\mathcal{D}$.
\een
\end{algorithm}


\begin{algorithm}[ht]
\caption{Recursive Projected CS with cluster-PCA (ReProCS-cPCA)}\label{ReProCS_del}
{\bf Parameters: } algorithm parameters: $\xi$, $\omega$, $\alpha$, $\tilde{\alpha}$, $K$, model parameters: $t_j$, $r_0$, $c_{j,\new}$, $\vartheta_j$ and $\tilde{c}_{j,i}$ 
\\
{\bf Input: } $n \times 1$ vector, $M_t$, and $n \times r_0$ basis matrix $\hat{P}_0$.
{\bf Output: } $n \times 1$ vectors $\Shat_t$ and $\Lhat_t$, and $n \times r_{(t)}$ basis matrix $\Phat_{(t)}$. 
\\
{\bf  Initialization: } 
Let $\Phat_{(t_\train)} \leftarrow \Phat_0$.
Let $j \leftarrow 1$, $k\leftarrow 1$. 
For $t > t_{\train}$, do the following:
\ben
\item {\bf Estimate $T_t$ and $S_t$ via Projected CS: }
\ben
\item \label{othoproj} Nullify most of $L_t$: compute $\Phi_{(t)} \leftarrow I-\Phat_{(t-1)} {\Phat_{(t-1)}}'$, $y_t \leftarrow \Phi_{(t)} M_t$
\item \label{Shatcs} Sparse Recovery: compute $\hat{S}_{t,\cs}$ as the solution of $\min_{x} \|x\|_1 \ s.t. \ \|y_t - \Phi_{(t)} x\|_2 \leq \xi$
\item \label{That} Support Estimate: compute $\hat{T}_t = \{i: \ |(\hat{S}_{t,\cs})_i| > \omega\}$
\item \label{LS} LS Estimate of $S_t$: compute $(\hat{S}_t)_{\hat{T}_t}= ((\Phi_t)_{\hat{T}_t})^{\dag} y_t, \ (\hat{S}_t)_{\hat{T}_t^{c}} = 0$
\een
\item {\bf Estimate $L_t$. } $\hat{L}_t = M_t - \hat{S}_t$.
\item \label{PCA} 
{\bf Update $\Phat_{(t)}$}: 
\ben
\item If $t \neq t_j + q\alpha-1$ for any $q=1,2, \dots K$ and $t \neq t_j + K \alpha + \vartheta_j \tilde{\alpha} -1$,
\ben
\item set $\Phat_{(t)} \leftarrow \Phat_{(t-1)}$
\een
\item {\bf Addition: Estimate $\Span(P_{j,\new})$ iteratively using proj-PCA: } If $t = t_j + k\alpha-1$
\ben

\item $\Phat_{j,\new,k} \leftarrow \text{proj-PCA} ([\hat{L}_t; t\in \mathcal{I}_{j,k}], \Phat_{j-1}, c_{j,\new})$

\item set $\Phat_{(t)} \leftarrow [\Phat_{j-1} \ \Phat_{j,\new,k}]$.

\item If $k=K$, reset $k \leftarrow 1$; else increment $k \leftarrow k+1$.
\een

\item {\bf Deletion: Estimate $\Span(P_j)$ by cluster-PCA:} If $t= t_j + K \alpha + \vartheta_j \tilde{\alpha} -1$,

\ben
\item For $i = 1,2,\cdots, \vartheta_j$,
  \bi
  \item $\hat{G}_{j,i} \leftarrow \text{proj-PCA}( [\hat{L}_t; t \in \tilde{\mathcal{I}}_{j,k}], [\hat{G}_{j,1},\hat{G}_{j,2}, \dots \hat{G}_{j,i-1}], \tilde{c}_{j,i})$
  \ei
  End for
\item set $\Phat_j \leftarrow [\hat{G}_{j,1},\cdots, \hat{G}_{j,\vartheta_j}]$ and set $\Phat_{(t)} \leftarrow \Phat_{j}$.
\item increment $ j \leftarrow j+1$.
\een

\een
\een
\end{algorithm}

\subsection{The ReProCS-cPCA algorithm}

ReProCS-cPCA is summarized in Algorithm \ref{ReProCS_del}. It proceeds as follows. The algorithms begins with the knowledge of $\Phat_0$ and initializes $\Phat_{(t_\train)} \leftarrow \Phat_0$. $\Phat_0$ can be computed as the top $r_0$ left singular vectors of ${\cal M}_{t_\train}$  (since, by assumption, ${\cal S}_{t_\train}$ is either zero or very small).
For $t > t_\train$, the following is done. Step 1 projects $M_t$ perpendicular to $\Phat_{(t-1)}$, solves the $\ell_1$ minimization problem, followed by support recovery and finally computes a least squares (LS) estimate of $S_t$ on its estimated support. This final estimate $\Shat_t$ is used to estimate $L_t$ as $\Lhat_t = M_t-\Shat_t$ in step 2. The sparse recovery error, $e_t: = \Shat_t -S_t$. Since $\Lhat_t = M_t - \Shat_t$, $e_t$ also satisfies $e_t =L_t-\Lhat_t$. Thus, a small $e_t$ (accurate recovery of $S_t$) means that $L_t$ is also recovered accurately. Step 3a is used at times when no subspace update is done. In step 3b, the estimated $\Lhat_t$'s are used to obtain improved estimates of $\Span(P_{j,\new})$ every $\alpha$ frames for a total of $K \alpha$ frames using the proj-PCA procedure given in Algorithm \ref{algo_pPCA}. As explained in \cite{rrpcp_perf}, within $K$ proj-PCA updates ($K$ chosen as given in Theorem \ref{thm2}), it can be shown that both $\|e_t\|_2$ and the subspace error,  $\SE_{(t)}: = \|(I - \Phat_{(t)} \Phat_{(t)}') P_{(t)}\|_2$, drop down to a constant times ${\zeta}$. In particular, if at $t=t_{j}-1$,  $\SE_{(t)} \le r \zeta$, then at $t= \tilde{t}_j:=t_j + K \alpha$, we can show that $\SE_{(t)} \le (r + c_{\max}) \zeta$. Here $r:=r_{\max} = r_0 + c_{\max}$.

To bring $\SE_{(t)}$ down to $r \zeta$ before $t_{j+1}$, we need a step so that by $t=t_{j+1}-1$ we have an estimate of only $\Span(P_j)$, i.e. we have ``deleted" $\Span(P_{j,\old})$. One simple way to do this is by standard PCA: at $t=\tilde{t}_j + \tilde{\alpha}-1$, compute $\Phat_j \leftarrow \text{proj-PCA}([\Lhat_t; t \in \tilde{\mathcal{I}}_{j,1}], [.], r_j)$ and let $\Phat_{(t)} \leftarrow \Phat_j$.
Using the $\sin \theta$ theorem and the Hoeffding corollaries, it can be shown that, as long as $f$ is small enough, doing this is guaranteed to give an accurate estimate of $\Span(P_j)$. However $f$ being small is not compatible with the slow subspace change assumption. Notice from Sec \ref{probdef} that $\lambda^- \le \gamma_\new$ and $\E[||L_t||_2^2] \le r \lambda^+$. Slow subspace change implies that $\gamma_\new$ is small. Thus, $\lambda^-$ is small.
However, to allow $L_t$ to have large magnitude, $\lambda^+$ needs to be large. Thus,  $f = \lambda^+ / \lambda^-$ cannot be small unless we require that $L_t$ has small magnitude for all times $t$.



In step 3c, we introduce a generalization of the above strategy called cluster-PCA, that removes the bound on $f$, but instead only requires that the eigenvalues of $\text{Cov}(L_t)$ be sufficiently clustered as explained in Sec \ref{eigencluster}. The main idea is to recover one cluster of entries of $P_j$ at a time. In the $k^{th}$ iteration, we apply proj-PCA on  $[\Lhat_t; t \in \tilde{I}_{j,k}]$ with $P \leftarrow [\hat{G}_{j,1}, \hat{G}_{j,2}, \dots \hat{G}_{j,k-1}])$ to estimate $\Span(G_{j,k})$. The first iteration uses $P \leftarrow [.]$, i.e. it computes standard PCA to estimate $\Span(G_{j,1})$.
By modifying the approach used in \cite{rrpcp_perf} for analyzing the addition step, we can show that since $\tilde{g}_{j,k}$ and $\tilde{h}_{j,k}$ are small enough (by Assumption \ref{assu1}), $\Span(G_{j,k})$ will be accurately recovered, i.e. $\|(I - \sum_{i=1}^{k} \hat{G}_{j,i} \hat{G}_{j,i}')G_{j,k}\|_2 \le \tilde{c}_{j,k} \zeta$. We do this $\vartheta_j$ times and finally we set $\Phat_j \leftarrow [\hat{G}_{j,1}, \hat{G}_{j,2} \dots \hat{G}_{j,\vartheta_j}]$ and $\Phat_{(t)} \leftarrow \Phat_j$. All of this is done at $t=\tilde{t}_j + \vartheta_j \tilde{\alpha}-1$.
Thus, at this time, $\SE_{(t)} = \|(I - \Phat_j \Phat_j') P_j\|_2 \le \sum_{k=1}^{\vartheta_j}  \|(I - \sum_{i=1}^{k} \hat{G}_{j,i} \hat{G}_{j,i}') G_{j,k} \|_2 \le  \sum_{k=1}^{\vartheta_j} \tilde{c}_{j,k} \zeta =  r_j \zeta \le r \zeta$.
Under the assumption that $t_{j+1} - t_j \ge K \alpha + \vartheta_{\max} \tilde{\alpha}$, this means that before the next subspace change time, $t_{j+1}$, $\SE_{(t)}$ is below $r \zeta$. 

We illustrate the ideas of subspace estimation by addition proj-PCA and cluster-PCA in Fig. \ref{add_del_proj_pca_diag2}. We discuss the connection between proj-PCA done in the addition step and the cluster-PCA (for deletion) step in Table \ref{tab_diff} given in Sec \ref{connect}.


\begin{figure}[t!]
\centerline{
\epsfig{file = 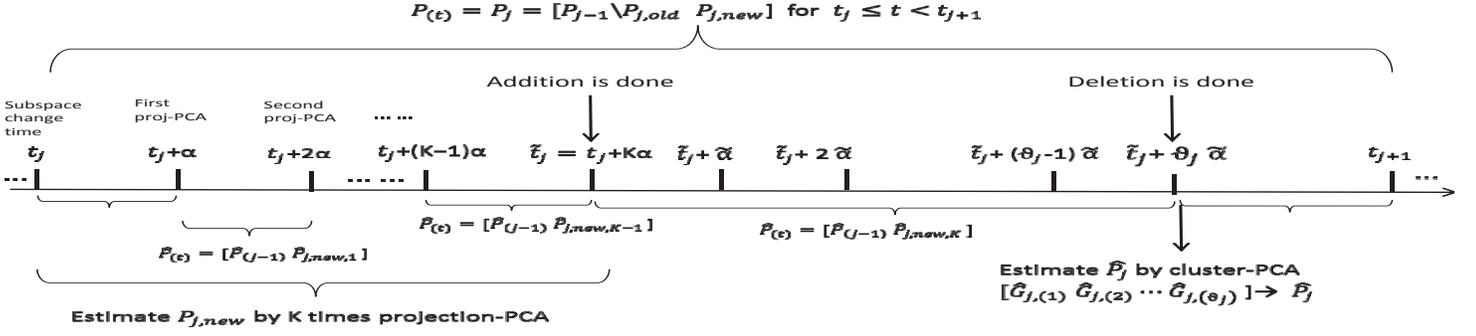, width =20cm, height = 5cm}
}
\caption{\small{A diagram illustrating subspace estimation by ReProCS-cPCA
}}
\label{add_del_proj_pca_diag2}
\end{figure}

\subsection{Practical Parameter Settings} \label{prac_algo}
The ReProCS-cPCA algorithm has parameters $\xi$, $\omega$, $\alpha$, $\tilde{\alpha}$, $K$ and it uses knowledge of model parameters $t_j$, $r_0$, $c_{j,\new}$, $\vartheta_j$ and $\tilde{c}_{j,i}$. If the model is known the algorithm parameters can be set as in Theorem \ref{thm2}. In practice, typically the model is unknown. In this case, the parameters $t_j$, $r_0$, $c_{j,\new}$, $\xi$, $\omega$, $K$ can be set as explained in \cite{rrpcp_perf}.
The parameters $\vartheta_j$ and $\tilde{c}_{j,i}$ for $i=1,2 \dots \vartheta_j$, can be set by computing the eigenvalues of $\frac{1}{\tilde{\alpha}} \sum_{t \in \tilde{I}_{j,1}}  \hat{L}_t \hat{L}_t'$ and clustering them using any standard clustering algorithm, e.g. k-means clustering or split-and-merge\footnote{One simple split-and-merge approach is as follows. Start with a single cluster. Split into two clusters: select the split so that $\tilde{g}_{\max}$ is minimized. Split each of these clusters into two parts again while ensuring $\tilde{g}_{\max}$ is minimized. Keep doing this for $d_1$ steps. Notice that, with every splitting, $\tilde{g}_{\max}$ will either remain the same or reduce, however $\tilde{h}_{\max}$ will either remain same or increase. Then, do a set of merge steps: in each step find the pair of consecutive clusters to merge that will minimize $\tilde{h}_{\max}$.}. We pick $\alpha$ and $\tilde{\alpha}$ somewhat arbitrarily. A thumb rule is that $\alpha$ and $\tilde\alpha$ need to be at least five to ten times $c_{\max}$ and $\max_j \max_{i=1,2 \dots \vartheta_j} \tilde{c}_{j,i}$ respectively. From simulation experiments, the algorithm is not very sensitive to the specific choice.

\subsection{The need for Projection-PCA}\label{ppca_need}
The reason standard PCA cannot be used and we need proj-PCA is because $e_t= \Lhat_t - L_t$ is correlated with $L_t$. The discussion here also applies to recursive or online PCA which is just a fast algorithm for computing standard PCA. In most existing works that analyze finite sample PCA, e.g. see \cite{nadler} and references therein, the noise or error in the ``data" used for PCA (here $\Lhat_t$'s) is uncorrelated with the true values of the data (here $L_t$'s) and is zero mean. Thus, when computing the eigenvectors of $(1/\alpha)\sum_t \Lhat_t \Lhat_t'$, the dominant term of the perturbation, $(1/\alpha) \sum_t \Lhat_t \Lhat_t' - (1/\alpha)\sum_t L_t L_t'$, is  $(1/\alpha)\sum_t e_t e_t'$ (the terms $(1/\alpha)\sum_t L_t e_t'$ and its transpose are close to zero w.h.p. due to law of large numbers). By assuming that the error/noise $e_t$ is small enough, the perturbation can be made small enough.

However, for our problem, because $e_t$ and $L_t$ are correlated, the dominant terms in the perturbation seen by standard PCA will be $(1/\alpha)\sum_t L_t {e_t}'$ and its transpose. Since $L_t$ can have large magnitude, the bound on the perturbation will be large and this will create problems when applying the $\sin \theta$ theorem (Theorem \ref{sin_theta}) to bound the subspace error. On the other hand, when using proj-PCA, $L_t$ gets replaced by $(I - \Phat_{j-1}\Phat_{j-1}')L_t$ (in the addition step) or by $(I - \sum_{i=1}^k \hat{G}_i \hat{G}_i')L_t$ (in cluster-PCA) and this results in significantly smaller perturbation.
We have explained this point in detail in Appendix F of \cite{rrpcp_perf}.


\section{Performance Guarantees} \label{perf_g}

We state the main result first and then discuss it in the next subsection. We give its corollary for the case where $f$ is small in Sec \ref{f_small_sec}. The proof outline is given in Sec \ref{detailed} and the proof is given in Sec \ref{thmproof}.

\subsection{Main Result}

\begin{theorem} \label{thm2}
Consider Algorithm \ref{ReProCS_del}.  Let $c:= c_{\max}$ and $r:= r_0 + c$.
Assume that $L_t$ obeys the model given in Assumption \ref{model_lt}. Also, assume that the initial subspace estimate is accurate enough, i.e. $\|(I - \Phat_0 \Phat_0') P_0\| \le r_0 \zeta$, for a $\zeta$ that satisfies
$$\zeta  \leq  \min(\frac{10^{-4}}{(r+c)^2},\frac{1.5 \times 10^{-4}}{(r+c)^2 f},\frac{1}{(r+c)^{3}\gamma_*^2}) \ \text{where} \ f := \frac{\lambda^+}{\lambda^-}$$
Let $\xi_0(\zeta), \rho, K(\zeta), \alpha_{\add}(\zeta), \alpha_{\del}(\zeta)$, $g_{j,k}$ be as defined in Definition \ref{defn_alpha}.
If the following conditions hold:
\ben
\item {\em (algorithm parameters) }
$\xi = \xi_0(\zeta), \   7 \rho \xi \leq \omega \leq S_{\min} - 7 \rho \xi,  \ K = K(\zeta), \ \alpha \ge \alpha_{\text{add}}(\zeta), \ \tilde{\alpha} \ge \alpha_{\del}(\zeta)$,

\item  {\em (denseness) }
\bea
&& \max_j \kappa_{2s} (P_{j-1}) \leq \kappa_{2s,*}^+ = 0.3, \ \max_j  \kappa_{2s}(P_{j,\new}) \leq \kappa_{2s,\new}^+ = 0.15,  \nn\\
&&  \max_j  \max_{0 \leq k \leq K} \kappa_{2s}(D_{j,\new,k}) \le \kappa_s^+ = 0.15, \  \max_j  \max_{0 \leq k \leq K} \kappa_{2s}(Q_{j,\new,k}) \le \tilde{\kappa}_{2s}^+ = 0.15, \nn \\  
&& \max_j \kappa_s ((I-\hat{P}_{j-1} {\hat{P}_{j-1}}' - \hat{P}_{j,\new,K} {\hat{P}_{j,\new,K}}')P_j) \leq \kappa_{s,e}^+ \nn
\eea
where $D_{j,\new,k}:= (I - \Phat_{j-1} \Phat_{j-1}'-\Phat_{j,\new,k} \Phat_{j,\new,k}')P_{j,\new}$,
and $Q_{j,\new,k}: = (I-P_{j,\new}{P_{j,\new}}')\Phat_{j,\new,k}$ and $\Phat_{j,\new,0} = [.]$,

\item  {\em (slow subspace change) }  
\bea
&& \max_{j} (t_{j+1} -t_j) > K  \alpha + \vartheta_{\max} \tilde{\alpha}, \nn \\
&& \max_{j} \max_{t \in \mathcal{I}_{j,k}} \|a_{t,\new}\|_{\infty} \leq \gamma_{\new,k}:= \min (1.2^{k-1} \gamma_{\new}, \gamma_*), \ \text{for all} \ k=1,2, \dots K, \nn \\ %
&& 14 \rho \xi_0 (\zeta) \le S_{\min}, \nn
\eea

\item {\em (small average condition number of $Cov(a_{t,\new})$) } 
$g_{j,k} \le g^+:=\sqrt{2}$,

\item {\em (clustered eigenvalues) } Assumption \ref{assu1} holds with $\tilde{g}_{\max},\tilde{h}_{\max}, \tilde{c}_{\min}$ satisfying $f_{dec}(\tilde{g}_{\max},\tilde{h}_{\max}) - \frac{f_{inc}(\tilde{g}_{\max},\tilde{h}_{\max})}{\tilde{c}_{\min} \zeta} > 0$ where $f_{dec}(\tilde{g}_{\max},\tilde{h}_{\max})$ and $f_{inc}(\tilde{g}_{\max},\tilde{h}_{\max})$ are defined in Definition \ref{defzetap} (also see Remark \ref{f_inc_rem} which weakens this requirement),

\een

then, with probability at least $1 -  2 n^{-10}$, at all times, $t$, 
\ben
\item 
$
\That_t = T_t \  \text{and} \  \|e_t\|_2 = \|L_t - \hat{L}_t\|_2 = \|\hat{S}_t - S_t\|_2 \le 0.18 \sqrt{c} \gamma_{\new} + 1.24 \sqrt{\zeta}.
$ 

\item the subspace error, $\SE_{(t)}$ satisfies
\bea
\SE_{(t)} &\leq&  \left\{  \begin{array}{ll}
0.6^{k-1} + r \zeta + 0.4 c \zeta  &  \ \text{if}  \    t \in \mathcal{I}_{j,k}, \ k=1,2,\cdots,K\\  
(r+c) \zeta                        & \  \text{if} \   t \in \cup_{k=1}^{\vartheta_j} \tilde{\mathcal{I}}_{j,k}  \\
r \zeta                            & \  \text{if} \  t \in \tilde{\mathcal{I}}_{j,\vartheta_j+1}
\end{array} \right. \nn \\
 & \le  &   \left\{  \begin{array}{ll}
0.6^{k-1} + 10^{-2} \sqrt{\zeta}  & \  \text{if}  \  t \in \mathcal{I}_{j,k}, \ k=1,2,\cdots,K \nn \\
10^{-2} \sqrt{\zeta}   & \  \text{if}  \  t \in (\cup_{k=1}^{\vartheta_j} \tilde{\mathcal{I}}_{j,k}) \cup \tilde{\mathcal{I}}_{j,\vartheta_j+1}  
\end{array} \right.
\eea

\item the error $e_t = \hat{S}_t - S_t = L_t - \hat{L}_t$ satisfies the following at various times
\bea
\|e_t\|_2  & \le &  \left\{  \begin{array}{ll}
1.17 [ 0.15 \cdot 0.72^{k-1} \sqrt{c}\gamma_{\new} + 0.15 \cdot 0.4 c \zeta \sqrt{c} \gamma_* + r \zeta \sqrt{r} \gamma_*]   & \ \ \text{if}  \ \   t \in \mathcal{I}_{j,k}, \ k=1,2,\cdots,K \\
1.17(r+c) \zeta \sqrt{r} \gamma_* &  \ \ \text{if} \ \ t \in \cup_{k=1}^{\vartheta_j} \tilde{\mathcal{I}}_{j,k} \\
1.17 r\zeta \sqrt{r} \gamma_* &  \ \ \text{if} \ \ t \in\tilde{\mathcal{I}}_{j,\vartheta_j+1} \\
\end{array} \right. \nn \\
 & \le &  \left\{  \begin{array}{ll}
0.18 \cdot 0.72^{k-1} \sqrt{c}\gamma_{\new} + 1.17 \cdot 1.06 \sqrt{\zeta}  & \ \ \text{if} \ \   t \in \mathcal{I}_{j,k}, \ k=1,2,\cdots,K  \nn \\ 
1.17 \sqrt{\zeta}  & \ \ \text{if} \ \ t \in  (\cup_{k=1}^{\vartheta_j} \tilde{\mathcal{I}}_{j,k}) \cup \tilde{\mathcal{I}}_{j,\vartheta_j+1}
\end{array} \right.
\eea

\een
\end{theorem}

The above result says the following. Assume that the initial subspace error is small enough. If the assumptions given in the theorem hold, then, w.h.p., we will get exact support recovery at all times. Moreover, the sparse recovery error (and the error in recovering $L_t$) will always be bounded by $0.18\sqrt{c} \gamma_\new$ plus a constant times $\sqrt{\zeta}$. Since $\zeta$ is very small, $\gamma_\new \ll S_{\min}$, and $c$ is also small, the normalized reconstruction error for $S_t$ will be small at all times, thus making this a meaningful result. In the second conclusion, we bound the subspace estimation error, $\SE_{(t)}$. When a subspace change occurs, this error is initially bounded by one. The above result shows that, w.h.p., with each adddition proj-PCA step, this error decays roughly exponentially and falls below $(r+c)\zeta$ within $K$ steps. After the cluster-PCA step, this error falls below $r\zeta$. By assumption, this occurs before the next subspace change time. Because of the choice of $\zeta$, both $(r+c)\zeta$ and $r \zeta$ are below $0.01 \sqrt{\zeta}$. The third conclusion shows that the sparse recovery error as well as the error in recovering $L_t$ decay in a similar fashion. 

\subsection{Discussion}\label{discuss}
Notice from Definition \ref{defn_alpha} that $K = K(\zeta)$ is larger if $\zeta$ is smaller. Also, both $\alpha_\add(\zeta)$ and $\alpha_\del(\zeta)$ are inversely proportional to $\zeta$. Thus, if we want to achieve a smaller lowest error level, $\zeta$, we need to compute both addition proj-PCA and cluster-PCA's over larger durations, $\alpha$ and $\tilde\alpha$ respectively, and we will need more number of addition proj-PCA steps $K$. Because of slow subspace change, this means that we also require a larger delay between subspace change times, i.e. larger $t_{j+1}-t_j$.

\subsubsection{Comparison with ReProCS}
The ReProCS algorithm of \cite{rrpcp_perf} is Algorithm \ref{ReProCS_del} with step 3c removed and replaced by $\Phat_j \leftarrow [\Phat_{j-1}, \Phat_{j,\new,K}]$.
Let us compare the above result with that for ReProCS for the subspace change model of Assumption \ref{model_lt} \cite[Corollary 43]{rrpcp_perf}.
First, ReProCS requires $\kappa_{2s}([P_0, P_{1,\new}, \dots P_{J,\new}]) \le 0.3$ whereas ReProCS-cPCA only requires $\max_j \kappa_{2s}(P_j) \le 0.3$. Moreover, ReProCS requires $\zeta$ to satisfy $\zeta  \leq  \min(\frac{10^{-4}}{(r_0+(J-1)c)^2},\frac{1.5 \times 10^{-4}}{(r_0+(J-1)c)^2 f},\frac{1}{(r_0+(J-1)c)^{3}\gamma_*^2})$ whereas in case of ReProCS-cPCA the denominators in the bound on $\zeta$ only contain $r + c = r_0 + 2c$ (instead of $r_0+(J-1)c$).

Because of the above, in Theorem \ref{thm2} for ReProCS-cPCA, the only place where $J$ (the number of subspace change times) appears is in the definitions of $\alpha_\add$ and $\alpha_\del$. Notice that $\alpha_\add$ and $\alpha_\del$ govern the delay  between subspace change times, $t_{j+1}-t_j$. Thus, with ReProCS-cPCA, $J$ can keep increasing, as long as $t_{j+1}-t_j$ also increases accordingly. Moreover, notice that the dependence of $\alpha_\add$ and $\alpha_\del$ on $J$ is only logarithmic and thus $t_{j+1}-t_j$ needs to only increase in proportion to $\log J$. 
On the other hand, for ReProCS (see \cite[Corollary 43]{rrpcp_perf}), $J$ appears in the denseness assumption, in the bound on $\zeta$ and in the definition of $\alpha_\add$. Thus, ReProCS needs a bound on $J$ that is indirectly imposed by the denseness assumption. 

The main extra assumptions that ReProCS-cPCA needs are (i) the clustering assumption (Assumption \ref{assu1} with $\tilde{h}_{\max}, \tilde{g}_{\max}$ being small enough to satisfying $f_{dec}(\tilde{g}_{\max},\tilde{h}_{\max}) - \frac{f_{inc}(\tilde{g}_{\max},\tilde{h}_{\max})}{\tilde{c}_{\min} \zeta} > 0$; and (ii) $\max_j \kappa_{s}((I-\hat{P}_{j-1} {\hat{P}_{j-1}}' - \hat{P}_{j,\new,K} {\hat{P}_{j,\new,K}}')P_j) < \kappa_{s,e}^+$. The second assumption is similar to the denseness assumption on $D_{j,\new,k}$ which is required by both ReProCS and ReProCS-cPCA. This is discussed in \cite{rrpcp_perf}. The clustering assumption is a practically valid one. We verified it for a video of moving lake waters shown in \url{http://www.ece.iastate.edu/~chenlu/ReProCS/ReProCS.htm} as follows. We first ``low-rankified" it to 90\% energy as explained in \cite[Sec X-B]{rrpcp_perf}. Note that, with one sequence, it is not possible to estimate $\Lambda_t$ (this would require an ensemble of sequences) and thus it is not possible to check if all $\Lambda_t$'s in $[\tilde{t}_j, t_{j+1}-1]$ are similar enough. However, by assuming that $\Lambda_t$ is the same for a long enough sequence, one can estimate it using a time average and then verify if its eigenvalues are sufficiently clustered. When this was done, we observed that the clustering assumption holds with $\tilde{g}_{\max} = 7.2$ and $\tilde{h}_{\max} = 0.34$. 

\subsubsection{Comparison with PCP}


We provide a qualitative comparison with the PCP result of \cite{rpca}. A direct comparison is not possible since the proof techniques used are very different and since we solve a recursive version of the problem where as PCP solves a batch one. Moreover, PCP provides guarantees for exact recovery of ${\cal S}_t$ and ${\cal L}_t$. In our result, we obtain guarantees for exact support recovery of the $S_t$'s (and hence of ${\cal S}_t$) and bounded error recovery of its nonzero values and of ${\cal L}_t$. Also, the PCP algorithm assumes no model knowledge, whereas our algorithm does assume knowledge of model parameters. Of course, in Sec \ref{prac_algo}, we have explained how to set the parameters in practice when the model is not known. 

Consider the denseness assumptions. Let ${\cal L}_t = U \Sigma V'$ be its SVD. Then, for $t \in [t_{j}, t_{j+1}-1]$, $U = [P_{0}, P_{1,\new}, P_{2,\new}, \dots P_{j,\new}]$ and $V = [a_1, a_2 \dots a_{t}]' \Sigma^{-1}$. The result for PCP \cite{rpca} assumes denseness of $U$ and of $V$: it requires $\kappa_1(U) \le \sqrt{\mu r / n}$ and  $\kappa_1(V) \le \sqrt{\mu r / n}$ for a constant $\mu \ge 1$. Moreover, it also requires $\|UV'\|_{\max} \le \sqrt{\mu r}/ n$. On the other hand, ReProCS-cPCA only requires $\kappa_{2s}(P_{j}) \le 0.3$ and  $\kappa_{2s}(P_{j,\new}) \le 0.15$. It does not need denseness of the entire $U$; it does not assume anything about denseness of $V$; and it does not need a bound on $\|UV'\|_{\max}$.

Another difference is that the result for PCP assumes that any element of the $n \times t$ matrix ${\cal S}_t$ is nonzero w.p. $\varrho$, and zero w.p. $1-\varrho$, independent of all others (in particular, this means that the support sets of the different $S_t$'s are independent over time). 
Our result for ReProCS-cPCA does not put any such assumption. However it does require denseness of the matrix $D_{j,\new,k}$ whose columns span the unestimated part of $\Span(P_{j,\new})$ for $t \in \mathcal{I}_{j,k+1}$. As demonstrated in Sec. \ref{sims}, this reduces ($\kappa_s(D_{j,\new,k})$ increases) if the support sets of $S_t$'s change very little over time. However, as long as, for most $k$,  $\kappa_s(D_{j,\new,k})$ is anything smaller than one, which happens as long as there is at least one support change during $\mathcal{I}_{j,k}$, the subspace error does decay down to a small enough value within a finite number of steps. The number of steps required for this increases as $\kappa_s(D_{j,\new,k})$ increases. Since $\kappa_s(D_{j,\new,k})$ cannot be computed in polynomial time, for the above discussion, we computed $\|{I_{T_t}}' D_{j,\new,k}\|_2/\|D_{j,\new,k}\|_2$ at $t=t_j+k\alpha-1$ for $k=0,1, \dots K$. In fact, our proof also only needs a bound on this latter quantity. 

Also, some additional assumptions that ReProCS-cPCA needs are (a) accurate knowledge of the initial subspace and slow subspace change; (b) denseness of $Q_{j,\new,k}$; (c) the independence of $a_t$'s over time; (d) condition number of the average covariance matrix of $a_{t,\new}$ is not too large; and (e) the clustering assumption. Assumptions (a), (b), (c) are discussed in detail in \cite{rrpcp_perf} and (a) is also verified for real data. 
As explained in \cite{rrpcp_perf}, (c) can possibly be replaced by a weaker random walk model assumption on $a_t$'s if we use the matrix Azuma inequality \cite{tail_bound} instead of matrix Hoeffding. Assumption (e) is discussed above.  (d) is also an assumption made for simplicity. It can be removed if a clustering assumption similar to Assumption \ref{assu1} holds for $(\Lambda_t)_\new = \text{Cov}(a_{t,\new})$ during $t \in [t_j, \tilde{t}_j-1]$ and we use an approach similar to cluster-PCA. If there are $\vartheta_{\new,j}$ clusters, we will need $\vartheta_{\new,j}$ proj-PCA steps to estimate $\Phat_{\new,k}$ (instead of the current one step). At the $l^{th}$ step, we use proj-PCA with $P$ being  $\Phat_{j-1}$ concatenated with the basis matrix estimates for the last $l-1$ clusters to recover the $l^{th}$ cluster.


\subsection{Special Case when $f$ is small} \label{f_small_sec}
If in a problem, $L_t$ has small magnitude for all times $t$, then $f$, which is the maximum condition number of $\text{Cov}(L_t)$ for any $t$, can be small. If this is the case, then the clustering assumption trivially holds with $\vartheta_j=1$, $\tilde{c}_{j,1} = r_j$, $\tilde{g}_{\max} =\tilde{g}_{j,1} = f$ and $\tilde{h}_{\max}={h}_{j,1} = 0$. Thus, $\vartheta_{\max} =1$. In this case, the following corollary holds.
\begin{corollary} \label{f_small}
Assume that the initial subspace estimate is accurate enough as given in Theorem \ref{thm2} with $\zeta$ as chosen there. Also assume that the first four conditions of Theorem \ref{thm2} hold. Then, if $f$ is small enough so that $f_{inc}(f,0) \le f_{dec}(f,0) \tilde{c}_{\min} \zeta$, then, all conclusions of Theorem \ref{thm2} hold.
\end{corollary}
Notice that the above corollary does not need Assumption \ref{assu1} to hold.

\section{Definitions, Proof Outline and Connection between addition and deletion steps} \label{detailed}
In Sec \ref{defs}, we define all the quantities that are needed for the proof. The proof outline is given in Sec \ref{outline}. We discuss how the proof strategy for the cluster-PCA (for deletion) step is related to that of addition proj-PCA in Sec \ref{connect}.

\subsection{Definitions} \label{defs}
Certain quantities are defined earlier in Assumptions \ref{model_lt} and \ref{assu1}, in Definitions \ref{defn_at} and \ref{defn_intervals}, in Algorithm \ref{ReProCS_del} and in Theorem \ref{thm2}.

\begin{definition}
In the sequel, we let
\ben
\item $c:= c_{\max}$ and $r := r_{\max} = r_0 + c$ and so $r_j = r_0 + \sum_{i=1}^{j} (c_{i,\new} - c_{i,\old}) \leq r$, 

\item $\phi^+ := 1.1735$ 

\een
\end{definition}

\begin{definition}\label{defn_alpha}
We define here the parameters used in Theorem \ref{thm2}.
\ben
\item Define $K(\zeta) := \left\lceil\frac{\log(0.6c\zeta)}{\log {0.6}} \right\rceil$
\item Define $\xi_0(\zeta)  :=   \sqrt{c} \gamma_{\new} + 1.06 \sqrt{\zeta}$ 
\item Define  $\rho  :=  \max_{t}\{\kappa_1 (\hat{S}_{t,\text{cs}} - S_t)\}$. Notice that $\rho \le 1$.
\item Define the condition number of the average of $\text{Cov}(a_{t,\new})$ over $t \in \mathcal{I}_{j,k}$ as
\bea
g_{j,k} \sdefn \frac{{\lambda_{j,\new,k}}^+}{{\lambda_{j,\new,k}}^-} \ \text{where} \nn \\
{\lambda_{j,\new,k}}^+ \sdefn \lambda_{\max}( \frac{1}{\alpha} \sum_{t \in \mathcal{I}_{j,k}} (\Lambda_t)_{\new} ),  \ \
{\lambda_{j,\new,k}}^- := \lambda_{\min}( \frac{1}{\alpha} \sum_{t \in \mathcal{I}_{j,k}} (\Lambda_t)_{\new} ), \nn
\eea

\item Let $K = K(\zeta)$. We define $\alpha_\add(\zeta)$ as in \cite{rrpcp_perf} the smallest value of $\alpha$ so that $(p_K(\alpha, \zeta))^{KJ} \ge 1-  n^{-10}$, where $p_K(\alpha,\zeta)$ is defined in \cite[Lemma 35]{rrpcp_perf}. An explicit value for it \cite{rrpcp_perf} is 
$$ \alpha_\add(\zeta)  = \lceil (\log 6 K J + 11 \log n) \frac{8 \cdot 24^2} {(\zeta \lambda^-)^2}
\max(\min(1.2^{4K} \gamma_{\new}^4, \gamma_*^4), \frac{16}{c^2}, 4(0.186 \gamma_\new^2 + 0.0034 \gamma_\new + 2.3)^2 )
\rceil$$
In words, $\alpha_{\text{add}}$ is the smallest value of the number of data points, $\alpha$, needed for an addition proj-PCA step to ensure that Theorem \ref{thm2} holds w.p. at least $(1 -   2n^{-10})$.

\item We define $\alpha_{\del}(\zeta)$ as the smallest value of $\alpha$ so that $\tilde{p}(\tilde{\alpha}, \zeta)^{\vartheta_{\max} J} \ge 1-  n^{-10}$ where $\tilde{p}(\tilde{\alpha},\zeta)$ is defined in Lemma \ref{tilde_zeta}.  We can compute an explicit value for it by using the fact that for any $x \le 1$ and $r \ge 1$, $(1-x)^r \ge 1-rx$ and that $\sum_{i=1}^6 e^{-\frac{\alpha}{d_i^2}} \le 6 e^{-\frac{\alpha}{\max_{i=1,2\dots 6} d_i^2}}$. We get
\bea
\alpha_\del(\zeta) : = \lceil (\log 6 \vartheta_{\max} J + 11 \log n) \frac{8 \cdot 10^2}{( \zeta \lambda^-)^2} \max( 4.2^2, 4 b_7^2 ) \rceil \nn 
\eea
where $b_7 := (\sqrt{r} \gamma_* + \phi^+ \sqrt{\zeta})^2$ and $\phi^+=1.1732$.
In words, $\alpha_{\text{del}}$ is the smallest value of the number of data points, $\tilde\alpha$, needed for a deletion proj-PCA step to ensure that Theorem \ref{thm2} holds w.p. at least $(1 -   2n^{-10})$.
\een
\end{definition}

\begin{definition}
\label{defzetap}
Define the following.
\ben
\item  $\zeta_{*}^+ := r \zeta$

\item define the series $\{{\zeta_{k}}^+\}_{k=0,1,2,\cdots K}$ as follows
\bea
\zeta_0^+ := 1,  \ \zeta_k^+ :=\frac{b + 0.125 c \zeta}{1 - (\zeta_*^+)^2 - (\zeta_*^+)^2 f - 0.25 c \zeta - b}, \ \text{for} \ k \geq 1,
\eea
where $b :=  C \kappa_s^+ g^+  \zeta_{k-1}^+ + \tilde{C} (\kappa_s^+)^2 g^+ (\zeta_{k-1}^+)^2 + C' f (\zeta_*^+)^2$, $\kappa_s^+ : =0.15$, $C := (\frac{2\kappa_s^+ \phi^+}{\sqrt{1-(\zeta_*^+)^2}} + \phi^+ )$,  $C' := ((\phi^+)^2 +  \frac{2\phi^+ }{\sqrt{1-(\zeta_*^+)^2}}
+ 1 + \phi^+ + \frac{\kappa_s^+ \phi^+}{\sqrt{1-(\zeta_*^+)^2}} + \frac{\kappa_s^+(\phi^+)^2 }{\sqrt{1-(\zeta_*^+)^2}} )$, $\tilde{C} := ( (\phi^+)^2  +  \frac{\kappa_s^+ (\phi^+)^2 }{\sqrt{1-(\zeta_*^+)^2}})$.

\item define the series $\{{\tilde\zeta_{k}}^+\}_{k=1,2,\cdots,\vartheta_j}$ as follows
\bea
{\tilde\zeta_{k}}^+: = \frac{f_{inc}(\tilde{g}_k,\tilde{h}_k)}{f_{dec}(\tilde{g}_k,\tilde{h}_k)} \nn
\eea
where  $f_{inc}(\tilde{g},\tilde{h}): =  (r+c) \zeta[ 3\kappa_{s,e}^+  \phi^+ \tilde{g} +  [\kappa_{s,e}^+ \phi^+ + \kappa_{s,e}^+ (1+ 2\phi^+)\frac{r^2\zeta^2}{\sqrt{1-r^2\zeta^2}}] \tilde{h} +  [\frac{r^2}{r+c}\zeta + 4r \zeta \kappa_{s,e}^+ \phi^+ + 2(r+c) \zeta(1+ {\kappa_{s,e}^+}^2) {\phi^+}^2] f+ 0.2 \frac{1}{r+c}]$,
and
$f_{dec}(\tilde{g},\tilde{h}):= 1- \tilde{h} - 0.2 \zeta - r^2 \zeta^2 f - r^2 \zeta^2 -  f_{inc}(\tilde{g},\tilde{h})$. Notice that $f_{inc}(\tilde{g},\tilde{h})$ is an increasing function of $\tilde{g},\tilde{h}$ and $f_{dec}(\tilde{g},\tilde{h})$ is a decreasing function of $\tilde{g},\tilde{h}$.
\een
As we will see, $\zeta_{*}^+$, $\zeta_{k}^+$, $\tilde\zeta_{k}^+$ are the high probability upper bounds on $\zeta_{j,*}$, $\zeta_{j,k}$, $\tilde\zeta_{j,k}$ (defined in Definition \ref{def_SEt}) under the assumptions of Theorem \ref{thm2}.
\end{definition}

\begin{definition}\label{defn_Phi}
For the addition step, define
\ben
\item $\Phi_{j,k} := I-\Phat_{j-1} {\Phat_{j-1}}' - \Phat_{j,\new,k} {\Phat_{j,\new,k}}'$ and $\Phi_{j,0} := I-\Phat_{j-1} {\Phat_{j-1}}'$.
\item $\phi_k := \max_j \max_{T:|T|\leq s}\|({(\Phi_{j,k})_T}'(\Phi_{j,k})_T)^{-1}\|_2$. It is easy to see that $\phi_k \le \frac{1}{1-\max_j \delta_s(\Phi_{j,k})}$.
\item $D_{j,\new,k} := \Phi_{j,k} P_{j,\new}$ and $D_{j,\new} := D_{j,\new,0} = \Phi_{j,0} P_{j,\new}$. 
\een
For the cluster-PCA step (for deletion), define
\ben
\item $\Psi_{j,k} : =  I - \sum_{i=0}^{k} \hat{G}_{j,i} \hat{G}_{j,i}'$. 

\item $G_{j,\text{det},k} := [G_{j,1} \cdots, G_{j,k-1}]$ and $\hat{G}_{j,\text{det},k} := [\hat{G}_{j,1} \cdots, \hat{G}_{j,k-1}]$. Notice that $\Psi_{j,k} =  I - \hat{G}_{j,\text{det},k+1}\hat{G}_{j,\text{det},k+1}'$.
\item $G_{j,\text{undet},k} := [G_{j,k+1} \cdots, G_{j,\vartheta_j}]$.

\item $D_{j,k} := \Psi_{j,k-1} G_{j,k}$,  $D_{j,\text{det},k} :=  \Psi_{j,k-1} G_{j,\text{det},k}$ and $D_{j,\text{undet},k} := \Psi_{j,k-1}G_{j,\text{undet},k} $. 

\een
Here, $G_{j,\text{det},k}$ contains the directions that are already detected before the $k^{th}$ step of cluster-PCA; $G_{j,k}$ contains the directions that are being detected in the current step; $G_{j,\text{undet},k}$ contains the as yet undetected directions.
\end{definition}

\begin{definition}
Let $\kappa_{s,*} := \max_j \kappa_{s}(P_{j-1})$, $\kappa_{s,\new} := \max_j \kappa_{s}(P_{j,\new})$, $\kappa_{s,k}:= \max_j \kappa_s (D_{j,\new,k})$, $\tilde{\kappa}_{s,k} := \max_j \kappa_{s}((I-P_{j,\new}{P_{j,\new}}') \Phat_{j,\new,k})$, $\kappa_{s,e} := \max_j \kappa_s(\Phi_K P_{j})$.
\end{definition}

\begin{definition}\label{defHk_del} \
\ben
\item Let $D_{j,k} \overset{QR}{=} E_{j,k} R_{j,k}$ denote its QR decomposition. Here, $E_{j,k}$ is a basis matrix while $R_{j,k}$ is upper triangular. \footnote{Notice that $0< \sqrt{1-r^2 \zeta^2} \leq \sigma_i(R_{j,k})$ by Lemma \ref{bound_R}, therefore, $R_{j,k}$ is invertible.}

\item Let $E_{j,k,\perp}$ be a basis matrix for the orthogonal complement of $\Span(E_{j,k}) = \Span(D_{j,k})$. To be precise, $E_{j,k,\perp}$ is a $n \times (n-\tilde{c}_{j,k})$ basis matrix that satisfies ${E_{j,k,\perp}}' E_{j,k} = 0$.

\item Using $E_{j,k}$ and $E_{j,k,\perp}$, define $\tilde{A}_{j,k}$, $\tilde{A}_{j,k,\perp}$, $\tilde{H}_{j,k}$, $\tilde{H}_{j,k,\perp}$ and $\tilde{B}_{j,k}$ as 
    \bea
    \tilde{A}_{j,k} &:=& \frac{1}{\tilde{\alpha}} \sum_{t \in \tilde{I}_{j,k}} {E_{j,k}}' \Psi_{j,k-1} L_t{L_t}' \Psi_{j,k-1} E_{j,k} \nn \\
    \tilde{A}_{j,k,\perp} &:=& \frac{1}{\tilde{\alpha}} \sum_{t \in \tilde{I}_{j,k}} {E_{j,k,\perp}}' \Psi_{j,k-1} L_t{L_t}' \Psi_{j,k-1} E_{j,k,\perp} \nn \\
    \tilde{H}_{j,k} &:=& \frac{1}{\tilde{\alpha}} \sum_{t \in \tilde{I}_{j,k}} {E_{j,k}}'\Psi_{j,k-1} (e_t{e_t}' - L_t{e_t}' - e_t {L_t}') \Psi_{j,k-1} E_{j,k} \nn \\
    \tilde{H}_{j,k,\perp} &:=& \frac{1}{\tilde{\alpha}} \sum_{t \in \tilde{I}_{j,k}} {E_{j,k,\perp}}' \Psi_{j,k-1} (e_t{e_t}' - L_t{e_t}' - e_t {L_t}') \Psi_{j,k-1} E_{j,k,\perp} \nn \\
    \tilde{B}_{j,k} &:=& \frac{1}{\tilde{\alpha}} \sum_{t \in \tilde{I}_{j,k}} {E_{j,k,\perp}}' \Psi_{j,k-1} \hat{L}_t{\hat{L}_t}' \Psi_{j,k-1} E_{j,k} = \frac{1}{\tilde{\alpha}} \sum_{t \in \tilde{I}_{j,k}} {E_{j,k,\perp}}'\Psi_{j,k-1} (L_t-e_t)({L_t}' - {e_t}') \Psi_{j,k-1} E_{j,k} \nn
    \eea
\item Define
\bea
&&\tilde{\mathcal{A}}_{j,k} := \left[ \begin{array}{cc} E_{j,k} & E_{j,k,\perp} \\ \end{array} \right]
\left[\begin{array}{cc}\tilde{A}_{j,k} \ & 0 \ \\ 0 \ & \tilde{A}_{j,k,\perp}  \\ \end{array} \right]
\left[ \begin{array}{c} {E_{j,k}}' \\ {E_{j,k,\perp}}' \\ \end{array} \right]\nn\\
&&\tilde{\mathcal{H}}_{j,k} := \left[ \begin{array}{cc} E_{j,k} & E_{j,k,\perp} \\ \end{array} \right]
\left[\begin{array}{cc} \tilde{H}_{j,k} \ & {\tilde{B}_{j,k}}' \ \\ \tilde{B}_{j,k} \ &  \tilde{H}_{j,k,\perp} \\ \end{array} \right]
\left[ \begin{array}{c} {E_{j,k}}' \\ {E_{j,k,\perp}}' \\ \end{array} \right]
\label{defn_tilde_Hk}
\eea

\item From the above, it is easy to see that $$\tilde{\mathcal{A}}_{j,k} + \tilde{\mathcal{H}}_{j,k} =\frac{1}{\tilde{\alpha}} \sum_{t \in \tilde{\mathcal{I}}_{j,k}} \Psi_{j,k-1} \hat{L}_t {\hat{L}_t}' \Psi_{j,k-1}.$$

\item Recall from Algorithm \ref{ReProCS_del} that
$$\tilde{\mathcal{A}}_{j,k} + \tilde{\mathcal{H}}_{j,k} = \frac{1}{\tilde{\alpha}} \sum_{t \in \tilde{\mathcal{I}}_{j,k}} \Psi_{j,k-1} \hat{L}_t {\hat{L}_t}' \Psi_{j,k-1} \overset{EVD}{=} \left[ \begin{array}{cc} \hat{G}_{j,k} & \hat{G}_{j,k,\perp} \\ \end{array} \right]
\left[\begin{array}{cc} \Lambda_{j,k} \ & 0 \ \\ 0 \ & \ \Lambda_{j,k,\perp} \\ \end{array} \right]
\left[ \begin{array}{c} \hat{G}_{j,k}' \\ \hat{G}_{j,k,\perp}' \\ \end{array} \right] $$
is the EVD of $\tilde{\mathcal{A}}_{j,k} + \tilde{\mathcal{H}}_{j,k}$. 
Here $\Lambda_k$ is a $\tilde{c}_{j,k} \times \tilde{c}_{j,k}$ diagonal matrix.


\een
\end{definition}

\begin{definition}
Let $\Phat_{j,*}: = \Phat_{j-1}= \Phat_{(t_j-1)}$. Recall that $P_{j,*}: = P_{(t_j-1)}= P_{j-1} $. In the sequel, we use the subscript $*$ to denote the quantity at $t=t_j-1$.
\end{definition}

\begin{definition}[Subspace estimation errors]\label{def_SEt}\  
\ben
\item Recall that the subspace error at time $t$ is $\SE_{(t)} := \|(I - \Phat_{(t)} \Phat_{(t)}') P_{(t)} \|_2$.

\item Define
$$\zeta_{j,*} :=\|(I - \hat{P}_{j,*} {\hat{P}_{j,*}}')P_{j,*}\|_2. $$
This is the subspace error at $t=t_j-1$, i.e. $\zeta_{j,*} = \SE_{(t_j-1)}$. 

\item For $k=0,1,2,\cdots,K$, define
$$\zeta_{j,k} := \|(I - \Phat_{j-1} \Phat_{j-1}' - \Phat_{j,\new,k} \Phat_{j,\new,k}') P_{j,\new}\|_2.$$
This is the error in estimating $\text{span}(P_{j,\new})$ after the $k^{th}$ iteration of the addition step.
\item For $k=1,2,\cdots,\vartheta_j$, define
$$\tilde{\zeta}_{j,k} : = \|(I - \sum_{i=1}^{k} \hat{G}_{j,i} \hat{G}_{j,i}')G_{j,k}\|_2. $$
This is the error in estimating $\Span(G_{j,k})$ after the $k^{th}$ iteration of the cluster-PCA step. 
\een
\end{definition}

\begin{remark}[Notational issue]
Notice that $\zeta$ is a given scalar satisfying the bound given in Theorem \ref{thm2}, while $\zeta_{j,k}, \zeta_{j,*}$ and $\tilde{\zeta}_{j,k}$ are as defined above. Since the basis matrix estimates are functions of the $\Lhat_t$'s, which in turn are depend on the $L_t$'s and $L_t = P_{(t)}a_t$, thus, $\zeta_{j,k}, \zeta_{j,*}$ and $\tilde{\zeta}_{j,k}$ are functions of the $a_t$'s. Thus, $\zeta_{j,k}, \zeta_{j,*}$ and $\tilde{\zeta}_{j,k}$ are, in fact, random variables.
\end{remark}

\begin{remark}\label{SE_rem} \  
\ben

\item Notice that $\zeta_{j,0} = \|D_{j,\new}\|_2$, $\zeta_{j,k} = \|D_{j,\new,k}\|_2$ and $\tilde{\zeta}_{j,k} = \|(I - \hat{G}_k \hat{G}_k') D_{j,k}\|_2 = \|\Psi_{j,k}G_{j,k}\|_2$.

\item Notice from the algorithm that (i) $\Phat_{j,\new,k}$ is perpendicular to $\Phat_{j,*}=\Phat_{j-1}$; and (ii) $\hat{G}_{j,k}$ is perpendicular to $[\hat{G}_{j,1},\hat{G}_{j,2},\dots \hat{G}_{j,k-1}]$.

\item For $t\in \mathcal{I}_{j,k}$, $P_{(t)} = P_j = [(P_{j-1} \setminus P_{j,\old}), \ P_{j,\new}]$, $\hat{P}_{(t)} = [\hat{P}_{j-1} \ \hat{P}_{j,\new,k}]$ and
$$SE_{(t)} = \|(I - \hat{P}_{j-1} {\hat{P}_{j-1}}' - \hat{P}_{j,\new,k}{\hat{P}_{j,\new,k}}')P_j\|_2 \leq  \|(I - \hat{P}_{j-1} {\hat{P}_{j-1}}' - \hat{P}_{j,\new,k}{\hat{P}_{j,\new,k}}')[ P_{j-1} \ P_{j,\new}]\|_2 \leq \zeta_{j,*} + \zeta_{j,k}$$
for $k=1,2 \dots K$. The last inequality uses the first item of this remark.

\item For $t \in \tilde{\mathcal{I}}_{j,k}$, $P_{(t)} = P_j$, $\hat{P}_{(t)} = [\hat{P}_{j-1} \ \hat{P}_{j,\new,K}]$ and
$$SE_{(t)} = SE_{(t_j + K\alpha-1)} \leq \zeta_{j,*} + \zeta_{j,K}$$

\item For $t \in \tilde{\mathcal{I}}_{j,\vartheta_j+1}$, $P_{(t)} = P_j = [G_{j,1},\cdots,G_{j,\vartheta_j}]$, $\hat{P}_{(t)} = \hat{P}_j = [\hat{G}_{j,1},\cdots,\hat{G}_{j,\vartheta_j}]$, and
$$SE_{(t)} = \zeta_{j+1,*} \leq \sum_{k=1}^{\vartheta_j} \tilde{\zeta}_{j,k}$$
The last inequality uses the first item of this remark.

\een
\end{remark}

\begin{remark} \label{etdef_rem}
Recall that $e_t:= \Shat_t-S_t$. Notice from Algorithm \ref{ReProCS_del} that
\ben
\item $e_t = L_t - \Lhat_t$.
\item If $\That_t = T_t$, then $e_t = {I_{T_t}} [{(\Phi_{(t)})_{T_t}}'(\Phi_{(t)})_{T_t}]^{-1} {I_{T_t}}' \Phi_{(t)}P_{(t)} a_t$. This follows using the definition of $\Shat_t$ given in step 1d of the algorithm and the fact that $(\Phi_{(t)})_T' \Phi_{(t)} =(\Phi_{(t)} I_T)' \Phi_{(t)} = I_T' \Phi_{(t)}$ for any set $T$.
 Thus, for $t \in [t_j,t_{j+1}-1]$, 
\bea
e_t = {I_{T_t}} [{(\Phi_{(t)})_{T_t}}'(\Phi_{(t)})_{T_t}]^{-1} {I_{T_t}}' \Phi_{(t)}P_j a_t = {I_{T_t}} [{(\Phi_{(t)})_{T_t}}'(\Phi_{(t)})_{T_t}]^{-1} {I_{T_t}}' \Phi_{(t)}[P_{j,*} a_{t,*} + P_{j,\new} a_{t,\new}]
\label{etdef}
\eea
with
\bea
\Phi_{(t)} = \left\{
\begin{array}{ll}
\Phi_{j,k-1} \ & \ \ t \in \mathcal{I}_{j,k}, \ k=1,2 \dots K \\
\Phi_{j,K} \ & \ \ t \in \mathcal{\tilde{I}}_{j,k}, \ k=1,2 \dots \vartheta_j \\
\Phi_{j+1,0} \ & \ \ t \in \mathcal{\tilde{I}}_{j,\vartheta_j+1}
\end{array}
\right. \nn
\eea
\een
\end{remark}

\begin{definition}
Define the random variable
$$X_{j,k_1,k_2} := \{ a_1,a_2,\cdots, a_{t_j + k_1 \alpha + k_2 \tilde{\alpha}-1}\}.$$
Recall that $a_t$'s are mutually independent over $t$.
\end{definition}

\begin{definition}
Define the set $\check\Gamma_{j,k_1,k_2}$ as follows.
\bea
\check\Gamma_{j,k,0}:= &&
\{ X_{j,k,0} : \zeta_{j,k} \leq \zeta_{k}^+,  \ \text{and} \ \hat{T}_t = T_t \ \text{and} \  e_t \ \text{satisfies} \ (\ref{etdef}) \ \text{for all} \ t \in \mathcal{I}_{j,k}\}, \ k=1,2,\dots K, \ j=1,2,3,\dots J \nn \\  %
\check\Gamma_{j,K,k}:= &&
\{ X_{j,K,k} :\tilde{\zeta}_{j,k} \leq \tilde{c}_{j,k}\zeta,    \ \text{and} \  \hat{T}_t = T_t \ \text{and} \  e_t \ \text{satisfies} \ (\ref{etdef}) \ \text{for all} \ t \in \mathcal{\tilde{I}}_{j,k}\}, \  k=1,2,\dots \vartheta_j, \ j=1,2,3,\dots J \nn \\
\check\Gamma_{j,K, \vartheta_j+1}:= &&
\{X_{j+1,0,0}: \hat{T}_t = T_t \ \text{and} \  e_t \ \text{satisfies} \ (\ref{etdef}) \ \text{for all} \  t \in \mathcal{\tilde{I}}_{j,\vartheta_j+1}\}, \ j=1,2,3,\dots J \nn
\eea
Define the set $\Gamma_{j,k_1,k_2}$ as follows.
\bea
\Gamma_{1,0,0}: =&&  \{X_{1,0,0}: \zeta_{1,*} \leq r \zeta, \ \text{and} \ \hat{T}_t = T_t \ \text{and} \  e_t \ \text{satisfies} \ (\ref{etdef}) \ \text{for all} \ t \in [t_{\text{train}},t_1-1] \}, \nn \\
\Gamma_{j,k,0}: = &&  \Gamma_{j,k-1,0} \cap \check\Gamma_{j,k,0}, \ k=1,2,\dots K, \ j=1,2,3,\dots J \nn \\
\Gamma_{j,K,k}: = && \Gamma_{j,K,k-1} \cap \check\Gamma_{j,K,k}, \ k=1,2,\dots \vartheta_j, \ j=1,2,3,\dots J \nn \\
\Gamma_{j+1,0,0}: =&&  \Gamma_{j,K,\vartheta_j} \cap \check\Gamma_{j,K,\vartheta_j+1}, \ j=1,2,3,\dots J \nn
\eea
Recall from the notation section that the event $\Gamma_{j,k_1,k_2}^e: = \{X_{j,k_1,k_2} \in \Gamma_{j,k_1,k_2} \}.$
\end{definition}

\begin{remark}
Notice that the subscript $j$ always appears as the first subscript, while $k$ is the last one. At many places in this paper, we remove the subscript $j$ for simplicity. Whenever there is only one subscript, it refers to the value of $k$, e.g., $\Phi_0$ refers to $\Phi_{j,0}$, $\Phat_{\new,k}$ refers to $\Phat_{j,\new,k}$ and so on.
\label{remove_j}
\end{remark}

\begin{table}
\caption{Comparing and contrasting the addition proj-PCA step and proj-PCA used in the deletion step (cluster-PCA)}
\begin{center}
\begin{tabular}{|l||l|}
  \hline  
 {\bf $k^\text{th}$ iteration of addition proj-PCA} & {\bf $k^\text{th}$ iteration of cluster-PCA in the deletion step}  \\ \hline
  done at $t= t_j+k \alpha-1$              &  done at $t=t_j + K \alpha + \vartheta_j \tilde\alpha-1$ \\ \hline
  goal: keep improving estimates of $\Span(P_{j,\new})$ &   goal: re-estimate $\Span(P_{j})$ and thus ``delete" $\Span(P_{j,\old})$ \\ \hline
  compute $\Phat_{j,\new,k}$ by proj-PCA on $[\hat{L}_t: t\in \mathcal{I}_{j,k}]$    &  compute $\hat{G}_{j,k}$ by proj-PCA on $[\hat{L}_t: t\in \tilde{\mathcal{I}}_{j,k}]$  \\
    with $P = \hat{P}_{j-1}$                                                         & with $P = \hat{G}_{j,\text{det},k} = [\hat{G}_{j,1}, \cdots, \hat{G}_{j,k-1}]$   \\ \hline
   start with $\|(I - \Phat_{j-1} {\Phat_{j-1}}')P_{j-1}\|_2 \leq r\zeta$  and $\zeta_{j,k-1} \leq \zeta_{k-1}^+ \le 0.6^{k-1} + 0.4 c \zeta $       & start with $\|(I - \hat{G}_{j,\text{det},k}{\hat{G}_{j,\text{det},k}}')G_{j,\text{det},k}\|_2 \leq r\zeta$ and $\zeta_{j,K} \leq c \zeta$     \\ \hline
%
   need small $g_{j,k}$ which is the                                                   & need small $\tilde{g}_{j,k}$ which is the \\
average of the condition number of $\text{Cov}(P_{j,\new}'L_t)$ averaged over $t\in\mathcal{I}_{j,k}$ &  maximum of the condition number of $\text{Cov}(G_{j,k}'L_t)$ over $t\in \tilde{\mathcal{I}}_{j,k}$  \\ \hline
no undetected subspace & extra issue: ensure perturbation due to $\Span(G_{j,\text{undet},k})$ is small; \\
                      & need small $\tilde{h}_{j,k}$ to ensure the above   \\ \hline
  $\zeta_{j,k}$ is the subspace error in estimating $\text{span}(P_{j,\new})$ after the $k^{th}$ step & $\tilde{\zeta}_{j,k}$ is the subspace error in estimating $\text{span}(G_{j,k})$  after the $k^{th}$ step  \\ \hline
%
end with $\zeta_{j,k} \leq  \zeta_k^+ \leq 0.6^k + 0.4 c\zeta$ w.h.p. & end with $\tilde{\zeta}_{j,k} \leq \tilde{c}_{j,k} \zeta$ w.h.p. \\   \hline
  stop when $k=K$ with $K$ chosen so that $\zeta_{j,K} \leq  c\zeta$  & stop when $k = \vartheta_j$ and $\tilde{\zeta}_{j,k} \leq \tilde{c}_{j,k}\zeta$ for all $k=1,2,\cdots,\vartheta_j$ \\ \hline
  after $K^{th}$ iteration: $\Phat_{(t)} \leftarrow [\Phat_{j-1} \ \Phat_{j,\new,K}]$ and $SE_{(t)} \leq (r+c)\zeta$ & after $\vartheta_j^{th}$ iteration: $\Phat_{(t)} \leftarrow  [\hat{G}_{j,1},\cdots, \hat{G}_{j,\vartheta_j}]$ and $SE_{(t)} \leq r\zeta$  \\ \hline
\end{tabular}
\end{center}
 \label{tab_diff}
\end{table}

\subsection{Proof Outline of Theorem \ref{thm2}} \label{outline}

The first part of the proof that analyzes the projected CS step and the addition step is essentially the same as that in \cite{rrpcp_perf}. The only difference is that, now, $\zeta_*^+ = r \zeta$ instead of $\zeta_*^+ = (r_0 + (j-1)c) \zeta$. In Lemma \ref{lem_add}, the final conclusions for this part are summarized: it shows that, for all $k=1,2, \dots K$, $\zeta_k^+$ decays roughly exponentially with $k$ and it bounds the probability of $\Gamma_{j,k,0}^e$ given $\Gamma_{j,k-1,0}^e$. The second part of the proof analyzes the projected CS step and the cluster-PCA step. The final conclusion for this part is summarized in Lemma \ref{lem_del}: it bounds the probability of $\Gamma_{j,K,k}^e$ given $\Gamma_{j,K,k-1}^e$. Theorem \ref{thm2} follows essentially by applying Lemmas \ref{lem_del} and \ref{lem_add} for each $j$ and $k$ and using Lemma \ref{subset_lem}.

Lemma \ref{lem_del}, in turn, follows by combining the results of Lemma \ref{lem_et} (which shows exact support recovery and bounds the sparse recovery error for $t \in \tilde{I}_{j,k}$ conditioned on $\Gamma_{j,K,k-1}^e$), and Lemma \ref{tilde_zeta} (which bounds the subspace recovery error at the $k^{th}$ step of cluster-PCA conditioned on $\Gamma_{j,K,k-1}^e$).
%
Lemma \ref{lem_et} uses the result of Lemma \ref{RIC_bnd} which bounds the RIC of $\Phi_k$ in terms of  $\zeta_*$, $\zeta_k$ and the denseness coefficients of $P_*$ and $P_{\new}$. Lemma \ref{tilde_zeta} is obtained as follows.  In Lemma \ref{bnd_tzetakp}, we  show that, under the theorem's assumptions, $\tilde\zeta_k^+ \le  \tilde{c}_{j,k} \zeta$. In Lemma \ref{defnPCA}, we bound  $\tilde\zeta_k$ in terms of $\lambda_{\min}(A_k)$, $\lambda_{\max}(A_{k,\perp})$ and $\|\mathcal{H}_k\|_2$ using Lemma \ref{sin_theta_weyl}. Next, in Lemma \ref{lem_bound_terms}, (i) we use Lemma \ref{lem_et} and the Hoeffding corollaries (Corollaries \ref{hoeffding_nonzero} and \ref{hoeffding_rec}) to bound each of these terms and (ii) then we use Lemma \ref{defnPCA} and these bounds to bound $\tilde\zeta_k$ by $\tilde\zeta_k^+$ with a certain probability conditioned on $\Gamma_{j,K,k-1}^e$. Finally, Lemma \ref{tilde_zeta} follows by combining  Lemma \ref{bnd_tzetakp} and Lemma \ref{lem_bound_terms}.


\subsection{Connection with Addition proj-PCA} \label{connect}
Our strategy for analyzing cluster-PCA and hence for proving Theorem \ref{thm2} is a generalization of that used to analyze the $k^{th}$ addition proj-PCA step in \cite{rrpcp_perf}. We discuss this in Table \ref{tab_diff}.

\section{Proof of Theorem \ref{thm2}} \label{thmproof}
The theorem is a direct consequence of Lemmas \ref{lem_add} and \ref{lem_del} given below.


\subsection{Two Main Lemmas}

The lemma below is a slight modification of \cite[Lemma 40]{rrpcp_perf}. It summarizes the final conclusions of the addition step. 

\begin{lem}[Final lemma for addition step]\label{lem_add}
Assume that all the conditions in Theorem \ref{thm2} holds. Also assume that $\mathbf{P}(\Gamma_{j,k-1,0}^e ) > 0$.
Then
\ben
\item $\zeta_0^+=1$, $\zeta_k^+ \leq  0.6^{k}  + 0.4 c\zeta$ for all $k=1,2,\dots K$;
\item $\mathbf{P}(\Gamma_{j,k,0}^e \ | \Gamma_{j,k-1,0}^e ) \geq p_k(\alpha,\zeta) \ge p_K(\alpha,\zeta)$  for all $k=1,2,\dots K$.
\een
where $\zeta_k^+$ is defined in Definition \ref{defzetap} and $p_k(\alpha,\zeta)$ is defined in \cite[Lemma 35]{rrpcp_perf}.
\end{lem}
%


The proof of the above lemma follows using the exact same approach as in the proof of Lemma 40 of \cite{rrpcp_perf} but with $\zeta_*^+ = r\zeta$ instead of $(r_0 + (j-1)c_{\max}) \zeta$ everywhere. We give the proof outline in Appendix \ref{proof_lem_add}. 
%

The lemma below summarizes the final conclusions for the cluster-PCA step. It is proved using lemmas given in Sec \ref{keylems}.
\begin{lem}[Final lemma for cluster-PCA]\label{lem_del}
Assume that all the conditions in Theorem \ref{thm2} hold. Also assume that $\mathbf{P}(\Gamma_{j,K,k-1}^e ) > 0$. Then,
\ben
\item for all $k=1,2,\dots \vartheta_j$, $\mathbf{P} ( \Gamma_{j,K,k}^e \ | \ \Gamma_{j,K,k-1}^e) \geq \tilde{p}(\tilde{\alpha},\zeta)$ where $\tilde{p}(\tilde{\alpha},\zeta)$ is defined in Lemma \ref{tilde_zeta};
\item $\mathbf{P} ( \Gamma_{j+1,0,0}^e \ | \ \Gamma_{j,K,\vartheta_j}^e) = 1$.
\een
\end{lem}
\begin{proof}
Notice that $\mathbf{P} ( \Gamma_{j,K,k}^e \ | \ \Gamma_{j,K,k-1}^e) = \mathbf{P} (\tilde{\zeta}_k \leq \tilde{c}_k \zeta \ \text{and} \ \That_t = T_t, \ \text{and} \ e_t \ \text{satisfies} \ (\ref{etdef}) \ \text{for all} \ t \in \tilde{I}_{j,k} \ | \ \Gamma_{j,K,k-1}^e)$ and $\mathbf{P} ( \Gamma_{j+1,0,0}^e \ | \ \Gamma_{j,K,\vartheta_j}^e)  = \mathbf{P} ( \hat{T}_t = T_t \ \text{and} \  e_t \ \text{satisfies} \ (\ref{etdef}) \ \text{for all} \ t \in \mathcal{I}_{j,\vartheta_j+1} )$.
The first claim of the lemma follows by combining Lemma \ref{tilde_zeta} and the last claim of Lemma \ref{lem_et}, both given below in Sec \ref{keylems}. The second claim follows using the last claim of Lemma \ref{lem_et}.
\end{proof}

\begin{remark}\label{Gamma_rem}
Under the assumptions of Theorem \ref{thm2}, it is easy to see that the following holds.
\ben
\item For any $k=1,2 \dots K$, $\Gamma_{j,k,0}^e$ implies that (i) $\zeta_{j,*} \le \zeta_{*}^+:= r \zeta$ and (ii) $\zeta_{j,k'} \le 0.6^{k'} + 0.4 c \zeta$ for all $k'=1,2,\dots k$
\bi
\item (i) follows from the definition of $\Gamma_{j,k,0}^e$ and $\zeta_{j,*} \le \sum_{k=1}^{\vartheta_{j-1}} \tilde{\zeta}_{j-1,k'} \le \sum_{k=1}^{\vartheta_{j-1}} \tilde{c}_{j-1,k'} \zeta =  r_{j-1} \zeta \le r \zeta = \zeta_*^+$; and (ii) follows from the definition of $\Gamma_{j,k,0}^e$ and the first claim of Lemma \ref{lem_add}. 
\ei

\item For any $k=1,2 \dots \vartheta_j+1$, $\Gamma_{j,K,k}^e$ implies (i) $\zeta_{j,*} \le \zeta_{*}^+$, (ii) $\zeta_{j,k'} \le 0.6^{k'} + 0.4 c \zeta$ for all $k'=1,2,\dots K$,  (iii) $\zeta_{j,K} \le c \zeta$, (iv) $\|\Phi_{j,K} P_j\|_2 \le (r+c) \zeta$, (v) $\tilde\zeta_{j,k'} \le \tilde{c}_{j,k'} \zeta$ for $k'=1,2,\dots k$ and (vi) $\sum_{k'=1}^{k} \tilde\zeta_{j,k'} \le  r_j \zeta \le r \zeta$.
\bi
\item (i) and (ii) follow because $\Gamma_{j,K,0}^e \subseteq \Gamma_{j,K,k}^e$, (iii) follows from (ii) using the definition of $K$, (iv) follows from (i) and (iii) using $\|\Phi_{j,K} P_j\|_2 \le \|\Phi_{j,K} [P_{j,*}, P_{j,\new}]\|_2 \le \zeta_{j,*} + \zeta_{j,K}$, and (v) follows from the definition of $\Gamma_{j,K,k}^e$.
\ei

\item $\Gamma_{J+1,0,0}^e$ implies (i) $\zeta_{j,*} \le \zeta_{*}^+$ for all $j$, (ii) $\zeta_{j,k} \le 0.6^{k} + 0.4 c \zeta$ for all $k=1,\cdots,K$ and all $j$, (iii) $\zeta_{j,K} \le c \zeta$ for all $j$.
\een
\end{remark}

\subsection{Proof of Theorem \ref{thm2}}


The theorem is a direct consequence of Lemmas \ref{lem_add} and \ref{lem_del} and Lemma \ref{subset_lem}.

Notice that $\Gamma_{j,0,0}^e \supseteq  \Gamma_{j,1,0}^e \dots \supseteq   \Gamma_{j,K,0}^e \supseteq  \Gamma_{j,K,1}^e \supseteq  \Gamma_{j,K,2}^e \dots \supseteq  \Gamma_{j,K,\vartheta}^e \supseteq  \Gamma_{j+1,0,0}^e$. Thus, by Lemma \ref{subset_lem}, $\mathbf{P}(\Gamma_{j+1,0,0}^e | \Gamma_{j,0,0}^e) = \mathbf{P}(\Gamma_{j+1,0,0}^e | \Gamma_{j,K,\vartheta}^e) \prod_{k=1}^{\vartheta} \mathbf{P}(\Gamma_{j,K,k}^e | \Gamma_{j,K,k-1}^e) \prod_{k=1}^{K} \mathbf{P}(\Gamma_{j,k,0}^e | \Gamma_{j,k-1,0}^e)$ and $\mathbf{P}(\Gamma_{J+1,0,0} | \Gamma_{1,0,0}) = \prod_{j=1}^{J}  \mathbf{P}(\Gamma_{j+1,0,0}^e | \Gamma_{j,0,0}^e)$.
Using Lemmas \ref{lem_add} and \ref{lem_del}, and the fact that $p_k(\alpha,\zeta) \ge p_K(\alpha,\zeta)$ \cite[Lemma 35]{rrpcp_perf}, we get $\mathbf{P}(\Gamma_{J+1,0,0}^e| \Gamma_{1,0,0}) \ge {p}_K(\alpha,\zeta)^{KJ} \tilde{p}(\tilde{\alpha},\zeta)^{\vartheta_{\max} J}$.
Also, $\mathbf{P}(\Gamma_{1,0,0}^e)=1$. This follows by the assumption on $\hat{P}_0$ and Lemma \ref{lem_et}. Thus, $\mathbf{P}(\Gamma_{J+1,0,0}^e) \ge {p}_K(\alpha,\zeta)^{KJ} \tilde{p}(\tilde{\alpha},\zeta)^{\vartheta_{\max} J}$.

Using the definitions of $\alpha_\add(\zeta)$ and $\alpha_\del(\zeta)$ and $\alpha \ge \alpha_\add$ and $\tilde{\alpha} \ge \alpha_\del$, $\mathbf{P}(\Gamma_{J+1,0,0}^e) \ge {p}_K(\alpha,\zeta)^{KJ} \tilde{p}(\tilde{\alpha},\zeta)^{\vartheta_{\max} J} \ge (1-n^{-10})^2 \ge 1- 2n^{-10}$.


The event $\Gamma_{J+1,0,0}^e$ implies that $\That_t=T_t$ and $e_t$ satisfies (\ref{etdef}) for all $t < t_{J+1}$. Using Remark \ref{SE_rem} and the third claim of Remark \ref{Gamma_rem}, $\Gamma_{J+1,0,0}^e$ implies that all the bounds on the subspace error hold. Using these, Remark \ref{etdef_rem}, $\|a_{t,\new}\|_2 \le \sqrt{c} \gamma_{\new,k}$ and $\|a_t\|_2 \le \sqrt{r} \gamma_*$, $\Gamma_{J+1,0,0}^e$ implies that all the bounds on $\|e_t\|_2$ hold (the bounds are obtained in in Lemmas \ref{lem_et} and \ref{cslem}).

Thus, all conclusions of the the result hold w.p. at least $1- 2n^{-10}$.

%



\section{Lemmas used to prove Lemma \ref{lem_del}}\label{keylems}

In this section, we remove the subscript $j$ at most places. The convention of Remark \ref{remove_j} applies. 

\subsection{Showing exact support recovery and getting an expression for $e_t$}

\begin{lem}[Bounding the RIC of $\Phi_k$] \label{RIC_bnd}
The following hold. 
\ben
\item $\delta_{s} (\Phi_0) = \kappa_s^2 (\Phat_*) \leq  \kappa_{s,*}^2 + 2 \zeta_*$
\item $\delta_{s}(\Phi_k)  = \kappa_s^2 ([\Phat_* \ \Phat_{\new,k}]) \leq \kappa_s^2 (\Phat_*) + \kappa_s^2 (\Phat_{\new,k})\leq \kappa_{s,*}^2 + 2\zeta_* + (\kappa_{s,\new} + \tilde{\kappa}_{s,k} \zeta_k + \zeta_*)^2$ for $k =1,2 \dots K$
\een
\end{lem}
\begin{proof}
The above lemma is the same as the last two claims of \cite[Lemma 28]{rrpcp_perf}. It follows using Lemma \ref{delta_kappa} and some linear algebraic manipulations.
\end{proof}

\begin{lem}[Sparse recovery, support recovery and expression for $e_t$]\label{lem_et} 
Assume that the conditions of Theorem \ref{thm2} hold.
\ben
\item For all $k=1,2,\dots \vartheta+1$,  $X_{j,K,k-1} \in \Gamma_{j,K,k-1}$ implies that
\ben
\item $\zeta_* \le \zeta_*^+:=r \zeta$, $\zeta_{K} \le c \zeta$, $\|\Phi_{K} P_{j}\|_2 \le (r+c) \zeta$,  
\item $\delta_s(\Phi_K) \le 0.1479$ and $\phi_K \le \phi^+:=1.1735$ 
\item for any $t \in \tilde{\mathcal{I}}_{j,k}$,
\ben
\item the projection noise $\beta_t:= (I - \Phat_{(t-1)} \Phat_{(t-1)}') L_t$ satisfies $\|\beta_t\|_2 \leq \sqrt{\zeta}$,
\item the CS error satisfies $\|\hat{S}_{t,\cs} - S_t\|_2 \leq 7 \sqrt{\zeta}$,
\item $\hat{T}_t = T_t$,
\item $e_t$ satisfies (\ref{etdef}) and $\|e_t\|_2 \leq \phi^+ \sqrt{\zeta}$.
\een

\een

\item For all $k=1,2,\dots \vartheta+1$, $\mathbf{P} ( T_t = \hat{T}_t \ \text{and} \ e_t \ \text{satisfies} \  (\ref{etdef}) \ \text{for  all } \ t \in \tilde{\mathcal{I}}_{j,k} \ | X_{j,K,k-1}) =1$ for all $X_{j,K,k-1} \in \Gamma_{j,K,k-1}$.
\item For all $k=1,2,\dots \vartheta+1$, $\mathbf{P} ( T_t = \hat{T}_t \ \text{and} \ e_t \ \text{satisfies} \  (\ref{etdef}) \ \text{for  all } \ t \in \tilde{\mathcal{I}}_{j,k} \ | \Gamma_{j,K,k-1}^e) =1$.
\een
\end{lem}

\begin{proof} 


Claim 1-a follows using Remark \ref{Gamma_rem}. 
Claim 1-b) follows using claim 1-a) and Lemma \ref{RIC_bnd}.
Claim 1-c) follows in a fashion similar to the proof of \cite[Lemma 30]{rrpcp_perf}. The main difference is that everywhere we use $\Phi_K L_t = \Phi_K P_j a_t$ and $\|\Phi_{K} P_{j}\|_2 \le (r+c) \zeta$.
Claim 1-c-i) uses this and the fact that for $t \in \tilde{\mathcal{I}}_{j,k}$, $\Phi_{(t)} = \Phi_K$, and $\sqrt\zeta \le \sqrt{\gamma_*^2/(r+c)^3}$.
Claim 1-c-ii) uses c-i), $\sqrt{\zeta} \le \xi$ (defined in the theorem), $\delta_{2s} (\Phi_{K}) \le 0.1479$, and Theorem \ref{candes_csbound}.
Claim 1-c-iii) uses c-ii), the definition of $\rho$, the choice of $\omega$ and the lower bound on $S_{\min}$ given in the theorem.
Claim 1-c-iv) uses claim c-iii) and Remark \ref{etdef_rem}. To get the bound on $\|e_t\|_2$ we use the first expression of (\ref{etdef}),  $\phi_K \le \phi^+:=1.1735$, and   $\sqrt\zeta \le \sqrt{\gamma_*^2/(r+c)^3}$.


Claim 2) is just a rewrite of claim 1). Claim 3) follows from claim 2) by Lemma \ref{rem_prob}.

\end{proof}

\subsection{A lemma needed for bounding the subspace error, $\tilde\zeta_k$}


\begin{lem}\label{bound_R}
Assume that $\tilde{\zeta}_{k'} \leq \tilde{c}_{k'} \zeta$ for $k'=1,\cdots, k-1$. Then
\ben
\item $\|D_{\text{det},k}\|_2 = \|\Psi_{k-1} G_{\text{det},k}\|_2 \leq r \zeta$.
\item $\|G_{\text{det},k} {G_{\text{det},k}}' - \hat{G}_{\text{det},k}{\hat{G}_{\text{det},k}}'\|_2 \leq 2 r\zeta$.
\item $0< \sqrt{1-r^2 \zeta^2} \leq \sigma_i(D_k) = \sigma_i(R_k) \leq 1$. Thus, $\|D_k\|_2 = \|R_k\|_2 \le 1$ and $\|D_k^{-1}\|_2 = \|R_k^{-1}\|_2 \le 1/\sqrt{1-r^2 \zeta^2} $.
\item $\|{D_{\text{undet},k}}'E_k \|_2 = \|{G_{\text{undet},k}}'E_k \|_2 \leq \frac{r^2 \zeta^2}{\sqrt{1-r^2\zeta^2}}$. 
\een
\end{lem}
\begin{proof} The first claim essentially follows by using the fact that $\hat{G}_1,\cdots,\hat{G}_{k-1}$ are mutually orthonormal and triangle inequality. Recall that $\Psi_{k-1} = (I - \hat{G}_{\text{det},k} {\hat{G}_{\text{det},k}}')$. The last three claims use this and the first claim and apply Lemma \ref{lemma0}. The last claim also uses the definition of $D_k$ and its QR decomposition. The complete proof is given in Appendix \ref{proof_lem_bound_R}. \end{proof}

\subsection{Bounding on the subspace error, $\tilde\zeta_k$}

\begin{lem}[Bounding $\tilde{\zeta_k}^+$] \label{bnd_tzetakp}
If
\beq
f_{dec}(\tilde{g}_{\max},\tilde{h}_{\max}) - \frac{f_{inc}(\tilde{g}_{\max},\tilde{h}_{\max})}{\tilde{c}_{\min} \zeta} > 0  \label{Func_del}
\eeq
then $f_{dec}(\tilde{g}_k,\tilde{h}_k) >0$ and $\tilde\zeta_k^+ \le \tilde{c}_k \zeta$.
\end{lem}
\begin{proof} Recall that $f_{inc}(.)$, $f_{dec}(.)$ are defined in Definition \ref{defzetap} and $\tilde{\zeta_k}^+ := \frac{f_{inc}(\tilde{g},\tilde{h})}{f_{dec}(\tilde{g},\tilde{h})}$.
Notice that $f_{inc}(.)$ is a non-decreasing function of $\tilde{g},\tilde{h}$, and $f_{dec}(.)$ is a non-increasing function. Using the definition of $\tilde{g}_{\max},\tilde{h}_{\max}, \tilde{c}_{\min}$ given in Assumption \ref{assu1}, the result follows.
\end{proof}

\begin{remark} \label{f_inc_rem}
If we ignore the small terms of $f_{inc}(.)$ and $f_{dec}(.)$, the above condition simplifies to requiring that $\frac{3\kappa_{s,e}^+ \phi^+ \tilde{g}_{\max} + \kappa_{s,e}^+ \phi^+ \tilde{h}_{\max}}{1-\tilde{h}_{\max}} \le \frac{\tilde{c}_{\min}}{r+c}$. Since $\tilde{g}_{\max} \ge 1$, the first term of the numerator is the largest one. To ensure that this condition holds we need $\kappa_{s,e}^+$ to be very small. However, as explained in Sec \ref{proof_lem_bound_terms}, if we also assume denseness of $D_{k}$, i.e. if we assume $\kappa_s(D_{k}) \le \kappa_{s,D}^+$ for a small enough $\kappa_{s,D}^+$, then the first term of the numerator can be replaced by $\max( 3\kappa_{s,e}^+ \kappa_{s,D}^+  \phi^+ \tilde{g}_{\max}, \kappa_{s,e}^+ \phi^+ \tilde{h}_{\max})$. This will relax the requirement on $\kappa_{s,e}^+$, e.g. now $\kappa_{s,e}^+ =\kappa_{s,D}^+= 0.3$ will work.
\end{remark}

\begin{lem}[Bounding $\tilde{\zeta_k}$] \label{defnPCA}
If $\lambda_{\min}(\tilde{A}_k) - \lambda_{\max}(\tilde{A}_{k,\perp}) - \|\tilde{\mathcal{H}}_k\|_2 >0$, then
\beq
\tilde{\zeta_k} \leq  \frac{\|\tilde{\mathcal{H}}_k\|_2}{\lambda_{\min} (\tilde{A}_k) - \lambda_{\max} (\tilde{A}_{k,\perp}) -  \|\tilde{\mathcal{H}}_k\|_2}
\label{zetakbnd_del}
\eeq
\end{lem}
\begin{proof}
Recall that $\tilde{A}_k$, $\tilde{A}_{k,\perp}$, $\tilde{\mathcal{H}}_k$ are defined in Definition \ref{defHk_del}. The result follows by using the fact that $\tilde{\zeta}_k = \|(I - \hat{G}_k \hat{G}_k') D_{j,k}\|_2 = \|(I - \hat{G}_k \hat{G}_k') E_{k} R_{k}\|_2 \le \|(I - \hat{G}_k \hat{G}_k') E_{k}\|_2$ and applying Lemma \ref{sin_theta_weyl} with $E \equiv E_k$ and $F \equiv \hat{G}_k$.
\end{proof}

\begin{lem}[High probability bounds for each of the terms in the $\tilde{\zeta}_k$ bound and for $\tilde\zeta_k$]\label{lem_bound_terms} 
 Assume that the conditions of Theorem \ref{thm2} hold. Also, assume that $\mathbf{P}(\Gamma_{j,K,k-1}^e)>0$. Then,  for all $1 \leq k \leq \vartheta_j$,
\ben
\item $\mathbf{P}(\lambda_{\min}(\tilde{A}_{k}) \geq \lambda_{k}^-(1-r^2 \zeta^2  - 0.1 \zeta) | \Gamma_{j,K,k-1}^e) >1- \tilde{p}_1(\tilde{\alpha},\zeta)$
with $\tilde{p}_1(\tilde{\alpha},\zeta)$ given in (\ref{prob1}).
\item $\mathbf{P}(\lambda_{\max}(\tilde{A}_{k,\perp}) \leq \lambda_k^- (\tilde{h}_k+r^2 \zeta^2 f+0.1 \zeta) | \Gamma_{j,K,k-1}^e) > 1-\tilde{p}_2(\tilde{\alpha},\zeta)$ with $\tilde{p}_2(\tilde{\alpha},\zeta)$ given in (\ref{prob2}).
\item $\mathbf{P}(\|\tilde{\mathcal{H}}_{k}\|_2 \leq \lambda_k^- f_{inc}(\tilde{g}_k,\tilde{h}_k) \ |\Gamma_{j,K,k-1}^e) \geq 1 - \tilde{p}_3(\tilde{\alpha},\zeta)$ with $\tilde{p}_3(\tilde{\alpha},\zeta)$ given in (\ref{prob3}).
\item $\mathbf{P}( \lambda_{\min} (\tilde{A}_k) -   \lambda_{\max} (\tilde{A}_{k,\perp})  - \|\tilde{\mathcal{H}}_k\|_2 \ge \lambda_k^- f_{dec}(\tilde{g}_k,\tilde{h}_k) \ |\Gamma_{j,K,k-1}^e) \geq \tilde{p}(\tilde{\alpha},\zeta) :=1- \tilde{p}_{1}(\tilde{\alpha},\zeta) - \tilde{p}_{2}(\tilde{\alpha},\zeta) - \tilde{p}_{3}(\tilde{\alpha},\zeta)$.
\item If $f_{dec}(\tilde{g}_k,\tilde{h}_k) >0$, then $\mathbf{P}(\tilde{\zeta}_k \leq \tilde\zeta_k^+ \ | \Gamma_{j,K,k-1}^e) \geq \tilde{p} (\tilde{\alpha},\zeta)$ \nn
\een
\end{lem}
\begin{proof}
Recall that $f_{inc}(.)$, $f_{dec}(.)$ and $\tilde\zeta_k^+$ are defined in Definition \ref{defzetap}. The proof of the first three claims is given in Sec \ref{proof_lem_bound_terms}. The fourth claim follows directly from the first three using the union bound on probabilities. The fifth claim follows from the fourth using Lemma \ref{defnPCA}.
\end{proof}

\begin{lem}[High probability bound on $\tilde{\zeta_k}$]\label{tilde_zeta}
Assume that the conditions of Theorem \ref{thm2} hold. Then,
\beq
\mathbf{P} (\tilde{\zeta}_k \leq \tilde{c}_k \zeta \ | \Gamma_{j,K,k-1}^e) \geq  \tilde{p}(\tilde{\alpha},\zeta) \nn
\eeq
\end{lem}
\begin{proof} This follows by combining Lemma \ref{bnd_tzetakp} and the last claim of Lemma \ref{lem_bound_terms}. \end{proof}

\subsection{Proof of Lemma \ref{lem_bound_terms}} \label{proof_lem_bound_terms}

\begin{proof}
We use $\frac{1}{\tilde{\alpha}}\sum_t$ to denote $\frac{1}{\tilde{\alpha}} \sum_{t \in \tilde{\mathcal{I}}_{j,k}}$.

For $ t \in \tilde{\mathcal{I}}_{j,k}$, let $a_{t,k} := {G_{j,k}}'L_t$, $a_{t,\text{det}} := {G_{\text{det},k}}'L_t = [G_{j,1},\cdots G_{j,k-1}]'L_t$ and $a_{t,\text{undet}}:= {G_{\text{undet},k}}'L_t = [G_{j,k+1}\cdots G_{j,\vartheta_j}]'L_t$. Then $a_t:= P_j'L_t$ can be split as
$a_t = [ a_{t,\text{det}}' \ a_{t,k}' \ a_{t,\text{undet}}']'$. 

This lemma follows using the following facts and the Hoeffding corollaries, Corollary \ref{hoeffding_nonzero} and \ref{hoeffding_rec}.
\ben
\item The statement {\em ``conditioned on r.v. $X$, the event ${\cal E}^e$ holds w.p. one for all $X \in \Gamma$" is equivalent to ``$\mathbf{P}({\cal E}^e|X)=1, \ \text{for all} \ X \in \Gamma$".} We often use the former statement in our proofs since it is often easier to interpret. 

\item The matrices $D_k$, $R_k$, $E_k$, $D_{\text{det},k}, D_{\text{undet},k}$, $\Psi_{k-1}$, $\Phi_K$ are functions of the r.v. $X_{j,K,k-1}$. All terms that we bound for the first two claims of the lemma are of the form $\frac{1}{\alpha} \sum_{t \in \mathcal{\tilde{I}}_{j,k}} Z_t$ where $Z_t= f_1(X_{j,K,k-1}) Y_t f_2(X_{j,K,k-1})$, $Y_t$ is a sub-matrix of $a_t a_t'$ and $f_1(.)$ and $f_2(.)$ are functions of $X_{j,K,k-1}$. For instance, one of the terms while bounding $\lambda_{\min}(\mathcal{A}_k)$ is $\frac{1}{\tilde{\alpha}} \sum_t R_k a_{t,k} {a_{t,k}}'{R_k}'$.

\item  $X_{j,K,k-1}$ is independent of any $a_{t}$ for $t \in  \mathcal{\tilde{I}}_{j,k}$ , and hence the same is true for the matrices  $D_k$, $R_k$, $E_k$, $D_{\text{det},k}, D_{\text{undet},k}$, $\Psi_{k-1}$, $\Phi_K$. Also, $a_{t}$'s for different $t \in \mathcal{\tilde{I}}_{j,k}$ are mutually independent. Thus, conditioned on  $X_{j,K,k-1}$, the $Z_t$'s defined above are mutually independent.
\label{X_at_indep}

\item All the terms that we bound for the third claim contain $e_t$. Using the second claim of Lemma \ref{lem_et}, conditioned on $X_{j,K,k-1}$, $e_t$ satisfies (\ref{etdef}) w.p. one whenever $X_{j,K,k-1} \in \Gamma_{j,K,k-1}$. Conditioned on $X_{j,K,k-1}$, all these terms are also of the form $\frac{1}{\alpha} \sum_{t \in \mathcal{\tilde{I}}_{j,k}} Z_t$ with $Z_t$ as defined above, whenever $X_{j,K,k-1} \in \Gamma_{j,K,k-1}$. Thus, conditioned on  $X_{j,K,k-1}$, the $Z_t$'s for these terms are mutually independent, whenever $X_{j,K,k-1} \in \Gamma_{j,K,k-1}$.

\item By Remark \ref{Gamma_rem}, $X_{j,K,k-1} \in \Gamma_{j,K,k-1}$ implies that $\zeta_{*} \le r \zeta$, $\tilde\zeta_{k'} \le c_{k'} \zeta, \ \text{for all} \ k'=1,2,\dots k-1$, $\zeta_K \le \zeta_K^+ \le c \zeta$, (iv) $\phi_K \le \phi^+$ (by Lemma \ref{lem_et}); (v) $\|\Phi_K P_j\|_2 \le (r+c)\zeta$; and (vi) all conclusions of Lemma \ref{bound_R} hold. 


\item By the clustering assumption, $ \lambda_k^- \le \lambda_{\min}(\E(a_{t,k}{a_{t,k}}')) \le \lambda_{\max}(\E(a_{t,k}{a_{t,k}}')) \le  \lambda_k^+$; $\lambda_{\max}(\E(a_{t,\text{det}}{a_{t,\text{det}}}')) \le \lambda_1^+ = \lambda^+$; and $\lambda_{\max}(\E(a_{t,\text{undet}}{a_{t,\text{undet}}}')) \le \lambda_{k+1}^+$. Also, $\lambda_{\max}(\E(a_t a_t')) \le \lambda^+$.

\item By Weyl's theorem,  for a sequence of matrices $B_t$, $\lambda_{\min}(\sum_t B_t) \ge \sum_t \lambda_{\min}(B_t)$ and $\lambda_{\max}(\sum_t B_t) \le \sum_t \lambda_{\max}(B_t)$.
\een


 Consider $\tilde{A}_k = \frac{1}{\tilde{\alpha}} \sum_t {E_k}' \Psi_{k-1} L_t{L_t}' \Psi_{k-1} E_k$. Notice that
${E_k}' \Psi_{k-1} L_t = R_k a_{t,k} + {E_k}'(D_{\text{det},k} a_{t,\text{det}} + D_{\text{undet},k} a_{t,\text{undet}})$. Let $Z_t = R_k a_{t,k} {a_{t,k}}'{R_k}'$ and let $Y_t = R_k a_{t,k}  ({a_{t,\text{det}}}'{D_{\text{det},k}}' + {a_{t,\text{undet}}}'{D_{\text{undet},k}}')E_k +  E_k'(D_{\text{det},k} a_{t,\text{det}} + D_{\text{undet},k} a_{t,\text{undet}}) {a_{t,k}}'{R_k}'$. Then
\beq
\tilde{A}_k \succeq \frac{1}{\tilde{\alpha}} \sum_t Z_t + \frac{1}{\tilde{\alpha}} \sum_t Y_t\label{lemmabound_1}
\eeq

Consider $\frac{1}{\tilde{\alpha}} \sum_t Z_t = \frac{1}{\tilde{\alpha}} \sum_t R_k a_{t,k} {a_{t,k}}'{R_k}'$. (a) As explained above, the $Z_t$'s are conditionally independent given $X_{j,K,k-1}$. (b) Using Ostrowoski's theorem and Lemma \ref{bound_R}, for all $X_{j,K,k-1} \in \Gamma_{j,K,k-1}$, $\lambda_{\min}( \E(\frac{1}{\tilde{\alpha}}\sum_t Z_t|X_{j,K,k-1})) = \lambda_{\min}( R_{k} \frac{1}{\tilde{\alpha}}\sum_t \E(a_{t,k}{a_{t,k}}') {R_{k}}') \ge \lambda_{\min} (R_{k} {R_{k}}')\lambda_{\min} (\frac{1}{\tilde{\alpha}}\sum_t \E(a_{t,k}{a_{t,k}}')) \geq (1- r^2 \zeta^2)\lambda_{k}^-$.
(c) Finally, using $\|R_k\|_2 \leq 1$ and $\|a_{t,k}\|_2 \leq \sqrt{\tilde{c}_k} \gamma_* $, conditioned on $X_{j,K,k-1}$, $0 \preceq Z_t  \preceq \tilde{c}_k \gamma_{*}^2 I  $ holds w.p. one for all $X_{j,K,k-1} \in \Gamma_{j,K,k-1}$.

Thus, applying Corollary \ref{hoeffding_nonzero} with $\epsilon = 0.1 \zeta \lambda^-$, and using $\tilde{c}_k \leq r$, for all $X_{j,K,k-1} \in  \Gamma_{j,K,k-1}$,
\beq
\mathbf{P}(\lambda_{\min} (\frac{1}{\tilde{\alpha}} \sum_t Z_t)
\geq   (1- r^2\zeta^2)\lambda_{k}^-  - 0.1 \zeta \lambda^- | X_{j,K,k-1})
\geq  1- \tilde{c}_k \exp (-\frac{\tilde{\alpha} \epsilon^2 }{8  (\tilde{c}_k  \gamma_{*}^2)^2})
\geq 1 - r \exp (-\frac{\tilde{\alpha} \cdot (0.1 \zeta \lambda^-)^2 }{8  r^2 \gamma_{*}^4})
  \label{lemma_add_A1}
\eeq


Consider $Y_t = R_k a_{t,k}  ({a_{t,\text{det}}}'{D_{\text{det},k}}' + {a_{t,\text{undet}}}'{D_{\text{undet},k}}')E_k +  E_k'(D_{\text{det},k} a_{t,\text{det}} + D_{\text{undet},k} a_{t,\text{undet}}) {a_{t,k}}'{R_k}'$. (a) As before, the $Y_t$'s are conditionally independent given $X_{j,K,k-1}$. (b) Since $\E[a_t]=0$ and $\text{Cov}[a_t]=\Lambda_t$ is diagonal, $\E(\frac{1}{\alpha}\sum_t Y_t|X_{j,K,k-1}) = 0$ whenever $X_{j,K,k-1} \in \Gamma_{j,K,k-1}$. (c) Conditioned on $X_{j,K,k-1}$, $\|Y_t\|_2 \le 2 \sqrt{\tilde{c}_k r} \gamma_*^2 r\zeta(1+ \frac{r\zeta}{\sqrt{1-r^2\zeta^2}}) \leq 2 r^2 \zeta \gamma_*^2 ( 1+ \frac{10^{-4}}{\sqrt{1-10^{-4}}}) \leq \frac{2}{r} ( 1+ \frac{10^{-4}}{\sqrt{1-10^{-4}}}) < 2.1$ holds w.p. one for all $X_{j,K,k-1} \in \Gamma_{j,K,k-1}$. This follows because $X_{j,K,k-1} \in \Gamma_{j,K,k-1}$ implies that $\|D_{\text{det},k}\|_2 \leq r\zeta$, $\|{E_k}' D_{\text{undet},k}\|_2 = \| {E_k}' G_{\text{undet},k}\|_2 \leq \frac{r^2 \zeta^2}{\sqrt{1-r^2\zeta^2}}$.
Thus, under the same conditioning, $-b I \preceq Y_t  \preceq b I$ with $b =2.1$ w.p. one.
Thus, applying Corollary \ref{hoeffding_nonzero} with $\epsilon = 0.1 \zeta \lambda^-$, we get
\beq
\mathbf{P}(\lambda_{\min} (\frac{1}{\tilde{\alpha}} \sum_t Y_t) \geq - 0.1 \zeta \lambda^- | X_{j,K,k-1} ) \geq 1- r \exp ( -\frac{\tilde{\alpha} (0.1 \zeta \lambda^-)^2} {8 \dot (4.2)^2 })  \ \text{for all $X_{j,K,k-1} \in  \Gamma_{j,K,k-1}$}
\label{lemma_add_A2}
\eeq

Combining (\ref{lemmabound_1}), (\ref{lemma_add_A1}) and (\ref{lemma_add_A2}) and using the union bound, $\mathbf{P} (\lambda_{\min}(\tilde{A}_k) \geq \lambda_{k}^-(1 - r^2\zeta^2) - 0.2 \zeta \lambda^-| X_{j,K,k-1}) \geq 1-\tilde{p}_1(\tilde{\alpha},\zeta) \ \text{for all $X_{j,K,k-1} \in  \Gamma_{j,K,k-1}$}$ where
\beq
\tilde{p}_1 (\tilde{\alpha},\zeta) := r \exp (-\frac{\tilde{\alpha} \cdot (0.1 \zeta \lambda^-)^2 }{8  r^2 \gamma_{*}^4}) + r \exp ( -\frac{\tilde{\alpha} (0.1 \zeta \lambda^-)^2} {8 \dot (4.2)^2}) \label{prob1}
\eeq
The first claim of the lemma follows by using $\lambda_{k}^- \ge \lambda^-$ and applying Lemma \ref{rem_prob} with $X \equiv X_{j,K,k-1}$ and $\calc \equiv \Gamma_{j,K,k-1}$.

Consider $\tilde{A}_{k,\perp} := \frac{1}{\alpha} \sum_t {E_{k,\perp}}' \Psi_{k-1} L_t {L_t}' \Psi_{k-1} E_{k,\perp}$. Notice that ${E_{k,\perp}}' \Psi_{k-1} L_t = {E_{k,\perp}}' (D_{\text{det},k} a_{t,\text{det}} + D_{\text{undet},k} a_{t,\text{undet}})$. Thus, $\tilde{A}_{k,\perp} = \frac{1}{\tilde{\alpha}} \sum_t Z_t$ with $Z_t={E_{k,\perp}}' (D_{\text{det},k} a_{t,\text{det}} + D_{\text{undet},k} a_{t,\text{undet}})(D_{\text{det},k} a_{t,\text{det}} + D_{\text{undet},k} a_{t,\text{undet}})'E_{k,\perp}$ which is of size $(n-\tilde{c}_k)\times (n-\tilde{c}_k)$.
(a) As before, given $X_{j,K,k-1}$, the $Z_t$'s are independent.
(b) Conditioned on $X_{j,K,k-1}$, $0 \preceq Z_t  \preceq r \gamma_*^2 I$ w.p. one for all $X_{j,K,k-1} \in  \Gamma_{j,K,k-1}$.
(c) $\E(\frac{1}{\alpha}\sum_t Z_t|X_{j,K,k-1}) \preceq (\lambda_{k+1}^+ + r^2\zeta^2 \lambda^+)I$ for all $X_{j,K,k-1} \in  \Gamma_{j,K,k-1}$.

Thus applying Corollary \ref{hoeffding_nonzero} with $\epsilon =0.1 \zeta \lambda^-$ and using $\tilde{c}_k \geq \tilde{c}_{\min}$, we get
 $$\mathbf{P}(\lambda_{\max}(\tilde{A}_{k,\perp})  \leq \lambda_{k+1}^+ + r^2 \zeta^2 \lambda^+ + 0.1 \zeta \lambda^- | X_{j,K,k-1}) \geq
 1- \tilde{p}_2(\tilde{\alpha},\zeta) \ \text{for all $X_{j,K,k-1} \in  \Gamma_{j,K,k-1}$}$$  
where
\beq
\tilde{p}_2(\tilde{\alpha},\zeta) := (n-\tilde{c}_{\min}) \exp (-\frac{\tilde{\alpha} (0.1 \zeta \lambda^-)^2}{8   r^2 \gamma_*^4 }) \label{prob2}
\eeq
The second claim follows using $\lambda_{k}^- \ge \lambda^-$, $f:=\lambda^+/\lambda^-$, $\tilde{h}_k := {\lambda_{k+1}}^+ / {\lambda_{k}}^-$ in the above expression and applying Lemma \ref{rem_prob}.

Consider the third claim. Using the expression for $\tilde{\mathcal{H}}_k$ given in Definition \ref{defHk_del}, it is easy to see that
\bea
\|\tilde{\mathcal{H}}_k \|_2 &\leq&  \max\{ \|\tilde{H}_k\|_2, \|\tilde{H}_{k,\perp}\|_2 \} + \|\tilde{B}_k\|_2 \leq \|\frac{1}{\tilde{\alpha}} \sum_t e_t {e_t}'\|_2 +  \max(\|T2\|_2, \|T4\|_2) + \|\tilde{B}_k\|_2
\label{add_calH1}
\eea
where $T2:= \frac{1}{\tilde{\alpha}} \sum_t  {E_{k}}' \Psi_{k-1}( L_t {e_t}' + e_t {L_t}')\Psi_{k-1} E_{k}$ and $T4 :=\frac{1}{\tilde{\alpha}} \sum_t {E_{k,\perp}}'\Psi_{k-1} (L_t {e_t}' + {e_t}'L_t)\Psi_{k-1} E_{k,\perp}$. The second inequality follows by using the facts that (i) $\tilde{H}_k = T1 - T2$ where $T1 := \frac{1}{\tilde{\alpha}} \sum_t {E_{k}}' \Psi_{k-1} e_t {e_t}'\Psi_{k-1} E_{k}$, (ii) $\tilde{H}_{k,\perp} = T3 - T4$ where $T3 := \frac{1}{\tilde{\alpha}} \sum_t {E_{k,\perp}}'\Psi_{k-1} e_t {e_t}'\Psi_{k-1}  E_{k,\perp}$, and (iii) $\max(\|T1\|_2, \|T3\|_2) \le \|\frac{1}{\tilde{\alpha}} \sum_t e_t {e_t}'\|_2$.

Next, we obtain high probability bounds on each of the terms on the RHS of (\ref{add_calH1}) using the Hoeffding corollaries. 

Consider $\|\frac{1}{\tilde{\alpha}} \sum_t e_t {e_t}'\|_2$. Let $Z_t = e_t {e_t}'$.
(a) As explained in the beginning of the proof, conditioned on $X_{j,K,k-1}$, the various $Z_t$'s in the summation are independent whenever $X_{j,K,k-1} \in \Gamma_{j,K,k-1}$.
(b) Conditioned on $X_{j,K,k-1}$, $0 \preceq Z_t \preceq b_1 I$ w.p. one for all $X_{j,K,k-1} \in  \Gamma_{j,K,k-1}$. Here $b_1:={\phi^+}^2 \zeta$.
(c) Using $\|\Phi_K P_j\|_2 \leq (r+c) \zeta$, $0 \preceq \frac{1}{\alpha} \sum_t \E(Z_t|X_{j,K,k-1}) \preceq b_2I, \ b_2:= (r+c)^2 \zeta^2 {\phi^+}^2 \lambda^+$ for all $X_{j,K,k-1} \in  \Gamma_{j,K,k-1}$.

Thus, applying Corollary \ref{hoeffding_nonzero} with $\epsilon = 0.1 \zeta \lambda^-$,
\beq
\mathbf{P} ( \|\frac{1}{\tilde{\alpha}} \sum_t  e_t {e_t}' \|_2 \leq b_2  +  0.1 \zeta \lambda^-| X_{j,K,k-1} ) \geq 1- n \exp(-\frac{\tilde{\alpha}( 0.1 \zeta \lambda^-)^2}{ 8 \cdot b_1^2})  \ \text{for all $X_{j,K,k-1} \in  \Gamma_{j,K,k-1}$}
\label{add_etet}
\eeq

Consider $T2$. Let $Z_t: = {E_{k}}' \Psi_{k-1} (L_t {e_t}' + e_t{L_t}')\Psi_{k-1} E_{k}$ which is of size $\tilde{c}_k \times \tilde{c}_k$. Then $T2 = \frac{1}{\tilde{\alpha}} \sum_t Z_t$.
(a) Conditioned on $X_{j,K,k-1}$, the various $Z_t$'s used in the summation are mutually independent whenever $X_{j,K,k-1} \in \Gamma_{j,K,k-1}$.
(b) Notice that ${E_{k}}'\Psi_{k-1} L_t  = R_{k} a_{t,k} + {E_k}' (D_{\text{det},k} a_{t,\text{det}} + D_{\text{undet},k}a_{t,\text{undet}})$ and
${E_{k}}'\Psi_{k-1} e_t = (R_{k}^{-1})' D_k' e_t = (R_{k}^{-1})' D_k' I_{T_t} [(\Phi_K)_{T_t}' (\Phi_K)_{T_t}]^{-1} {I_{T_t}}' \Phi_K P_j a_{t}$.
Thus conditioned on $X_{j,K,k-1}$, $\|Z_t\|_2 \leq  2 b_3$ w.p. one for all $X_{j,K,k-1} \in  \Gamma_{j,K,k-1}$. Here,
$b_3:= \frac{\sqrt{r\zeta}}{\sqrt{1-r^2\zeta^2}} \phi^+ \gamma_*$. This follows using $\|(R_{k}^{-1})'\|_2\leq 1/\sqrt{1-r^2\zeta^2}$, $\|e_t\|_2\leq \phi^+ \sqrt{\zeta}$ and $\|E_k'\Psi_{k-1}L_t\|_2 \le \|L_t\|_2 \leq \sqrt{r} \gamma_*$.
(c) Also, $\|\frac{1}{\alpha} \sum_t \E(Z_t|X_{j,K,k-1})\|_2 \leq  2 b_4$ where $b_4:=  \kappa_{s,e} (r+c) \zeta \phi^+ ( \lambda_k^+ + r \zeta \lambda^+ + \frac{r^2\zeta^2}{\sqrt{1-r^2\zeta^2}} \lambda_{k+1}^+)$.

Thus, applying Corollary \ref{hoeffding_rec} with $\epsilon = 0.1 \zeta \lambda^-$, for all $X_{j,K,k-1} \in \Gamma_{j,K,k-1}$,
\beq
\mathbf{P}( \|T2\|_2 \leq 2 b_4  + 0.1 \zeta \lambda^- | X_{j,K,k-1}) \geq 1- \tilde{c}_k\exp(-\frac{\tilde{\alpha}(0.1 \zeta \lambda^-)^2}{32 \cdot 4 b_3^2}) \nn
\eeq

Consider $T4$. Let $Z_t: = {E_{k,\perp}}'\Psi_{k-1} (L_t {e_t}' + e_t{L_t}')\Psi_{k-1} E_{k,\perp}$ which is of size $(n-\tilde{c}_k)\times (n-\tilde{c}_k)$. Then $T4 = \frac{1}{\tilde{\alpha}} \sum_t Z_t$. (a) conditioned on $X_{j,K,k-1}$, the various $Z_t$'s used in the summation are mutually independent whenever $X_{j,K,k-1} \in  \Gamma_{j,K,k-1}$.
(b) Notice that ${E_{k,\perp}}'\Psi_{k-1} L_t ={E_{k,\perp}}' (D_{\text{det},k} a_{t,\text{det}} + D_{\text{undet},k}a_{t,\text{undet}})$. Thus, conditioned on $X_{j,K,k-1}$, $\|Z_t\|_2 \leq  2b_5$ w.p. one for all $X_{j,K,k-1} \in  \Gamma_{j,K,k-1}$. Here $b_5:=   \sqrt{r\zeta} \phi^+ \gamma_*$.
(c) Also, for all $X_{j,K,k-1} \in \Gamma_{j,K,k-1}$, $\|\frac{1}{\tilde{\alpha}} \sum_t \E(Z_t|X_{j,K,k-1})\|_2 \leq 2 b_6, \ b_6:= \kappa_{s,e}  (r+c) \zeta \phi^+ (\lambda_{k+1}^+ + r \zeta \lambda^+)$.
Applying Corollary \ref{hoeffding_rec} with $\epsilon = 0.1 \zeta \lambda^-$, for all $X_{j,K,k-1} \in \Gamma_{j,K,k-1}$,
\beq
\mathbf{P}( \|T4\|_2 \leq 2b_6  +0.1 \zeta \lambda^-| X_{j,K,k-1}) \geq 1- (n-\tilde{c}_k) \exp(-\frac{\tilde{\alpha} (0.1 \zeta \lambda^-)^2}{32 \cdot4 b_5^2}) \geq 1- (n-\tilde{c}_{\min}) \exp(-\frac{\tilde{\alpha} (0.1 \zeta \lambda^-)^2}{32 \cdot4 b_5^2})\nn
\eeq

Consider $\max(\|T2\|_2,\|T4\|_2)$. Since $b_3 = b_5$ and $b_4 > b_6$, so $2b_6  + \epsilon < 2b_4  + \epsilon$. Therefore, for all $X_{j,K,k-1} \in  \Gamma_{j,K,k-1}$,
$$\mathbf{P}( \|T4 \|_2 \leq 2 b_4 + 0.1 \zeta \lambda^- | X_{j,K,k-1} ) \geq 1- (n-\tilde{c}_{k}) \exp(-\frac{\tilde{\alpha}(0.1 \zeta \lambda^-)^2}{32 \cdot 4 b_3^2})$$
By union bound, for all $X_{j,K,k-1} \in  \Gamma_{j,K,k-1}$,
\beq
\mathbf{P}( \max(\|T2 \|_2,\|T4 \|_2)\leq 2b_4 +  0.1 \zeta \lambda^- |X_{j,K,k-1}) \geq 1- n \exp(-\frac{\tilde{\alpha} (0.1 \zeta \lambda^-)^2}{32 \cdot 4b_3^2})
\label{add_maxT}
\eeq

Notice that if we also introduce an extra denseness coefficient $\kappa_{s,D} := \max_j \max_k \kappa_s(D_k)$, then $\mathbf{P}( \|T2\|_2 \leq 2 \kappa_{s,D} b_4  + 0.1 \zeta \lambda^- | X_{j,K,k-1}) \geq 1- \tilde{c}_k\exp(-\frac{\tilde{\alpha}(0.1 \zeta \lambda^-)^2}{32 \cdot 4 b_3^2})$. Thus, $\mathbf{P}( \max(\|T2 \|_2,\|T4 \|_2)\leq 2\max(\kappa_{s,D} b_4, b_6) +  0.1 \zeta \lambda^- |X_{j,K,k-1}) \geq 1- n \exp(-\frac{\tilde{\alpha} (0.1 \zeta \lambda^-)^2}{32 \cdot 4b_3^2})$. This would help to get a looser bounds on $\tilde{g}_{\max}$ and $\tilde{h}_{\max}$ in Theorem \ref{thm2}.

Consider $\|\tilde{B}_k\|_2$. Let $Z_t := {E_{k,\perp}}'\Psi_{k-1} (L_t-e_t)({L_t}'-{e_t}')\Psi_{k-1} E_{k}$ which is of size $(n-\tilde{c}_k)\times \tilde{c}_k$. Then $\tilde{B}_k = \frac{1}{\tilde{\alpha}} \sum_t Z_t$.
(a) conditioned on $X_{j,K,k-1}$, the various $Z_t$'s used in the summation are mutually independent whenever $X_{j,K,k-1} \in  \Gamma_{j,K,k-1}$.
(b) Notice that ${E_{k,\perp}}'\Psi_{k-1} (L_t-e_t) = {E_{k,\perp}}'( D_{\text{det},k}a_{t,\text{det}} + D_{\text{undet},k} a_{t,\text{undet}} - \Psi_{k-1} e_t)$ and  ${E_{k}}' \Psi_{k-1} (L_t - e_t) = R_{k} a_{t,k}+ {E_{k}}'( D_{\text{det},k} a_{t,\text{det}} + D_{\text{undet},k} a_{t,\text{undet}} - \Psi_{k-1} e_t)$. Thus,  conditioned on $X_{j,K,k-1}$,
$ \|Z_t\|_2 \leq b_7$
w.p. one for all $X_{j,K, k-1} \in  \Gamma_{j,K, k-1}$. Here $b_7 := (\sqrt{r}\gamma_* + \phi^+ \sqrt{\zeta})^2$.
(c)  $\|\frac{1}{\tilde{\alpha}} \sum_t \E(Z_t|X_{j,K,k-1})\|_2 \leq b_8$ for all $X_{j,K,k-1} \in  \Gamma_{j,K,k-1}$ where
\bea
b_8 :=&&  (r+c) \zeta  \kappa_{s,e}  \phi^+ \lambda_k^+
+  [(r+c)\zeta \kappa_{s,e} \phi^+ + (r+c)\zeta \kappa_{s,e} \frac{r^2\zeta^2}{\sqrt{1-r^2\zeta^2}}]\lambda_{k+1}^+  [r^2\zeta^2 + 2(r+c)r \zeta^2 \kappa_{s,e} \phi^+ + (r+c)^2 \zeta^2 \kappa_{s,e}^2 {\phi^+}^2] \lambda^+  \nn
\eea
Thus, applying Corollary \ref{hoeffding_rec} with $\epsilon=0.1 \zeta \lambda^-$,
\beq
\mathbf{P} (\|\tilde{B}_k\|_2 \leq b_8 + 0.1 \zeta \lambda^- | X_{j,K,k-1}) \geq 1 - n \exp(-\frac{\tilde{\alpha} (0.1 \zeta \lambda^-)^2}{32 \cdot b_7^2}) \ \text{for all $X_{j,K,k-1} \in   \Gamma_{j,K,k-1}$}
\label{Bk}
\eeq

Using (\ref{add_calH1}), (\ref{add_etet}), (\ref{add_maxT}) and (\ref{Bk}) and the union bound, for any $X_{j,K,k-1} \in  \Gamma_{j,K,k-1}$,
$$\mathbf{P} (\|\tilde{\mathcal{H}}_k\|_2 \leq b_9 + 0.2 \zeta \lambda^-|X_{j,K,k-1}) \geq 1- \tilde{p}_3(\tilde{\alpha},\zeta)$$
where $b_9 := b_2 +2b_4+ b_8$ and
\beq
\tilde{p}_3(\tilde{\alpha},\zeta):= n \exp(-\frac{\tilde{\alpha} \epsilon^2}{8 \cdot b_1^2}) + n \exp(-\frac{\tilde{\alpha} \epsilon^2}{32\cdot 4 b_3^2}) + n \exp(-\frac{\tilde{\alpha} \epsilon^2}{32 \cdot b_7^2}) \label{prob3}
\eeq
with $b_1 = {\phi^+}^2 \zeta$, $b_3 := \sqrt{r\zeta} \phi^+ \gamma_*$, $b_7 := (\sqrt{r} \gamma_* + \phi^+ \sqrt{\zeta})^2$.
Using $\lambda_{k}^- \ge \lambda^-$, $f:= \lambda^+/ \lambda^-$, $\tilde{g}_k := \lambda_k^+/\lambda_k^-$ and $\tilde{h}_k := \lambda_{k+1}^+ / \lambda_k^-$,  and then applying Lemma \ref{rem_prob}, the third claim of the lemma follows.
\end{proof}

%

\section{Simulation experiments}
\label{sims}

\subsubsection{Data Generation}\label{gendata}
The simulated data is generated as follows.
The measurement matrix $\mathcal{M}_t := [M_1, M_2,\cdots, M_t]$ is of size $2048 \times 5200$. It can be decomposed as a sparse matrix $\mathcal{S}_t:= [S_1, S_2,\cdots, S_t]$ plus a low rank matrix $\mathcal{L}_t:= [L_1, L_2,\cdots, L_t]$.

The sparse matrix $\mathcal{S}_t := [S_1,S_2,\cdots,S_t]$ is generated as follows. For $1 \leq t \leq t_{\train} = 200$, $S_t = 0$.
 For $t_{\train} < t \leq 5200$, $S_t$ has $s$ nonzero elements. The initial support $T_0 = \{1,2,\dots s\}$. Every $\Delta$ time instants we increment the support indices by 1. For example, for $t \in [t_\train +1, t_\train + \Delta-1]$, $T_t = T_0$, for $t \in [t_\train + \Delta, t_\train + 2\Delta-1]$, $T_t = \{2,3,\dots s+1\}$ and so on. Thus, the support set changes in a highly correlated fashion over time and this results in the matrix ${\cal S}_t$ being low rank. The larger the value of $\Delta$, the smaller will be the rank of ${\cal S}_t$ (for $t > t_\train+\Delta$).
 The signs of the nonzero elements of $S_t$ are $\pm 1$ with equal probability  and the magnitudes are uniformly distributed between $2$ and $3$. Thus, $S_{\min} = 2$.

The low rank matrix $\mathcal{L}_t := [L_1,L_2,\cdots,L_t]$ where $L_t := P_{(t)} a_t$ is generated as follows:
 There are a total of $J=2$ subspace change times, $t_1=301$ and $t_2= 2501$. $r_0 = 36$, $c_{1,\new} = c_{2,\new} =1$ and $c_{1,\old} = c_{2,\old} = 3$. Let $U$ be an $2048 \times (r_0 + c_{1,\new} + c_{2,\new})$ orthonormalized random Gaussian matrix.
 For $1\leq t \leq t_1 -1$, $P_{(t)} = P_0$ has rank $r_0$ with  $P_0 = U_{[1,2,\cdots,36]}$. 
 For $t_1  \leq t \leq t_2-1$, $P_{(t)} = P_1 = [P_0\setminus P_{1,\old} \ P_{1,\new}]$ has rank $r_1 = r_0 + c_{1,\new} - c_{1,\old} = 34$ with $ P_{1,\new} = U_{[37]}$ and $P_{1,\old} = U _{[9,18,36]}$.
For $t\geq t_2$, $P_{(t)} = P_2 = [P_1\setminus P_{2,\old} \ P_{2,\new}]$ has rank $r_2 = r_1 + c_{2,\new} - c_{2,\old} = 32$ with $ P_{2,\new} = U_{[38]}$ and $P_{1\old} = U _{[8,17,35]}$.
$a_t$ is independent over $t$.  The various $(a_t)_i$'s are also mutually independent for different $i$. 
 For $1\leq t < t_1$, we let $(a_t)_i$ be uniformly distributed between $-\gamma_{i,t}$ and $\gamma_{i,t}$, where  
\beq
\gamma_{i,t}=
\begin{cases} 400 & \text{if $i=1,2,\cdots,9, \forall t$,}
\\
30 &\text{if $i=10,11,\cdots,18, \forall t$.}
\\
2 &\text{if $i=19,20,\cdots,27, \forall t$.}
\\
1 &\text{if $i=28,29\cdots,36, \forall t$.}
\end{cases} \label{gamma0}
\eeq
For $t_1  \leq t < t_2$, $a_{t,*}$ is an $r_0 - c_{1,\old}$ length vector, $a_{t,\new}$ is a $c_{1,\new}$ length vector and $L_t := P_{(t)} a_t = P_1 a_t = (P_0\setminus P_{1,\old}) a_{t,*,nz} + P_{1,\new} a_{t,\new}$. Now, $(a_{t,*,nz})_i$ is uniformly distributed between $-\gamma_{i,t}$ and $\gamma_{i,t}$ for $i=1,2,\cdots,35$ and $a_{t,\new}$ is uniformly distributed between $-\gamma_{\new,t}$ and $\gamma_{\new,t}$, where
\bea
\gamma_{i,t} &=&
\begin{cases} 400 & \text{if $i=1,2,\cdots,8, \forall t$,}
\\
30 &\text{if $i=9,10,\cdots,16 \forall t$.}
\\
2 &\text{if $i=17,18,\cdots,24, \forall t$.}
\\
1 &\text{if $i=25,26,\cdots,33, \forall t$.}
\end{cases} \label{gamma0} \nn\\
\gamma_{\new,t} &=&
\begin{cases} 1.1^{k-1} & \text{if $t_1 + (k-1) \alpha \leq t \leq t_1 +k\alpha-1, k=1,2,3,4$,}
\\
1.1^{4-1} = 1.331  &\text{if $t \geq t_1 + 4\alpha$.}
\end{cases} \label{gamma1}
\eea
 For $t\geq t_2$, $a_{t,*}$ is an $r_1 - c_{2,\old}$ length vector, $a_{t,\new}$ is a $c_{2,\new}$ length vector and $L_t := P_{(t)} a_t = P_2 a_t = [P_0\setminus P_{1,\old} \ P_{1,\new}] a_{t,*} + P_{2,\new} a_{t,\new}$. Also, $(a_{t,*})_i$ is uniformly distributed between $-\gamma_{i,t}$ and $\gamma_{i,t}$ for $i=1,2,\cdots,r_1-c_{2,\old}$ and $a_{t,\new}$ is uniformly distributed between $-\gamma_{\new,t}$ and $\gamma_{\new,t}$ where
\bea
\gamma_{i,t} &=&
\begin{cases} 400 & \text{if $i=1,2,\cdots,7, \forall t$,}
\\
30 &\text{if $i=8,9,\cdots,14, \forall t$.}
\\
2 &\text{if $i=15,16,\cdots,21, \forall t$.}
\\
1.331 &\text{if $i=22, \forall t$.}
\\
1 &\text{if $i=23,24,\cdots,31, \forall t$.}
\end{cases} \label{gamma0} \\
\gamma_{\new,t} &=&
\begin{cases} 1.1^{k-1} & \text{if $t_2 + (k-1) \alpha \leq t \leq t_2 +k\alpha-1, k=1,2,\cdots,7$,}
\\
1.1^{7-1} = 1.7716  &\text{if $t \geq t_2 + 7\alpha$.}
\end{cases} \label{gamma2}
\eea
Thus for the above model, $S_{\min}=2$, $\gamma_* = 400$, $\gamma_{\new} =1$, $\lambda^+ = 53333$, $\lambda^- = 0.3333$ and $f:=\frac{\lambda^+}{\lambda^-} = 1.6 \times 10^{5}$.
One way to get the clusters of $\{1,2,\cdots,r_j\}$ is as follows.
\ben
\item For $t_1  \leq t < t_2$ with $j=1$, let $\group_{1,(1)} = \{ 1,2,\cdots, 8\}$, $\group_{1,(2)} = \{9,10,\cdots,16\}$ and $\group_{1,(3)} = \{17,18,\cdots,34\}$. Thus, $\tilde{c}_{1,1} = \tilde{c}_{1,2} = 8$,  $\tilde{c}_{1,3} = 18$, $\tilde{g}_{j,1} = \tilde{g}_{j,2} = 1$, $\tilde{g}_{j,3} = 4$, $\tilde{h}_{j,1} = 0.0056$, $\tilde{h}_{j,2} = 0.0044$.
\item For $t\geq t_2$ with $j=2$, let $\group_{1,(1)} = \{ 1,2,\cdots, 7\}$, $\group_{1,(2)} = \{8,10,\cdots,14\}$ and $\group_{1,(3)} = \{17,18,\cdots,32\}$. Thus, $\tilde{c}_{1,1} = \tilde{c}_{1,2} = 7$, $\tilde{c}_{1,3} = 16$,  $\tilde{g}_{j,1} = \tilde{g}_{j,2} = 1$, $\tilde{g}_{j,3} = 4$, $\tilde{h}_{j,1} = 0.0056$, $\tilde{h}_{j,2} = 0.0044$.
\item Therefore, $\tilde{g}_{\max} = 4$, $\tilde{h}_{\max}= 0.0056 $ and $\tilde{c}_{\min} = 7$.
\een

We used $\mathcal{L}_{t_{\train}} + \mathcal{N}_{t_{\train}}$ as the training sequence to estimate $\Phat_0$. Here  $\mathcal{N}_{t_{\train}} = [N_1, N_2 ,\cdots, N_{t_{\train}}]$ is i.i.d. random noise with each $(N_t)_i$ uniformly distributed between $-10^{-3}$ and $10^{-3}$. This is done to ensure that $\Span(\Phat_0) \neq \Span(P_0)$ but only approximates it.

\subsubsection{Results}

\begin{figure}
\hspace{-2mm}
\centerline{
\subfigure[subspace error, $\SE_{(t)}$]{
\epsfig{file = 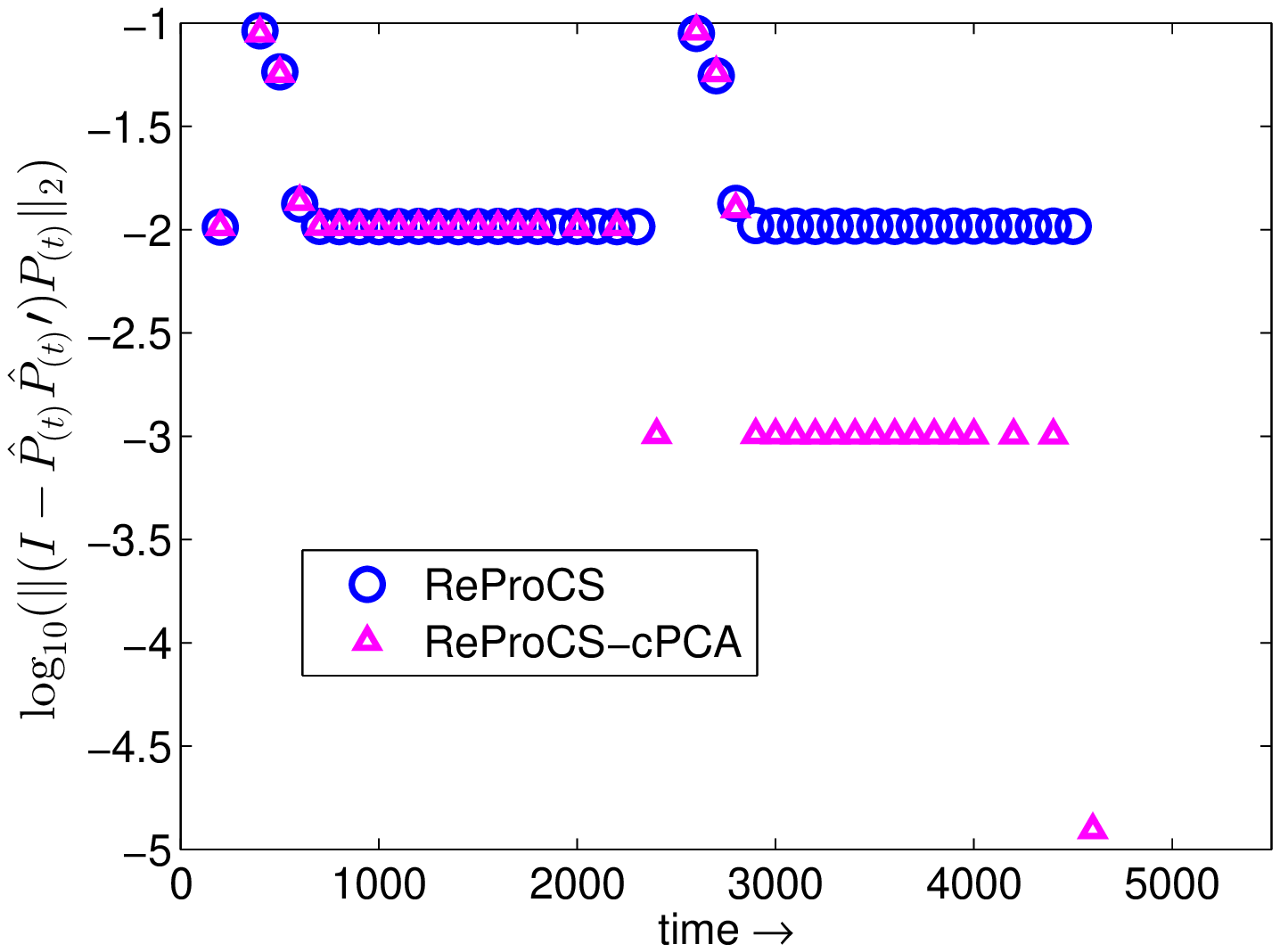, width =6cm, height=5cm}
}
\subfigure[recon error of $S_t$]{
\epsfig{file = 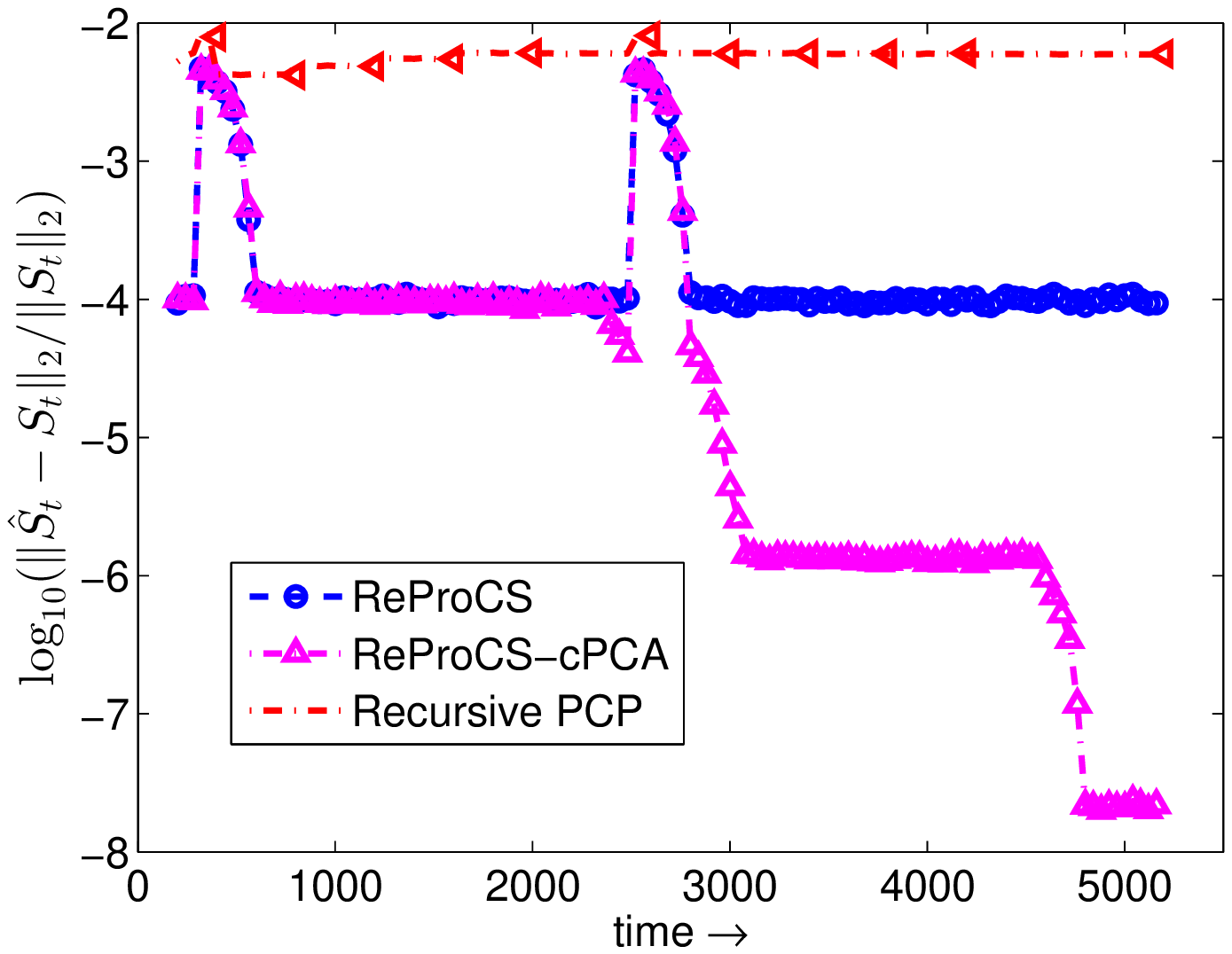, width =6cm, height=5cm}
}
\subfigure[plot of $\frac{\|{I_{T_t}}' D_{j,\new,k}\|_2}{\|D_{j,\new,k}\|_2}$]{
\epsfig{file = 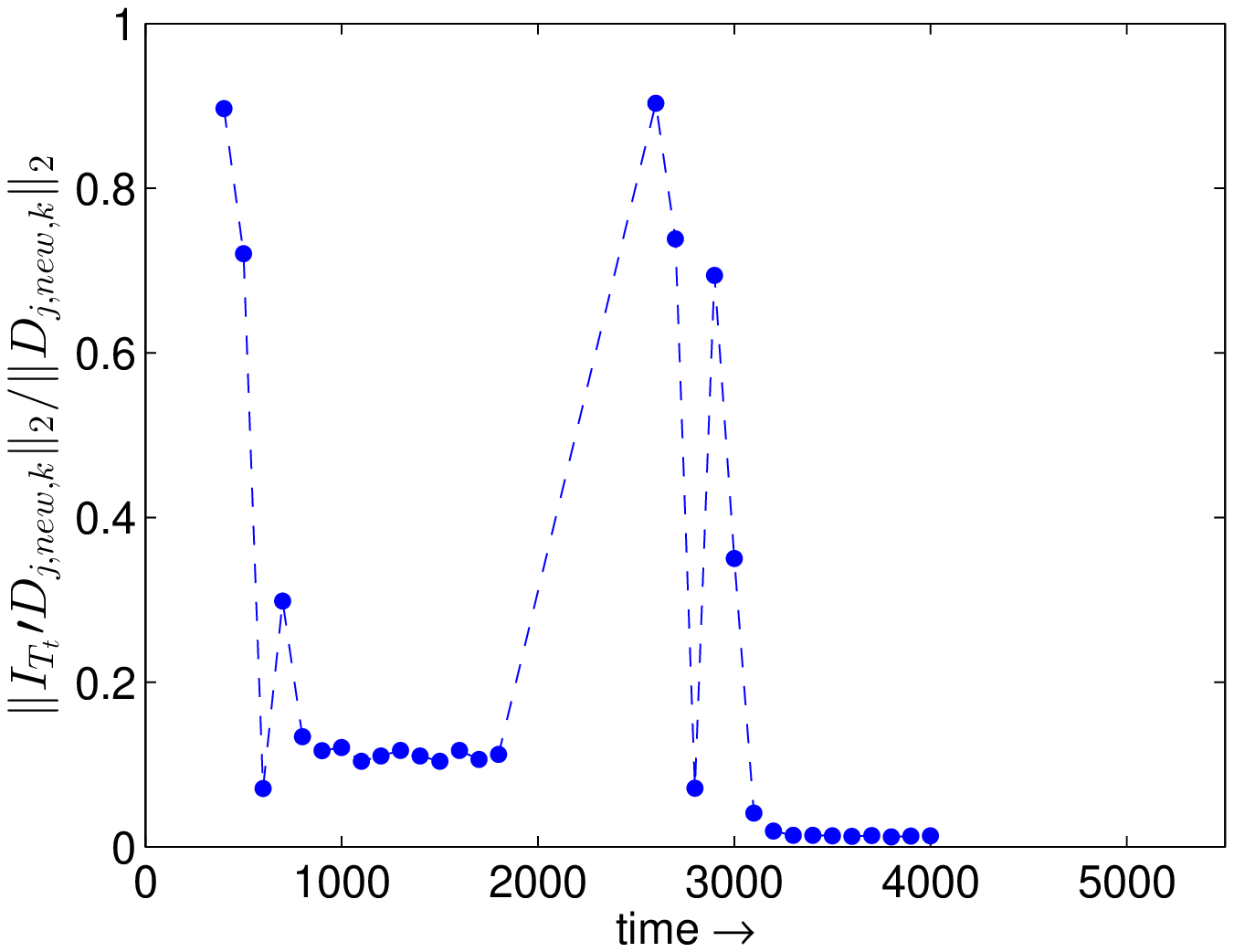, width=6cm,height=5cm}
}
}
\caption{\small{$r_0 =36$, $s = \max_t |T_t| = 20$ and $\Delta =10$. The times at which PCP is done are marked by red triangles in (b).
}\label{s20_Delta_10_del}
}
\end{figure}

\begin{figure}
\hspace{-2mm}
\centerline{
\subfigure[subspace error, $\SE_{(t)}$]{
\epsfig{file = 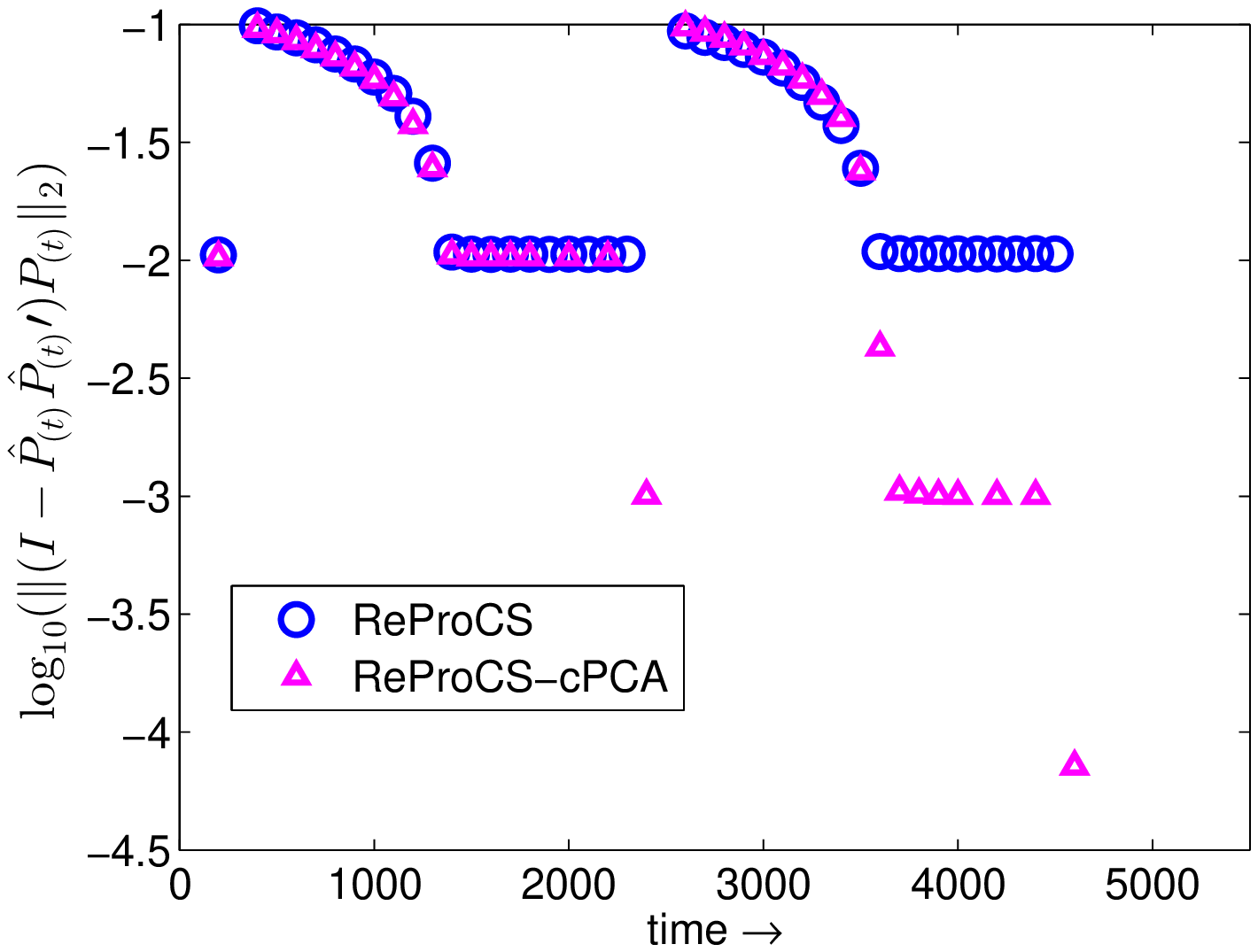, width =6cm, height=5cm}
}
\subfigure[recon error of $S_t$]{
\epsfig{file = 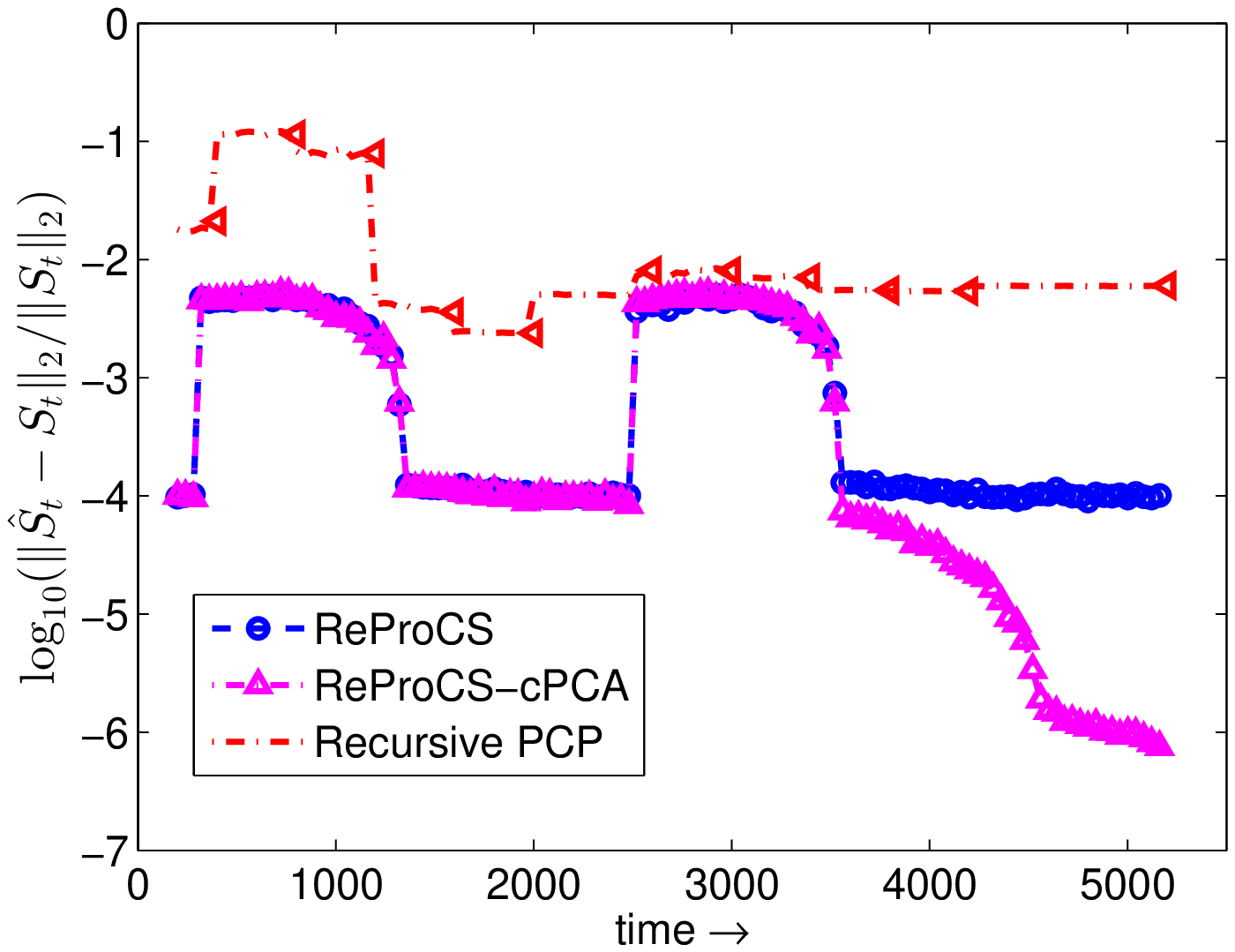, width =6cm, height=5cm}
}
\subfigure[plot of $\frac{\|{I_{T_t}}' D_{j,\new,k}\|_2}{\|D_{j,\new,k}\|_2}$]{
\epsfig{file = 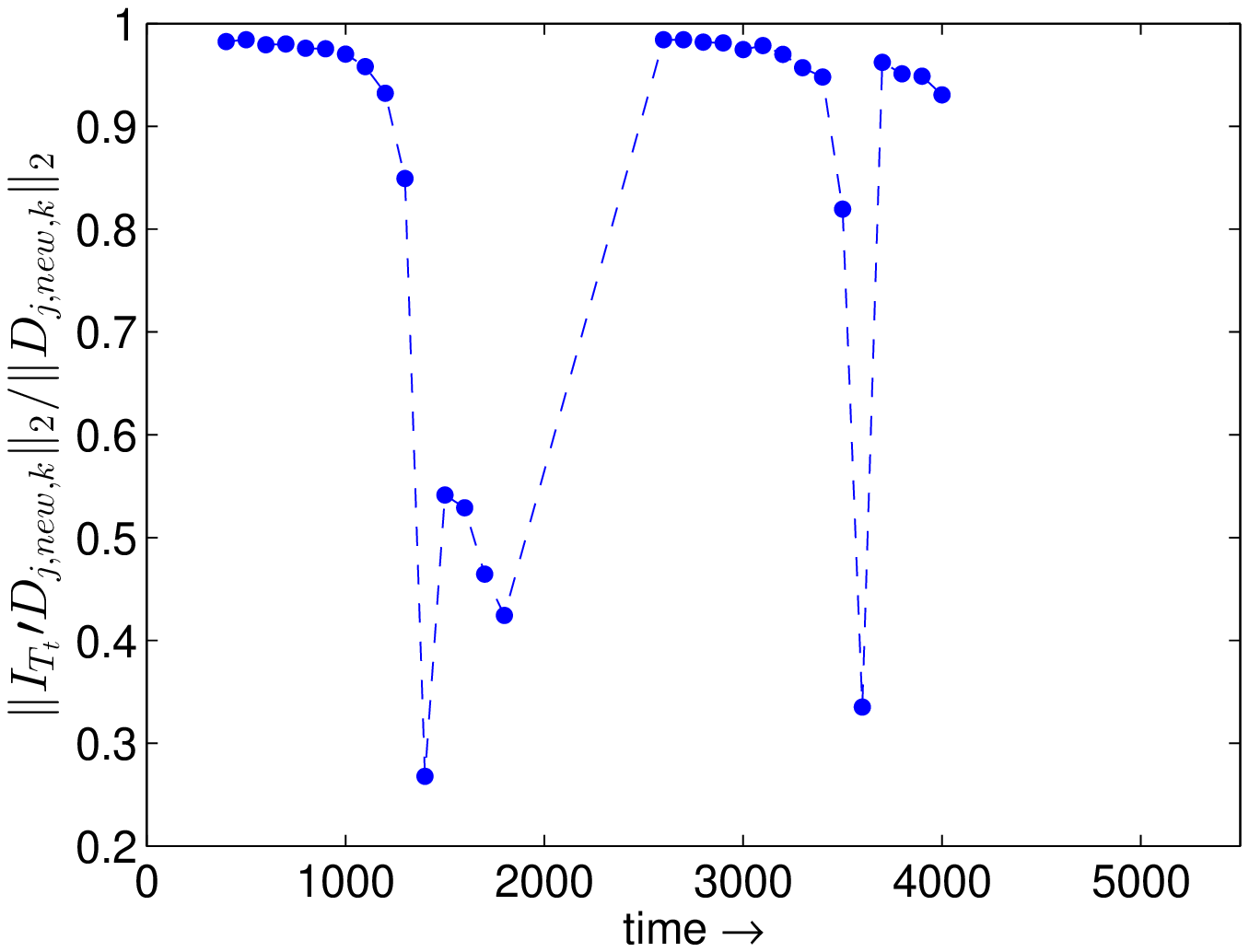, width=6cm,height=5cm}
}
}
\caption{\small{$r_0 =36$, $s = \max_t |T_t| = 20$ and $\Delta =50$.  The times at which PCP is done are marked by red triangles in (b).}\label{s20_Delta_50_del}
}
\end{figure}

For Fig. \ref{s20_Delta_10_del} and Fig. \ref{s20_Delta_50_del}, we used $s=20$. We used $\Delta = 10$ for Fig. \ref{s20_Delta_10_del} and $\Delta= 50$ for Fig. \ref{s20_Delta_50_del}.
Because of the correlated support change, the $2048 \times t$ sparse matrix $\mathcal{S}_t = [S_1 ,S_2,\cdots,S_t]$ is rank deficient in either case, e.g. for Fig. \ref{s20_Delta_10_del}, $\mathcal{S}_t$ has rank $29,39,49,259$ at $t=300,400,500,2600$; for Fig. \ref{s20_Delta_50_del}, $\mathcal{S}_t$ has rank $21,23,25,67$ at $t=300,400,500,2600$. We plot the subspace error $\SE_{(t)}$ and the normalized error for $S_t$, $\frac{\|\hat{S}_t-S_t\|_2}{\|S_t\|_2}$ averaged over 100 Monte Carlo simulations.

As can be seen from Fig. \ref{s20_Delta_10_del} and Fig. \ref{s20_Delta_50_del}, the subspace error $\SE_{(t)}$ of ReProCS and ReProCS-cPCA decreased exponentially and stabilized. Furthermore, ReProCS-cPCA outperforms over ReProCS greatly when deletion steps are done (i.e., at $t=2400$ and $4600$). The averaged normalized error for $S_t$ followed a similar trend.

We also compared against PCP \cite{rpca}. At every $t = t_j + 4 k\alpha$, we solved (\ref{pcp_prob}) with $\lambda = 1/\sqrt{\max(n,t)}$ as suggested in \cite{rpca} to recover ${\cal S}_t$ and ${\cal L}_t$. We used the estimates of $S_t$ for the last $4 \alpha$ frames as the final estimates of $\Shat_t$. So, the $\Shat_t$ for $t=t_j+1, \dots t_j + 4 \alpha$ is obtained from PCP done at $t=t_j + 4 \alpha$, the $\Shat_t$ for $t=t_j+4\alpha + 1, \dots t_j + 8 \alpha$ is obtained from PCP done at $t=t_j + 8 \alpha$ and so on. Because of the correlated support change, the error of PCP was larger in both cases.

We also plot the ratio $\frac{\|{I_{T_t}}' D_{j,\new,k}\|_2}{\|D_{j,\new,k}\|_2}$ at the projection PCA times. This serves as a proxy for $\kappa_s(D_{j,\new,k})$ (which has exponential computational complexity). As can be seen from Fig. \ref{s20_Delta_10_del} and Fig. \ref{s20_Delta_50_del}, this ratio is less than 1 and it becomes larger when $\Delta$ increases ($T_t$ becomes more correlated over $t$).  

We implemented ReProCS-cPCA using Algorithm \ref{ReProCS_del} with $\alpha = 100$, $\tilde{\alpha} = 200$ and $K=15$. The algorithm is not very sensitive to these choices. Also, we let $\xi = \xi_t$ and $\omega = \omega_t$ vary with time. Recall that $\xi_t$ is the upper bound on $\|\beta_t\|_2$. We do not know $\beta_t$. All we have is an estimate of $\beta_t$ from $t-1$, $\hat{\beta}_{t-1} = (I - \Phat_{t-1} {\Phat_{t-1}}')\Lhat_{t-1}$. We used a value a little larger than $\|\hat{\beta}_{t-1}\|_2$; we let $\xi_t = 2\|\hat{\beta}_{t-1}\|_2$. The parameter $\omega_t$ is the support estimation threshold. One reasonable way to pick this is to use a percentage energy threshold of $\Shat_{t,\cs}$\cite{modcsjournal}. For a vector $v$, define the $99\%$-energy set of $v$ as $T_{0.99}(v) := \{i: |v_i| \geq v^{0.99}\}$ where the 99\% energy threshold, $v^{0.99}$, is the largest value of $|v_i|$ so that $\|v_{T_{0.99}}\|_2^2 \geq 0.99 \|v\|_2^2$. It is computed by sorting $|v_i|$ in non-increasing order of magnitude. One keeps adding elements to $T_{0.99}$ until $\|v_{T_{0.99}}\|_2^2 \geq 0.99 \|v\|_2^2$. We used $\omega_t= 0.5(\Shat_{t,\cs})^{0.99}$.

\section{Conclusions and Future Work} \label{conc}
We studied the problem of recursive sparse recovery in the presence of large but structured noise (noise lying in a ``slowly changing" low dimensional subspace). We introduced the ReProCS with cluster-PCA (ReProCS-cPCA) algorithm that addresses some of the limitations of our earlier work on ReProCS \cite{rrpcp_perf} and of PCP \cite{rpca}. 
Under mild assumptions, we showed that, w.h.p., ReProCS-cPCA can exactly recover the support set of $S_t$ at all times; and the reconstruction errors of both $S_t$ and $L_t$ are upper bounded by a time-invariant and small value at all times.
In ongoing work, we are studying the undersampled measurements case. Open questions include (i) how to analyze a practical version of ReProCS-cPCA (which does not assume knowledge of signal model parameters), and (ii) how to study the correlated $a_t$'s case (e.g. the case where $a_t$'s satisfy a linear random walk model). The starting point for (ii) would be to try to use the matrix Azuma inequality \cite{tail_bound} instead of Hoefdding.

\appendices

\section{Proof of Lemma \ref{lem_add}}
\label{proof_lem_add}
The proof follows by using the following three lemmas.

\begin{lem}[Exponential decay of $\zeta_{k}^+$]  
\label{defn_zeta+}
Assume that all the conditions of Theorem \ref{thm2} hold. Let $\zeta_*^+ = r \zeta$.
Define the series ${\zeta_{k}}^+$ as in Definition \ref{defzetap}. Then,
\ben
\item  $\zeta_0^+=1$ and $\zeta_k^+ \le 0.6^k + 0.4c\zeta$ for all $k =1,2, \dots K$,  
\item the denominator of $\zeta_k^+$ is positive for all $k =1,2, \dots K$.
\een
\end{lem}
\begin{proof}
This lemma is the same as \cite[Lemma 37]{rrpcp_perf} but with $\zeta_*^+$ defined differently.
\end{proof}

\begin{lem}[Sparse recovery, support recovery and expression for $e_t$]\label{cslem} %
Assume that all conditions of Theorem \ref{thm2} hold.
\ben
\item If $\zeta_{*} \le \zeta_{*}^+ := r \zeta$ and $\zeta_{k-1} \le \zeta_{k-1}^+ \le 0.6^{k-1} + 0.4 c \zeta$,  then for all $t \in  \mathcal{I}_{j,k}$, for any $k=1,2,\dots K$,
\ben
\item the projection noise $\beta_t$ satisfies $\|\beta_t\|_2 \leq \zeta_{k-1}^+ \sqrt{c} \gamma_{\new,k} + \zeta_{*}^+ \sqrt{r} \gamma_* \le \sqrt{c} 0.72^{k-1} \gamma_{\new} + 1.06 \sqrt{\zeta} \le \xi$.
\item the CS error satisfies $\|\hat{S}_{t,\cs} - S_t\|_2 \le7 \xi$.
\item  $\hat{T}_t = T_t$ 
\item $e_t$ satisfies (\ref{etdef})  
 and $\|e_t\|_2 \leq \phi^+ [\kappa_s^+ \zeta_{k-1}^+ \sqrt{c} \gamma_{\new,k} + \zeta_{*}^+ \sqrt{r} \gamma_*] \le
 0.18 \cdot 0.72^{k-1} \sqrt{c}\gamma_{\new} + 1.17 \cdot 1.06 \sqrt{\zeta}$
\een
\item For all $k=1,2,\dots K$, $\mathbf{P}(\That_t = T_t \ \text{and} \ e_t \ \text{satisfies (\ref{etdef})}  \text{ for all } t \in \mathcal{{I}}_{j,k}  | X_{j,k-1,0} ) = 1$ for all $X_{j,k-1,0} \in \Gamma_{j,k-1,0}$.
\item For all $k=1,2,\dots K$, $\mathbf{P}(\That_t = T_t \ \text{and} \ e_t \ \text{satisfies (\ref{etdef})}  \text{ for all } t \in \mathcal{{I}}_{j,k}  | \Gamma_{j,k-1,0}^e)=1$.
\een
\end{lem}
\begin{proof}
The first claim is the same as \cite[Lemma 30]{rrpcp_perf} but with $\zeta_*^+$ defined differently. The proof follows in an analogous fashion. The second claim follows from the first using Remark \ref{Gamma_rem}. The third claim follows using Lemma \ref{rem_prob}.
\end{proof}

\begin{lem}[High probability bound on $\zeta_{k}$]
Assume that all the conditions of Theorem \ref{thm2} hold.  Let $\zeta_*^+ = r \zeta$.  Then, for all $k=1,2, \dots K$, $$\mathbf{P}(\zeta_{k} \le \zeta_k^+|\egam_{j,k-1,0}) \ge p_k(\alpha,\zeta)$$ where $\zeta_k^+$ is defined in Definition \ref{defzetap} and $p_k(\alpha,\zeta)$ is defined in \cite[Lemma 35]{rrpcp_perf}.
\label{lem1_simple}
\end{lem}
\begin{proof} 
Using Lemma \ref{defn_zeta+}, (i) $\zeta_0^+=1$ and $\zeta_{k-1}^+ \le 0.6^{k-1} + 0.4c\zeta$ and (ii) the denominator of  $\zeta_{k}^+$ is positive. Using this and the theorem's conditions, the above lemma follows exactly as in  \cite[Lemma 35]{rrpcp_perf}. The only difference is that $\zeta_*^+$ is defined differently. Also, $\Gamma_{j,k}:=\Gamma_{j,k,0}$. The proof proceeds by first bounding $\zeta_k$ (in a fashion similar to the bound in Lemma \ref{defnPCA}); using Lemma \ref{cslem} to get an expression for $e_t$; and finally using Corollaries \ref{hoeffding_nonzero} and \ref{hoeffding_rec} to get high probability bounds on each of the terms in the bound on $\zeta_k$.
\end{proof}

\begin{proof}[Proof of Lemma \ref{lem_add}]
Lemma \ref{lem_add} follows by combining Lemma \ref{lem1_simple} and the third claim of Lemma \ref{cslem} and using the fact that $\mathbf{P}({\Gamma}_{j,k,0}^e|{\Gamma}_{j,k-1,0}^e) = \mathbf{P}(\zeta_{k} \le \zeta_k^+, \ \That_t = T_t \ \text{and} \ e_t \ \text{satisfies (\ref{etdef})}  \text{ for all } t \in \mathcal{{I}}_{j,k}|{\Gamma}_{j,k-1,0}^e)$.
%
\end{proof}

\section{Proof of Lemma \ref{bound_R}}\label{proof_lem_bound_R}
\begin{proof} [Proof of Lemma \ref{bound_R}]
\ben
\item The first claim follows because $\|D_{\text{det},k}\|_2 = \|\Psi_{k-1} G_{\text{det},k}\|_2  = \| \Psi_{k-1} [G_1 G_2 \cdots G_{k-1}]\|_2 \leq \sum_{k_1=1}^{k-1}\|\Psi_{k-1} G_{k_1}\|_2
\leq \sum_{k_1=1}^{k-1} \|\Psi_{k_1} G_{k_1}\|_2 =  \sum_{k_1=1}^{k-1} \tilde{\zeta}_{k_1} \leq  \sum_{k_1=1}^{k-1} \tilde{c}_{k_1} \zeta \leq r\zeta$. The first inequality follows by triangle inequality. The second one follows because $\hat{G}_1,\cdots,\hat{G}_{k-1}$ are mutually orthonormal and so $\Psi_{k-1} = \prod_{k_2=1}^{k-1}(I - \hat{G}_{k_2}{\hat{G}_{k_2}}')$.

\item By the first claim, $\|(I - \hat{G}_{\text{det},k} {\hat{G}_{\text{det},k}}') G_{\text{det},k} \|_2 =\|\Psi_{k-1} G_{\text{det},k}\|_2  \leq r\zeta$. By item 2) of Lemma \ref{lemma0} with $P = G_{\text{det},k}$ and $\hat{P} = \hat{G}_{\text{det},k}$, the result $\|G_{\text{det},k} {G_{\text{det},k}}' - \hat{G}_{\text{det},k}{\hat{G}_{\text{det},k}}'\|_2 \leq 2 r\zeta$ follows.

\item Recall that $D_k \overset{QR}{=} E_k R_k$ is a QR decomposition where $E_k$ is orthonormal and $R_k$ is upper triangular. Therefore, $\sigma_i(D_k)  = \sigma_i (R_k)$. Since $\|(I - \hat{G}_{\text{det},k} {\hat{G}_{\text{det},k}}')G_{\text{det},k}\|_2 =\|\Psi_{k-1} G_{\text{det},k}\|_2 \leq r\zeta$ and $G_k' G_{\text{det},k} = 0$, by item 4) of Lemma \ref{lemma0} with $P=G_{\text{det},k}$, $\Phat=\hat{G}_{\text{det},k}$ and $Q=G_k$, we have $\sqrt{1-r^2\zeta^2} \leq \sigma_i((I - \hat{G}_{\text{det},k} {\hat{G}_{\text{det},k}}')G_k)=\sigma_i(D_k)\leq 1$.

\item Since $D_k \overset{QR}{=} E_k R_k$, so  $\|{D_{\text{undet},k}}'E_k \|_2 =\|{D_{\text{undet},k}}'D_k R_k^{-1} \|_2 = \|{G_{\text{undet},k}}'\Psi_{k-1}' \Psi_{k-1} G_k R_k^{-1} \|_2 =  \|{G_{\text{undet},k}}'\Psi_{k-1} G_k R_k^{-1} \|_2 = \|{G_{\text{undet},k}}'D_k R_k^{-1} \|_2 = \|{G_{\text{undet},k}}'E_k\|_2$.
Since $E_k = D_k R_k^{-1} = ( I -\hat{G}_{\text{det},k} {\hat{G}_{\text{det},k}}') G_k R_k^{-1}$,
\bea
\|{G_{\text{undet},k}}'E_k \|_2  &=& \| {G_{\text{undet},k}}' ( I -\hat{G}_{\text{det},k} {\hat{G}_{\text{det},k}}') G_k R_k^{-1}\|_2 \nn \\
&\leq& \| {G_{\text{undet},k}}' ( I -\hat{G}_{\text{det},k} {\hat{G}_{\text{det},k}}') G_k \|_2 (1/\sqrt{1-r^2 \zeta^2}) =  \| {G_{\text{undet},k}}' \hat{G}_{\text{det},k} {\hat{G}_{\text{det},k}}' G_k \|_2(1/\sqrt{1-r^2 \zeta^2}) \nn
\eea
By item 3) of Lemma \ref{lemma0} with $P = {G}_{\text{det},k}$, $\Phat = \hat{G}_{\text{det},k}$ and $Q= G_{\text{undet},k}$, we get $\| {G_{\text{undet},k}}' \hat{G}_{\text{det},k}\|_2 \leq r\zeta$. By item 3) of Lemma \ref{lemma0} with $\Phat = \hat{G}_{\text{det},k}$ and $Q= G_k$, we get $ \| {\hat{G}_{\text{det},k}}'G_{k} \|_2 \leq r\zeta$.  Therefore, $\|{G_{\text{undet},k}}'E_k \|_2 = \|{E_k}' G_{\text{undet},k}\|_2 \leq \frac{r^2 \zeta^2}{\sqrt{1-r^2\zeta^2}}$.
\een
\end{proof}

\bibliographystyle{IEEEtran}
\bibliography{tipnewpfmt}
\end{document}